\documentclass[11pt]{article}
\usepackage[utf8]{inputenc}
\usepackage[T1]{fontenc}
\usepackage{lmodern}
\usepackage{amsmath,amssymb,amsthm,mathtools}
\usepackage[letterpaper,margin=1in]{geometry}
\usepackage{booktabs,array,tabularx}
\usepackage{float}
\usepackage{tikz}
\usepackage{etoolbox}
\usepackage[section]{placeins}
\usepackage{enumitem}
\usepackage{microtype}
\usepackage{parskip}
\usepackage{titlesec}
\usepackage[numbers]{natbib}
\usepackage{hyperref}
\hypersetup{breaklinks=true,colorlinks=true,linkcolor=blue,citecolor=blue,urlcolor=blue,pdfborder={0 0 0},pdftitle={Asymptotic Rank of Weighted Graph Tensors via Completely Positive Bounds},pdfauthor={Huanyi Zhou}}
\titlespacing*{\section}{0pt}{1.5\baselineskip}{1\baselineskip}
\titlespacing*{\subsection}{0pt}{1.1\baselineskip}{.5\baselineskip}
\setlist[itemize]{leftmargin=2em,itemsep=.6em}
\theoremstyle{plain}
\newtheorem{theorem}{Theorem}[section]
\newtheorem{lemma}[theorem]{Lemma}
\newtheorem{proposition}[theorem]{Proposition}
\newtheorem{corollary}[theorem]{Corollary}

\theoremstyle{definition}
\newtheorem{definition}[theorem]{Definition}
\newtheorem{remark}[theorem]{Remark}

\newcommand{\C}{\mathbb C}
\newcommand{\F}{\mathbb F}
\newcommand{\WT}{\mathsf{WT}}
\newcommand{\CP}{\mathrm{CP}}
\newcommand{\SDP}{\mathrm{SDP}}
\newcommand{\COP}{\mathrm{COP}}
\newcommand{\N}{\mathbb N}
\newcommand{\Tens}{\mathsf{Tens}}
\newcommand{\Arr}{\mathsf{Arr}}
\newcommand{\Ksum}{\oplus_{\!K}}
\newcommand{\Kprod}{\otimes_{\!K}}
\newcommand{\Cprod}{\otimes_{\!C}}
\newcommand{\bigKsum}{\mathop{\bigoplus\nolimits_{\!K}}\displaylimits}
\newcommand{\bigKprod}{\mathop{\bigotimes\nolimits_{\!K}}\displaylimits}
\newcommand{\bigCprod}{\mathop{\bigotimes\nolimits_{\!C}}\displaylimits}
\newcommand{\GT}{\mathsf{GT}}
\newcommand{\Id}{\mathsf{I}}
\newcommand{\R}{\operatorname{R}}
\newcommand{\Q}{\operatorname{Q}}
\newcommand{\aR}{\widetilde{\operatorname{R}}}
\newcommand{\aQ}{\widetilde{\operatorname{Q}}}
\newcommand{\Poly}{\mathcal C}
\DeclareMathOperator{\MaxCut}{MaxCut}
\DeclareMathOperator{\MinCut}{MinCut}
\DeclareMathOperator{\supp}{supp}
\DeclareMathOperator{\rank}{rank}
\title{Asymptotic Rank of Weighted Graph Tensors\\via Completely Positive Bounds}
\author{Huanyi Zhou\thanks{Tsinghua University. \href{mailto:zhouhuan24@mails.tsinghua.edu.cn}{zhouhuan24@mails.tsinghua.edu.cn}.}}
\date{}
\makeatletter
\patchcmd{\l@section}{1.0em}{.45em}{}{\PackageError{manuscript}{Contents spacing patch failed}{}}
\makeatother

\begin{document}
	\maketitle
	\begin{abstract}
		We give a conic upper bound on the asymptotic rank of every weighted graph tensor over every field. If $Q$ is its $n\times n$ edge-dimension matrix and $W_{i,j}=\log_2 Q_{i,j}$, then
		\[
		\MaxCut(W)\le\log_2\aR_{\F}(\WT(Q))\le U_{\CP}^{\Delta}(W).
		\]
		Here $U_{\CP}^{\Delta}(W)$ maximizes $\sum_{1\le i<j\le n}W_{i,j}P_{i,j}$ over nonnegative symmetric zero-diagonal matrices $P$ such that $J-2P$ is copositive and $P_{i,j}\le P_{i,k}+P_{k,j}$ for all distinct vertices. The matrix $J$ is all-one, and copositivity means nonnegativity of the quadratic form on nonnegative vectors.
		
		\smallskip
		We begin by extending OpenAI's polynomial multiplication method, together with its subsequent arbitrary-field developments, from three tensor legs to arbitrary order. We prove
		\[
		\left\lfloor\frac{n^2}{4}\right\rfloor\le\log_2\aR_{\F}(\GT_2(K_n))\le\frac{n^2}{4}.
		\]
		The edge exponent is at most $n/(2(n-1))$ and tends to the optimal value $1/2$, improving the previous Coppersmith--Winograd bounds and answering Question~1.3.1 of Christandl, Vrana, and Zuiddam. The bounds coincide for every even $n$.
		
		\smallskip
		The key construction separates blocks sharing several tensor legs at a cost linear in the number of blocks. Its entropy inequality, together with polynomial degenerations, forces diagonal growth of a symmetric spectral profile; comparing this growth with an interpolation bound gives the complete-graph bound. Although the construction uses the symmetry of complete graphs, grouping vertices turns its conclusion into quadratic constraints for every nonnegative vertex weighting. Completely positive combinations of these constraints and path restrictions then yield the weighted theorem, connecting Strassen's asymptotic spectrum with linear conic programming.
		
		\smallskip
		For every odd $m\ge3$, the theorem gives $\log_2\aR_{\F}(\GT_2(C_m))\le m-1+1/(m+1)$. This recovers $9/4$ on the triangle and improves the earlier odd-cycle estimates for $m\ge5$, leaving a gap of $1/(m+1)$ above the maximum cut size. On weighted triangles, the program recovers the rectangular matrix multiplication curve of Uzna\'nski and Amanov, van den Berg, and Zuiddam. The CP and semidefinite programs, with or without path constraints, all have value $9/4$ on the unweighted triangle, giving this upper bound a natural conic interpretation. Further applications give explicit bounds for complete multipartite graphs, split graphs, and wheels, including new cases of maximum cut equality.
	\end{abstract}
	\newpage
	\setcounter{tocdepth}{2}
	\begingroup
	\tableofcontents
	\endgroup
	\newpage

	\section{Introduction}
	\label{sec:introduction}
	\subsection{Weighted graph tensors and matrix multiplication}
	\label{subsec:intro-background}
	A weighted graph tensor is formed from identity tensors on the edges of a graph, with one tensor leg for each vertex and a possibly different dimension for each edge. For a graph with $n\ge2$ labeled vertices, these dimensions are recorded in a symmetric $n\times n$ matrix $Q$ with positive integer entries and diagonal entries one. Over a field $\F$, the identity tensor associated with a pair $(i,j)$, where $1\le i<j\le n$, is
	\[
	\sum_{a=0}^{Q_{i,j}-1}e_a^{(Q_{i,j})}\Cprod e_a^{(Q_{i,j})},
	\]
	with one leg at each endpoint. Here $e_a^{(q)}$ is a standard basis vector of $\F^q$, and $\Cprod$ denotes \emph{tensor concatenation}, the tensor product with its legs kept separate. Adding the one-dimensional factor $e_0^{(1)}$ at every other vertex gives an $n$-leg edge tensor $E_{i,j}^{(Q_{i,j})}$, whose notation is symmetric in the two endpoints. These edge tensors together define the \emph{weighted graph tensor associated with $Q$}, denoted $\WT(Q)$. They are combined using the Kronecker product $\Kprod$, which groups all factors at each vertex $v$ into one leg of dimension $D_v$:
	\[
	\WT(Q)=\bigKprod_{1\le i<j\le n}E_{i,j}^{(Q_{i,j})},
	\qquad D_v=\prod_{j=1}^{n}Q_{v,j}\quad(1\le v\le n).
	\]
	Dimension one represents an absent edge. These tensors describe multilinear computation with nonuniform dimensions and networks of pairwise entanglement~\cite{cvz2019,vc2017}; independent edge dimensions also appear in~\cite[Section~2.3]{cgmz2021}. Section~\ref{sec:preliminaries} gives the full coordinate definition.
	
	Ordinary graph tensors are the special case in which every edge has the same dimension. For a simple undirected graph $G$ and an integer $q\ge2$, the \emph{graph tensor of $G$ with edge dimension $q$}, denoted $\GT_q(G)$, is defined by $\GT_q(G)=\WT(Q)$, where $Q_{i,j}=q$ on edges and $Q_{i,j}=1$ on nonedges. For the triangle $K_3$, its graph tensor is isomorphic to the matrix multiplication tensor:
	\[
	\GT_q(K_3)\cong\sum_{i=0}^{q-1}\sum_{j=0}^{q-1}\sum_{h=0}^{q-1}
	e_{(i,j)}\Cprod e_{(j,h)}\Cprod e_{(h,i)}.
	\]
	Here $e_{(i,j)}=e_{qi+j}^{(q^2)}$. Unequal dimensions on the three edges give rectangular matrix multiplication.
	
	Tensor rank $\R_{\F}(T)$ is the minimum number of simple tensors needed to express $T$. The savings available when many copies are computed together are measured by the asymptotic rank
	\[
	\aR_{\F}(T)=\lim_{N\to\infty}\R_{\F}(T^{\Kprod N})^{1/N}.
	\]
	The constant $\omega_{\F}$ is the matrix multiplication exponent over $\F$: the infimum exponent for the number of arithmetic operations required to multiply two square matrices. The classical connection with tensor rank originates in Strassen's theory of bilinear complexity~\cite{strassen1973}. Its exponent characterization holds over every field~\cite[Proposition~15.1]{bcs1997} and gives
	$\log_2\aR_{\F}(\GT_2(K_3))=\omega_{\F}$. This identity recasts the matrix multiplication exponent problem as an asymptotic rank problem for the triangle graph tensor.
	
	Since Strassen's algorithm~\cite{strassen1969}, improvements in matrix multiplication have drawn on approximate algorithms and interpolation~\cite{bini1979,bini1980}, the asymptotic sum inequality~\cite{schonhage1981}, and the laser method and Coppersmith--Winograd tensors~\cite{strassen1986,strassen1987,cw1990}. The latter approach led to a long sequence of refinements through the analysis of tensor powers and asymmetric components~\cite{stothers2010,daviestothers2013,williams2012,legall2014,almanwilliams2021,dwz2023,vxxz2024,alman2025,dupont2026}.
	
	OpenAI introduced a polynomial multiplication approach that first gave a bound below $2.258$~\cite{openai2258} and then led to $\omega_{\C}\le9/4$~\cite{openai2026}. It combines polynomial multiplication with Strassen's asymptotic spectrum~\cite{strassen1988}: blocks that share a tensor leg are separated, and the resulting spectral growth is compared with an interpolation budget. Subsequent work developed arbitrary-field and rectangular versions~\cite{lu2026,uznanski2026,abz2026}, while Pass~\cite{pass2026} gave a constructive approach through polynomial packings. These developments provide a new starting point for higher-order graph tensors.
	
	Our main contribution is a conic upper bound for arbitrary weighted graph tensors over every field. A higher-order polynomial construction first gives a complete-graph bound. Grouping vertices turns this estimate into quadratic constraints for all nonnegative vertex weightings; completely positive combinations and path restrictions then give the weighted bound. This passage from a symmetric tensor construction to a general weighted optimization problem is the central connection of the paper.
	
	\subsection{From three tensor legs to complete graphs}
	\label{subsec:intro-complete}
	Earlier work on graph tensors shows that extending a matrix multiplication construction to higher order can reduce its cost per edge. For a graph with at least one edge, its \emph{graph exponent} is $\Gamma_{\F}(G)=\log_2\aR_{\F}(\GT_2(G))$, and its \emph{edge exponent} is $\tau_{\F}(G)=\Gamma_{\F}(G)/|E(G)|$. Definition~\ref{def:weighted-exponent} extends these parameters to nonnegative edge weights. Christandl, Vrana, and Zuiddam~\cite[Corollary~1.2.6]{cvz2019} extended the Coppersmith--Winograd construction to obtain
	\[
	\tau_{\C}(K_n)\le\log_7(9/2)<0.772943\qquad(n\ge4).
	\]
	This was smaller than the edge exponent then available for $K_3$ through the ordinary laser method. Brand et al.~\cite[Appendix~C, Theorem~48]{brand2026} refined the four-party construction to obtain $\tau_{\C}(K_4)<0.772318$. Covering $K_n$ by its four-vertex cliques extends this bound to every $n\ge4$, by the clique-covering argument of~\cite[Theorem~1.1.25]{cvz2019}.
	
	OpenAI's matrix multiplication bound and its arbitrary-field extensions give $\tau_{\F}(K_3)\le3/4$ over every field~\cite{openai2026,abz2026}. Place a triangle tensor on each vertex triple, adding one-dimensional legs at the other vertices. Their Kronecker product places $n-2$ copies of each edge on $K_n$. Submultiplicativity of asymptotic rank therefore gives the current benchmark $\tau_{\F}(K_n)\le3/4$ for every $n\ge3$. Can a higher-order use of polynomial multiplication improve this triangle-based bound?
	
	We extend the polynomial multiplication method of~\cite{openai2026,abz2026} to arbitrary order and obtain the following bound. It improves the triangle-based benchmark for every $n\ge4$ and is asymptotically optimal per edge. The proof separates blocks on $n-2$ shared legs at linear cost and compares the growth of a symmetric spectral profile with interpolation; Subsection~\ref{subsec:intro-polynomial-method} gives the details.
	\begin{theorem}[complete graph tensors]
		\label{thm:main}
		For every field $\F$ and every integer $n\ge2$,
		\[
		\left\lfloor\frac{n^2}{4}\right\rfloor
		\le\log_2\aR_{\F}(\GT_2(K_n))\le\frac{n^2}{4}.
		\]
	\end{theorem}
	The lower bound follows from flattening across a maximum cut. Consequently, the two bounds agree for every even complete graph. The upper bound gives
	$\tau_{\F}(K_n)\le n/(2(n-1))$, which tends to $1/2$ and is below $2/3$ for every $n\ge5$. This answers Question~1.3.1 of~\cite{cvz2019} and resolves Question~1.3.2 there for every $K_{2n}$. Table~\ref{tab:complete} compares both the graph exponent and the exponent per edge with the preceding higher-order bounds.
	Concurrent work of Khoruzhii~\cite{khoruzhii2026} independently obtains the complete-graph bound; its current version also treats arbitrary fields. The \hyperref[sec:concurrent-work]{Concurrent Work} section describes the overlap and the additional weighted-graph conclusions developed here.
	
	\begin{table}[H]
		\centering
		\setlength{\tabcolsep}{3.2pt}
		\begin{tabular}{@{}crrrrrr@{}}
			\toprule
			& \multicolumn{3}{c}{Graph exponent $\Gamma_{\F}(K_n)$}
			& \multicolumn{3}{c}{Edge exponent $\tau_{\F}(K_n)$}\\
			\cmidrule(lr){2-4}\cmidrule(l){5-7}
			Graph & CVZ & BCK$^+$26 & This paper & CVZ & BCK$^+$26 & This paper\\
			\midrule
			$K_4$ & 4.637658 & 4.633908 & 4.000000 & 0.772943 & 0.772318 & 0.666667\\
			$K_5$ & 7.729429 & 7.723180 & 6.250000 & 0.772943 & 0.772318 & 0.625000\\
			$K_6$ & 11.594144 & 11.584770 & 9.000000 & 0.772943 & 0.772318 & 0.600000\\
			$K_7$ & 16.231801 & 16.218678 & 12.250000 & 0.772943 & 0.772318 & 0.583334\\
			$K_8$ & 21.642401 & 21.624904 & 16.000000 & 0.772943 & 0.772318 & 0.571429\\
			$K_{10}$ & 34.782430 & 34.754310 & 25.000000 & 0.772943 & 0.772318 & 0.555556\\
			\bottomrule
		\end{tabular}
		\caption{Comparison with earlier Coppersmith--Winograd bounds for complete graph tensors. The CVZ column uses~\cite[Corollary~1.2.6]{cvz2019}; BCK$^+$26 uses the $K_4$ bound of~\cite[Appendix~C, Theorem~48]{brand2026} propagated by clique covering. Both comparisons are over $\C$. The more recent triangle-based benchmark is $\tau_{\F}(K_n)\le3/4$, as explained above. Our bounds hold over every field and are exact for even $n$. All decimals are rounded upward to six places.}
		\label{tab:complete}
	\end{table}

	It suffices to state these results with edge dimension two. Corollary~\ref{cor:base-independence} proves that, for every integer $q\ge2$,
	\[
	\begin{aligned}
		\log_q\aR_{\F}(\GT_q(G))&=\log_2\aR_{\F}(\GT_2(G)),\\
		\log_q\aQ_{\F}(\GT_q(G))&=\log_2\aQ_{\F}(\GT_2(G)).
	\end{aligned}
	\]
	Here $\aQ_{\F}$ is the asymptotic subrank. Both identities follow by sandwiching powers of the edge dimensions.
	
	The complete-graph estimate also controls graphs with nonuniform edge dimensions. Grouping vertices converts it into quadratic inequalities for every nonnegative vertex weighting. These inequalities lead to the weighted theorem below; Subsections~\ref{subsec:intro-polynomial-method} and~\ref{subsec:intro-weighted-method} explain the two stages of the argument.
	
	\subsection{Completely positive bounds for weighted graph tensors}
	\label{subsec:intro-result}
	To express the resulting bound, put $W_{i,j}=\log_2 Q_{i,j}$ for $1\le i,j\le n$, so that $W$ is symmetric, nonnegative, and zero-diagonal. We cover these logarithmic edge weights by sums of nonnegative outer products. A real symmetric matrix is \emph{completely positive} if it can be written as
	$Z=\sum_{\nu=1}^{t}a^{(\nu)}(a^{(\nu)})^{\mathsf T}$
	with $a^{(\nu)}\in\mathbb R_{\ge0}^n$. The cone of such matrices is denoted by $\CP_n$. Define the covering value
	\[
	U_{\CP}(W)=\frac14\min\left(
	\mathbf1^{\mathsf T}Z\mathbf1\;\middle|\;
	Z\in\CP_n,\quad Z_{i,j}\ge W_{i,j}\ (1\le i<j\le n)
	\right).
	\]
	Here $\mathbf1\in\mathbb R^n$ is the all-one vector, and the subscript $\CP$ denotes the completely positive cone. Each factor supplies edge capacities $a_i^{(\nu)}a_j^{(\nu)}$ at cost
	$\tfrac14(\sum_{i=1}^{n}a_i^{(\nu)})^2$.
	Completely positive and copositive optimization are established forms of linear conic programming~\cite{berman2003,dur2010}; the tensor interpretation of this covering problem is the main connection developed here.
	
	Proposition~\ref{prop:cp-duality} gives the equivalent dual formulation, which maximizes
	$\sum_{1\le i<j\le n}W_{i,j}P_{i,j}$
	over nonnegative symmetric zero-diagonal matrices $P$ satisfying
	\[
	4\sum_{1\le i<j\le n}a_i a_jP_{i,j}
	\le\left(\sum_{i=1}^{n}a_i\right)^2
	\qquad\text{for every }a\in\mathbb R_{\ge0}^n.
	\]
	Equivalently, $J-2P$ is copositive: $a^{\mathsf T}(J-2P)a\ge0$ for every $a\ge0$, where $J$ is the all-one matrix. The strengthened value $U_{\CP}^{\Delta}(W)$ adds the path triangle inequalities
	$P_{i,j}\le P_{i,k}+P_{k,j}$ for distinct vertices. These inequalities express that an edge resource can be transmitted through an intermediate vertex. We obtain the following main theorem.
	
	\begin{theorem}[weighted graph tensors]
		\label{thm:weighted-main}
		Let $\F$ be any field and let $Q$ be a symmetric $n\times n$ positive integer matrix, where $n\ge2$, with diagonal entries equal to one. If $W$ is its logarithmic edge-weight matrix, then
		\[
		\MaxCut(W)\le\log_2\aR_{\F}(\WT(Q))\le U_{\CP}^{\Delta}(W).
		\]
	\end{theorem}
	Here $\MaxCut(W)$ is the maximum total weight crossing a cut. Theorem~\ref{thm:weighted-cp} proves this statement, and also
	$U_{\CP}^{\Delta}(W)\le U_{\CP}(W)$.
	Specializing the edge dimensions gives the unweighted consequence.
	\begin{corollary}[ordinary graph tensors]
		\label{cor:intro-graph}
		For every field $\F$ and every finite simple undirected graph $G$ with at least two vertices,
		\[
		\MaxCut(G)\le\log_2\aR_{\F}(\GT_2(G))\le U_{\CP}^{\Delta}(G),
		\]
		where the conic parameter is evaluated on the adjacency matrix of $G$.
	\end{corollary}
	
	The conic formulation also clarifies the relation to semidefinite programming. A positive semidefinite matrix admits an outer-product decomposition with arbitrary real factors; complete positivity requires those factors to be nonnegative. Replacing $\CP_n$ in the covering problem by the positive semidefinite cone $\mathbb S_+^n$ gives $U_{\SDP}(W)$, exactly the Goemans--Williamson semidefinite bound for maximum cut~\cite{gw1995}.
	
	Since the edge demands are nonnegative, every feasible semidefinite cover is entrywise nonnegative. On such matrices, complete positivity and positive semidefiniteness coincide for $n\le4$, but differ for $n\ge5$~\cite{diananda1962,maxfieldminc1962}. Hence $U_{\SDP}(W)\le U_{\CP}(W)$, with equality for $n\le4$; Section~\ref{sec:cp} proves equality of the path-strengthened values as well. Evaluating $U_{\CP}$ is NP-hard in general (Proposition~\ref{prop:cp-hardness}), whereas $U_{\SDP}$ can be approximated to prescribed additive accuracy in polynomial time for rational weights. Extending the tensor-rank bound to $U_{\SDP}$ remains open: the grouping argument represents only nonnegative factors, while a semidefinite factorization may require signed factors.
	
	The weighted theorem improves the explicit odd-cycle estimates for every odd $m\ge5$, giving $\Gamma_{\F}(C_m)\le m-1+1/(m+1)$, and proves maximum cut equality for further complete multipartite and split graphs. On weighted triangles it recovers the known rectangular matrix multiplication curve as the exact value of the conic programs. Subsection~\ref{subsec:intro-consequences} states these consequences and the wheel bounds in detail.
	
	\subsection{Technical overview: the complete-graph bound}
	\label{subsec:intro-polynomial-method}
	The following five parts outline the proof of Theorem~\ref{thm:main}: separation gives an entropy inequality, polynomial degenerations turn it into growth estimates for a symmetric spectral profile, and interpolation converts those estimates into the complete-graph bound.
	
	\textbf{Part 1: separation and the entropy inequality.}
	For $n\ge3$, we extend the polynomial resource of~\cite{openai2026,abz2026} while retaining each input polynomial as a separate tensor leg. For positive integers $a_1,\ldots,a_{n-1}$, put $D=1+\sum_{j=1}^{n-1}(a_j-1)$ and define the polynomial multiplication tensor
	\[
	\Poly_n(a_1,\ldots,a_{n-1})
	=\sum_{\substack{0\le i_j<a_j\\1\le j\le n-1}}
	e_{i_1}^{(a_1)}\Cprod\cdots\Cprod e_{i_{n-1}}^{(a_{n-1})}
	\Cprod e_{i_1+\cdots+i_{n-1}}^{(D)}.
	\]
	Lemma~\ref{lem:polynomial-rank} shows that this tensor has asymptotic rank $D$ over every field. Interpolation first takes place over a rational function field; clearing denominators then gives a formal degeneration over $\F$ itself. Coefficient extraction removes this formal parameter at subexponential cost.
	
	The main difficulty is to separate blocks on several shared legs without losing the linear resource budget. Let $e=(n-1,n)$ and let
	$A=A_0\boxplus_e\cdots\boxplus_e A_{M-1}$
	be a partial direct sum: the first $n-2$ legs use common coordinates, while the last two legs have disjoint coordinate sectors for different blocks. Our construction gives
	\[
	A^{\Ksum D_{n,M}}\rightsquigarrow
	\bigKsum_{h=0}^{M-1}\bigl(A_h\Kprod E^{(M)}_{n-1,n}\bigr),
	\qquad D_{n,M}=1+(n-1)(M-1).
	\]
	Here $\rightsquigarrow$ denotes degeneration, $\Ksum$ is a full direct sum, and $E^{(M)}_{i,j}$ is an $M$-dimensional edge tensor on legs $i,j$, with one-dimensional factors on all other legs. The auxiliary tensor $\Poly_n(M,\ldots,M)$ has input labels $g_1,\ldots,g_{n-2},u\in\{0,\ldots,M-1\}$ and output label $v=\sum_{j=1}^{n-2}g_j+u$. The labels $g_j$ serve as tentative block labels on the shared legs. For a block labeled $h$, appropriate local integer weights have total
	$\sum_{j=1}^{n-2}(g_j-h)^2$.
	In the formal multilinear expression, make the coordinate substitutions $x\mapsto\varepsilon^{w(x)}x$, where $w(x)$ is the weight assigned to the coordinate symbol $x$ and $\varepsilon$ is a formal parameter (Remark~\ref{rem:weights}). The constant coefficient retains precisely the terms in which every shared label equals $h$, while $u$ remains free. Thus the output blocks become independent on every leg and retain an edge of dimension $M$. The weights are integer exponents, so this separation works in every characteristic. Theorem~\ref{thm:separation} gives these substitutions and identifies the resulting tensor after relabeling its coordinates.
	
	The retained edge is what makes this construction useful asymptotically. Fix a spectral point $\phi\in\Phi_{n,\F}$, that is, an additive, multiplicative, normalized tensor parameter that is monotone under restriction. Such parameters are also monotone under degeneration, by Lemma~\ref{lem:degeneration-asymptotic}. Define the symmetric zero-diagonal matrix $P^{(\phi)}$ by
	\[
	P_{i,j}^{(\phi)}=P_{j,i}^{(\phi)}
	=\log_2\phi(E^{(2)}_{i,j})\quad(1\le i<j\le n),
	\qquad p=P_{n-1,n}^{(\phi)}.
	\]
	For a partial sum $A=A_1\boxplus_e\cdots\boxplus_e A_s$ of nonzero blocks and a probability vector $\theta=(\theta_1,\ldots,\theta_s)$, retain in $A^{\Kprod N}$ the sector words with frequencies $\theta$. For rational $\theta$ and suitable $N$, their number is $M_N=\exp(NH(\theta_1,\ldots,\theta_s)+o(N))$, where $H(\theta_1,\ldots,\theta_s)=-\sum_{i=1}^{s}\theta_i\ln\theta_i$ and $0\ln0=0$. Every retained block has spectral value $\prod_{i=1}^{s}\phi(A_i)^{N\theta_i}$. Separation therefore gives
	\[
	(n-1)M_N\phi(A)^N
	\ge M_N^{1+p}\prod_{i=1}^{s}\phi(A_i)^{N\theta_i}.
	\]
	Canceling $M_N$, taking $N$-th roots, and letting $N$ tend to infinity yields the entropy inequality
	\[
	\phi(A)\ge\exp\bigl(pH(\theta_1,\ldots,\theta_s)\bigr)
	\prod_{i=1}^{s}\phi(A_i)^{\theta_i}.
	\]
	Continuity in the probabilities extends it to every probability vector. This is Corollary~\ref{cor:entropy}: the linear separation cost cancels the number of blocks, while the edge left in each block contributes the entropy factor.
	
	\textbf{Part 2: a symmetric spectral profile.}
	To use the entropy inequality in every input direction, average over all permutations of the tensor legs. Let $S_n$ be the permutation group on the $n$ legs, and write $\pi T$ for the tensor obtained by moving old leg $d$ to new leg $\pi(d)$. Put
	\[
	t_\phi=\frac{1}{\binom n2}\sum_{1\le i<j\le n}P_{i,j}^{(\phi)},
	\qquad \phi_\pi(T)=\phi(\pi T)\quad(\pi\in S_n),
	\]
	and define
	\[
	P_\phi(a_1,\ldots,a_{n-1})
	=\left(\prod_{\pi\in S_n}
	\phi_\pi\bigl(\Poly_n(a_1,\ldots,a_{n-1})\bigr)
	\right)^{1/(n!t_\phi)}.
	\]
	Lemma~\ref{lem:edge-weights} gives $t_\phi>0$, so this normalization is well defined. Every pair of legs has the same average spectral weight:
	$\sum_{\pi\in S_n}P_{i,j}^{(\phi_\pi)}=n!t_\phi$.
	Consequently, multiplying entropy inequalities over all permutations and taking the power $1/(n!t_\phi)$ turns an entropy factor into $\exp(H(\theta_1,\ldots,\theta_s))$. Lemma~\ref{lem:profile-basic} also gives symmetry, the normalization $P_\phi(1,\ldots,1,b)=b$, and the interpolation upper bound
	\[
	P_\phi(a_1,\ldots,a_{n-1})
	\le\left(1+\sum_{j=1}^{n-1}(a_j-1)\right)^{1/t_\phi}.
	\]
	We next obtain lower growth estimates from two explicit tensor degenerations.
	
	\textbf{Part 3: discrete concavity.}
	The first degeneration compares neighboring input lengths. Fix positive integers $a_1,\ldots,a_{n-2}$ and an integer $b\ge2$. A change of bases in homogeneous polynomial multiplication separates a quotient block of input length $b+1$ from a kernel block of length $b-1$. Multiplication preserves the kernel, so the only interaction between these blocks can be removed by local weights. Lemma~\ref{lem:concavity-degeneration} gives
	\[
	\Poly_n(a_1,\ldots,a_{n-2},b)\Kprod E^{(2)}_{n-1,n}
	\rightsquigarrow
	\Poly_n(a_1,\ldots,a_{n-2},b+1)\boxplus_e\Poly_n(a_1,\ldots,a_{n-2},b-1).
	\]
	Apply the entropy inequality to the two blocks and average over leg permutations. If $A=P_\phi(a_1,\ldots,a_{n-2},b+1)$ and $B=P_\phi(a_1,\ldots,a_{n-2},b-1)$, the result is
	\[
	2P_\phi(a_1,\ldots,a_{n-2},b)
	\ge e^{H(\theta,1-\theta)}A^\theta B^{1-\theta}
	\qquad(0\le\theta\le1).
	\]
	The factor two on the left is the normalized contribution of the extra edge. Choosing $\theta=A/(A+B)$ makes the right side equal to $A+B$, proving
	\[
	2P_\phi(a_1,\ldots,a_{n-2},b)
	\ge P_\phi(a_1,\ldots,a_{n-2},b+1)+P_\phi(a_1,\ldots,a_{n-2},b-1).
	\]
	Thus forward differences in each input coordinate are nonincreasing. This is the discrete concavity of Lemma~\ref{lem:concavity}.
	
	\textbf{Part 4: shifted tripling.}
	We next obtain a complementary comparison over a longer interval. Put $S=\sum_{j=1}^{n-2}(a_j-1)$ and $C=\Poly_n(a_1,\ldots,a_{n-2},h)$ for a positive integer $h$. Three consecutive sectors in the last input and output give a second degeneration (Lemma~\ref{lem:tripling-degeneration}):
	\[
	\Poly_n(a_1,\ldots,a_{n-2},3h+S)
	\rightsquigarrow C\boxplus_e B_n(a_1,\ldots,a_{n-2},h)\boxplus_e C.
	\]
	Here $B_n(a_1,\ldots,a_{n-2},h)$ is the aligned middle block of Definition~\ref{def:middle-polynomial-block}. As an individual tensor, it is locally isomorphic to $\sigma C$, where $\sigma$ exchanges the last two legs. The shift $S$ accommodates the total degree in the shared inputs. Local weights remove interactions between distinct sectors while preserving all three blocks.
	
	Use the entropy inequality with probabilities $(1/3,1/3,1/3)$. For each leg permutation,
	\[
	\phi_\pi\bigl(\Poly_n(a_1,\ldots,a_{n-2},3h+S)\bigr)
	\ge 3^{P_{n-1,n}^{(\phi_\pi)}}
	\phi_\pi(C)^{2/3}\phi_\pi(\sigma C)^{1/3}.
	\]
	In the product over $\pi\in S_n$, the factors for $\sigma C$ are a reordering of those for $C$. After normalization, the entropy factor becomes three, giving
	\[
	P_\phi(a_1,\ldots,a_{n-2},3h+S)\ge3P_\phi(a_1,\ldots,a_{n-2},h).
	\]
	Symmetry gives the same inequality in every input coordinate (Lemma~\ref{lem:tripling}). This controls a long increment of the profile; concavity will turn it into a bound on a single increment.
	
	\textbf{Part 5: diagonal growth and the complete-graph bound.}
	We now combine the two inequalities to bound growth when all input lengths increase together. Write $D_a=P_\phi(a,\ldots,a)$ and increase the $n-1$ input lengths from $a$ to $a+1$ one at a time. At any one step, let $f(b)$ be the profile with only that coordinate variable, and let $S$ be the sum of the other input degrees. Shifted tripling and concavity give
	\[
	2f(a)\le f(3a+S)-f(a)
	\le(2a+S)\bigl(f(a+1)-f(a)\bigr).
	\]
	Along this path, $S\le(n-2)a$, so each step multiplies the profile by at least $1+2/(na)$. Multiplying the $n-1$ steps gives
	\[
	D_{a+1}\ge\left(1+\frac{2}{na}\right)^{n-1}D_a
	\ge\left(\frac{a+1}{a}\right)^{2(n-1)/n}D_a,
	\]
	where the second inequality is Bernoulli's inequality. Since $D_1=1$, this telescopes to the lower bound in
	\[
	a^{2(n-1)/n}\le P_\phi(a,\ldots,a)
	\le\bigl((n-1)(a-1)+1\bigr)^{1/t_\phi}.
	\]
	The right side is the interpolation budget from Part~2. Letting $a$ tend to infinity compares the two polynomial growth rates and gives $t_\phi\le n/(2(n-1))$. Finally, the complete graph tensor is the Kronecker product of its binary edge tensors, so multiplicativity gives
	\[
	\phi(\GT_2(K_n))
	=2^{\sum_{1\le i<j\le n}P_{i,j}^{(\phi)}}
	=2^{\binom n2 t_\phi}\le2^{n^2/4}.
	\]
	Taking the maximum over $\phi$ and applying Strassen's spectral duality~\cite{strassen1988,wz2026} proves the upper bound in Theorem~\ref{thm:main} for $n\ge3$. For $n=2$, $\GT_2(K_2)$ is the $2\times2$ identity matrix, whose asymptotic rank is $2$. Flattening across a maximum cut gives the lower bound stated there. Sections~\ref{sec:separation} and~\ref{sec:polynomial-method} supply the full constructions and growth argument.
	
	\subsection{Technical overview: the weighted bound}
	\label{subsec:intro-weighted-method}
	The polynomial argument uses complete-graph symmetry to control the sum of spectral edge weights. The next step uses only the resulting bound, not symmetry of the target graph: grouping vertices tests the spectral weights against every nonnegative vector.
	
	The following three parts explain how Theorem~\ref{thm:main} yields Theorem~\ref{thm:weighted-main}: grouping vertices produces quadratic spectral constraints, completely positive covers combine them into a weighted upper bound, and path restrictions strengthen that bound.
	
	\textbf{Part 1: grouping complete graphs.}
	Take nonnegative integers $a_1,\ldots,a_n$ with sum $A\ge2$, and partition the vertices of $K_A$ into groups of these sizes. Merge the tensor legs in each nonempty group and use a one-dimensional leg for an empty group. Between groups $i,j$ there are $a_i a_j$ binary edges, whose labels combine into an edge of dimension $2^{a_i a_j}$. Make coordinate substitutions that retain precisely the terms in which every internal edge has label zero. Each external labeling then occurs exactly once, so the resulting tensor is $\WT(Q_a)$, where $(Q_a)_{i,j}=2^{a_i a_j}$ for $i\ne j$ and $(Q_a)_{i,i}=1$ for $1\le i\le n$.
	
	Grouping and restriction do not increase rank, including after taking Kronecker powers. For every spectral point $\phi\in\Phi_{n,\F}$, multiplicativity on edge tensors and Theorem~\ref{thm:main} therefore give
	\[
	\begin{aligned}
		\sum_{1\le i<j\le n}a_i a_jP_{i,j}^{(\phi)}
		&=\log_2\phi(\WT(Q_a))\\
		&\le\log_2\aR_{\F}(\WT(Q_a))
		\le\log_2\aR_{\F}(\GT_2(K_A))\le\frac{A^2}{4}.
	\end{aligned}
	\]
	The inequality is immediate when $A\le1$. Scaling proves it for nonnegative rational vectors, and continuity then gives
	\[
	4\sum_{1\le i<j\le n}a_i a_jP_{i,j}^{(\phi)}
	\le\left(\sum_{i=1}^{n}a_i\right)^2
	\qquad(a\in\mathbb R_{\ge0}^n).
	\]
	Equivalently, $J-2P^{(\phi)}$ is copositive. This is Lemma~\ref{lem:quadratic-spectral-constraint}: an unweighted complete-graph bound produces every nonnegative quadratic constraint in the conic program.
	
	\textbf{Part 2: completely positive covers.}
	Apply the quadratic constraint to each factor of a feasible cover $Z=\sum_{\nu=1}^{t}a^{(\nu)}(a^{(\nu)})^{\mathsf T}$ with $a^{(\nu)}\ge0$. Since spectral edge weights are nonnegative and $Z_{i,j}\ge W_{i,j}$,
	\[
	\begin{aligned}
		\sum_{1\le i<j\le n}W_{i,j}P_{i,j}^{(\phi)}
		&\le\sum_{1\le i<j\le n}Z_{i,j}P_{i,j}^{(\phi)}\\
		&\le\frac14\sum_{\nu=1}^{t}\left(\sum_{i=1}^{n}a_i^{(\nu)}\right)^2
		=\frac14\mathbf1^{\mathsf T}Z\mathbf1.
	\end{aligned}
	\]
	The spectral formula of Lemma~\ref{lem:weighted-spectral-formula} identifies the maximum of the left side with $\log_2\aR_{\F}(\WT(Q))$. Minimizing the right side over covers therefore proves the basic bound $\log_2\aR_{\F}(\WT(Q))\le U_{\CP}(W)$. Proposition~\ref{prop:cp-duality} gives the equivalent copositive maximization over matrices $P$ introduced above. Every spectral matrix is feasible there, so the conic program relaxes the spectral maximization.
	
	\textbf{Part 3: path restrictions and the strengthened bound.}
	There is a second constraint on spectral edge weights that does not come from grouping. For distinct vertices $i,j,k$ and a positive integer $q$, two edges on the path $i,k,j$ restrict to one edge on $i,j$:
	\[
	E_{i,j}^{(q)}\preceq E_{i,k}^{(q)}\Kprod E_{k,j}^{(q)}.
	\]
	On the middle leg $k$, the two edge labels $a,b$ have coordinate $qa+b$. In the formal multilinear notation of Definition~\ref{def:multilinear-notation}, make the coordinate substitutions
	\[
	x_{qa+b}^{(k)}\longmapsto\delta_{a,b}y_0^{(k)}
	\qquad(0\le a,b<q),
	\]
	leaving every other coordinate symbol unchanged. Here $\delta_{a,b}$ is one when $a=b$ and zero otherwise, and $y_0^{(k)}$ is the sole output symbol on leg $k$. The resulting expression is
	\[
	\left(\sum_{a=0}^{q-1}x_a^{(i)}x_a^{(j)}\right)y_0^{(k)}
	\prod_{\substack{1\le v\le n\\v\notin\{i,j,k\}}}x_0^{(v)},
	\]
	which represents $E_{i,j}^{(q)}$: each common edge label occurs once, and every other leg is one-dimensional. Applying $\phi$, using multiplicativity, and taking base-two logarithms with $q=2$ gives
	\[
	P_{i,j}^{(\phi)}\le P_{i,k}^{(\phi)}+P_{k,j}^{(\phi)}.
	\]
	These are the path inequalities of Lemma~\ref{lem:spectral-triangle}, valid for every vertex triple, including pairs with zero edge weight. Thus every spectral matrix remains feasible after adding the path constraints to the copositive program. It follows that
	\[
	\log_2\aR_{\F}(\WT(Q))
	=\max_{\phi\in\Phi_{n,\F}}\left(
	\sum_{1\le i<j\le n}W_{i,j}P_{i,j}^{(\phi)}\right)
	\le U_{\CP}^{\Delta}(W)\le U_{\CP}(W).
	\]
	Finally, flattening across a cut has rank equal to the product of its crossing edge dimensions. Taking logarithms and maximizing over cuts gives the lower bound $\MaxCut(W)$ (Proposition~\ref{prop:cut}). This proves both inequalities in Theorem~\ref{thm:weighted-main}. The odd-cycle application below shows that the path constraints can strictly improve the upper bound.
	
	\subsection{Applications and exact graph exponents}
	\label{subsec:intro-consequences}
	We now evaluate the weighted bound on concrete graphs. These evaluations recover the rectangular matrix multiplication curve, sharpen the known bounds for odd cycles, and give exact graph exponents whenever a conic upper bound matches the weight of a maximum cut. To state the consequences uniformly, we write $\Gamma_{\F}(W)$ for the \emph{weighted graph exponent}, the positively homogeneous continuous extension of $\log_2\aR_{\F}(\WT(Q))$ to nonnegative logarithmic edge weights, and $\Gamma_{\F}(G)$ for its value on the adjacency matrix of $G$. When the total edge weight is positive, dividing by it gives the weighted edge exponent $\tau_{\F}(W)$, extending $\tau_{\F}(G)$. Weighted bipartite graphs give the basic exact case: their graph exponent equals their total edge weight. The examples below illustrate the further consequences, with full proofs in Section~\ref{sec:applications}.
	
	\textbf{Rectangular matrix multiplication.}
	Uzna\'nski~\cite{uznanski2026} derived a rectangular exponent bound from the polynomial method in characteristic zero; Amanov, van den Berg, and Zuiddam~\cite{abz2026} extended the same curve to every field. We recover their bound by evaluating the three-vertex CP program. Here $\omega_{\F}(a,b,c)$ is the rectangular matrix multiplication exponent for multiplying matrices of dimensions $\lceil N^a\rceil\times\lceil N^b\rceil$ and $\lceil N^b\rceil\times\lceil N^c\rceil$; Definition~\ref{def:rectangular-exponent} recalls the precise normalization. For $a\ge b\ge c\ge0$, let
	\[
	\Psi(a,b,c)=
	\begin{cases}
		a+b,&c(a+b)\le ab,\\[1mm]
		\dfrac{(ab+ac+bc)^2}{4abc},&c(a+b)>ab.
	\end{cases}
	\]
	Then $\omega_{\F}(a,b,c)\le\Psi(a,b,c)$, and all four conic values---CP and SDP, with and without path inequalities---equal $\Psi(a,b,c)$ on the corresponding weighted triangle. This gives the rectangular curve a natural characterization by completely positive programming. At $a=b=c=1$, their common value is $9/4$. Thus the constant $9/4$ in the matrix multiplication bound is also the exact value of a natural conic optimization problem.
	
	\textbf{Odd cycles.}
	The gain from the path inequalities is visible on odd cycles. For every odd integer $m\ge3$, we prove
	\[
	m-1\le\Gamma_{\F}(C_m)\le\frac{m^2}{m+1}
	=m-1+\frac1{m+1}.
	\]
	The upper bound is the exact value of both path-strengthened conic programs. It follows by averaging path inequalities around the cycle and combining them with the complete-graph quadratic inequality. For $m\ge5$, this improves the estimate obtained from the odd-cycle formula of Christandl and Zuiddam~\cite[Theorem~3.10]{cz2019}, also recorded in~\cite[Section~1.1]{cvz2019}, using the current bound $\alpha_{\C}\ge1/2$~\cite{uznanski2026,abz2026}. Here $\alpha_{\C}$ is the dual matrix multiplication exponent of Definition~\ref{def:rectangular-exponent}; Table~\ref{tab:cycles} compares these explicit estimates. At $m=3$, our bound recovers $9/4$. The remaining gap above the maximum cut tends to zero as the cycle length grows.
	
	\begin{table}[H]
		\centering
		\begin{tabular}{@{}crrrr@{}}
			\toprule
			& \multicolumn{2}{c}{Graph exponent $\Gamma_{\F}(C_m)$}
			& \multicolumn{2}{c}{Edge exponent $\tau_{\F}(C_m)$}\\
			\cmidrule(lr){2-3}\cmidrule(l){4-5}
			Graph & CVZ/CZ & This paper & CVZ/CZ & This paper\\
			\midrule
			$C_5$ & 4.444445 & 4.166667 & 0.888889 & 0.833334\\
			$C_7$ & 6.461539 & 6.125000 & 0.923077 & 0.875000\\
			$C_9$ & 8.470589 & 8.100000 & 0.941177 & 0.900000\\
			$C_{11}$ & 10.476191 & 10.083334 & 0.952381 & 0.916667\\
			$C_{13}$ & 12.480000 & 12.071429 & 0.960000 & 0.928572\\
			\bottomrule
		\end{tabular}
		\caption{Upper bounds for odd cycle tensors. The CVZ/CZ columns evaluate
			the bound of~\cite[Theorem~3.10]{cz2019}, also recorded in~\cite[Section~1.1]{cvz2019},
			using the current bound $\alpha_{\C}\ge 1/2$ from~\cite{uznanski2026,abz2026}.
			The historical comparison is over $\C$; our columns hold over every field.
			All decimals are rounded upward to six places.}
		\label{tab:cycles}
	\end{table}

	\textbf{Complete multipartite graphs.}
	Nonuniform edge dimensions also help analyze unweighted graphs: grouping the independent parts of a complete multipartite graph gives a smaller weighted tensor with the same asymptotic rank (Proposition~\ref{prop:independent-blowup}). For $G=K_{s_1,\ldots,s_k}$, write $N=\sum_{i=1}^{k}s_i$ for its order and $b=\max(s_1,\ldots,s_k)$ for its largest part size. Theorem~\ref{thm:multipartite-cp} gives
	\[
	\Gamma_{\F}(G)\le
	\begin{cases}
		b(N-b),&b\ge N/2,\\
		N^2/4,&b\le N/2.
	\end{cases}
	\]
	The first branch equals the maximum cut size. The second also gives equality whenever the part sizes can be divided into two groups of equal total size. This resolves Question~1.3.2 of~\cite{cvz2019} for all these graphs; examples include $K_{2,2,2,2}$ and $K_{2,3,4,5}$.
	
	\textbf{Split graphs.}
	A complete split graph $J_{a,b}=K_a\vee\overline K_b$ consists of a clique of size $a$, an independent set of size $b$, and all edges between them. We prove
	\[
	\Gamma_{\F}(J_{a,b})=ab\quad(b\ge a),
	\qquad
	\left\lfloor\frac{(a+b)^2}{4}\right\rfloor
	\le\Gamma_{\F}(J_{a,b})\le\frac{(a+b)^2}{4}\quad(b<a).
	\]
	In particular, maximum cut equality holds when $b\ge a$ and for every complete split graph of even order. This includes book graphs with at least two pages and supplies further exact cases of Question~1.3.2. Arbitrary split graphs inherit the upper bound from a complete split graph with the same specified partition.
	
	\textbf{Wheel graphs.}
	Let $W_m$ denote the wheel with an $m$-vertex rim and one hub. Weighting the hub in the quadratic constraint couples the odd-cycle argument to the spokes and gives
	\[
	\left\lfloor\frac{3m}{2}\right\rfloor\le\Gamma_{\F}(W_m)\le
	\begin{cases}
		25m/16,&m\text{ even},\\
		(5m+1)^2/(16(m+1)),&m\text{ odd}.
	\end{cases}
	\]
	For example, $7\le\Gamma_{\F}(W_5)\le7+1/24$. The displayed upper bounds are the exact values of the path-strengthened conic parameters.
	
	The graph exponent bounds above compare a conic upper bound with a maximum cut lower bound. For asymptotic subrank, cuts give a complete characterization through the minimum cut. If $W$ is the logarithmic edge-weight matrix of $Q$, then
	\[
	\log_2\aQ_{\F}(\WT(Q))=\MinCut(W).
	\]
	Vrana and Christandl~\cite{vc2017} already state the logarithmically weighted minimum cut formula and allow parallel edges. Appendices~\ref{sec:cuts} and~\ref{sec:subrank} give the weighted flattening bounds and an explicit tensor proof of this formula over every field, including finite fields. The proof tracks coefficient extraction and its subexponential cost, then passes from multigraphs to arbitrary edge dimensions by binary approximation.
	
	\textbf{Organization.}
	Section~\ref{sec:preliminaries} introduces tensor operations, weighted graph exponents, and linear conic programming. Sections~\ref{sec:separation} and~\ref{sec:polynomial-method} establish the higher-order polynomial construction and the complete-graph bound. Section~\ref{sec:cp} derives completely positive covers, adds path inequalities, and compares them with semidefinite bounds. Section~\ref{sec:applications} develops the applications above, with explicit certificates and graph illustrations. The appendices supply the weighted cut arguments used throughout.
	
	\textbf{AI methodology.}
	Under the author's direction, GPT-6 Pro assisted in developing the higher-order constructions and the weighted and arbitrary-field arguments in Sections~\ref{sec:separation}--\ref{sec:applications}. GPT-6 Astra Ultra assisted with the mathematical exposition and organization throughout the manuscript. The \hyperref[sec:ai-disclosure]{AI Disclosure} describes these roles and the author's responsibility for the paper in more detail.

	\section{Preliminaries}
	\label{sec:preliminaries}
	The proofs combine coordinate substitutions with spectral and conic inequalities. We first fix the tensor operations and the local maps represented by these substitutions, then recall Strassen's asymptotic spectrum~\cite[Section~3]{strassen1988} and its semiring formulation; see Wigderson and Zuiddam~\cite{wz2026} for a detailed account. We apply this framework to weighted graph tensors before introducing the linear conic programs used for their upper bounds.
	
	Fix an arbitrary field $\F$. Unless a coefficient semiring $\mathcal R$ is explicitly specified, all tensors and local maps are over $\F$. The real numbers used for spectral values, entropy, and convex optimization need not belong to $\F$. The subscript $C$ in $\Cprod$ stands for concatenation, independently of the coefficient field.
	
	Readers familiar with tensor operations may first consult the aligned partial direct sum in Definition~\ref{def:partial-direct-sum}, then proceed to the spectral framework in Subsection~\ref{subsec:asymptotic-spectrum}. The coordinate conventions can be revisited when needed in the separation and polynomial constructions.
	
	We write $\N=\{0,1,2,\ldots\}$, $[q]=\{1,\ldots,q\}$, and $[q]_0=\{0,\ldots,q-1\}$ for every positive integer $q$. All tensor coordinate indices start at $0$. We use $\log$ for the logarithm to base $2$ and $\ln$ for the natural logarithm. For a positive integer $s$, a \emph{probability vector} is a vector $\theta=(\theta_1,\ldots,\theta_s)$ with $\theta_h\geq0$ for $1\le h\le s$ and $\sum_{h=1}^{s}\theta_h=1$. Its entropy is
	\[
	H(\theta_1,\ldots,\theta_s)=-\sum_{h=1}^{s}\theta_h\ln\theta_h,
	\]
	where $0\ln 0=0$. Thus entropy factors below take the form $e^{H(\theta_1,\ldots,\theta_s)}$.
	
	\subsection{Tensor coordinates and operations}
	\label{subsec:tensor-coordinates}
	We begin with coordinate operations, and then specify which coordinate choices are identified. The distinction matters when several tensor blocks share a leg.
	
	\begin{definition}[tensor arrays]
		\label{def:arrays}
		Let $\mathcal R$ be a commutative semiring: $(\mathcal R,+_{\mathcal R},0_{\mathcal R})$ and $(\mathcal R,\cdot_{\mathcal R},1_{\mathcal R})$ are commutative monoids, multiplication distributes over addition, $0_{\mathcal R}\cdot_{\mathcal R}x=0_{\mathcal R}$ for all $x\in\mathcal R$, and $1_{\mathcal R}\ne0_{\mathcal R}$. For positive integers $k,n_1,\ldots,n_k$, a \emph{$k$-th order tensor array} of dimension $(n_1,\ldots,n_k)$ over $\mathcal R$ is a function
		\[
		T:[n_1]_0\times\cdots\times[n_k]_0\longrightarrow\mathcal R.
		\]
		We write $T_{i_1,\ldots,i_k}$ for its value at $(i_1,\ldots,i_k)$ and set $\dim(T)=(n_1,\ldots,n_k)$. The $d$-th index, for $1\le d\le k$, is called the $d$-th \emph{leg}. Let $\Arr_k(\mathcal R)$ be the disjoint union of these array spaces over all positive dimension vectors. When $\mathcal R=\F$, we call its elements \emph{$k$-th order tensor arrays over $\F$}. An array is nonzero if its support is nonempty.
		
		For $A,B\in\Arr_k(\mathcal R)$ with the same dimension $(n_1,\ldots,n_k)$ and $c\in\mathcal R$, define their \emph{entrywise sum} and scalar multiple by
		\[
		(A+B)_{i_1,\ldots,i_k}=A_{i_1,\ldots,i_k}+_{\mathcal R}B_{i_1,\ldots,i_k},
		\quad
		(cA)_{i_1,\ldots,i_k}=c\cdot_{\mathcal R}A_{i_1,\ldots,i_k}
		\]
		for $0\le i_d<n_d$ and $1\le d\le k$. The \emph{support} of $T$ is
		\[
		\supp(T)=\bigl\{(i_1,\ldots,i_k)\in[n_1]_0\times\cdots\times[n_k]_0\;\bigm|\;
		T_{i_1,\ldots,i_k}\ne0_{\mathcal R}\bigr\}.
		\]
	\end{definition}
	
	The entrywise sum is different from the direct sum defined next. We retain the subscript $K$ in both operations that preserve tensor order.
	\begin{definition}[Kronecker direct sum and product]
		\label{def:kronecker}
		Let $k$ be a positive integer, and let $A,B\in\Arr_k(\mathcal R)$ have dimensions $(n_1,\ldots,n_k)$ and $(m_1,\ldots,m_k)$, respectively. Their \emph{Kronecker direct sum} $A\Ksum B$ has dimension $(n_1+m_1,\ldots,n_k+m_k)$ and entries
		\[
		(A\Ksum B)_{u_1,\ldots,u_k}=
		\begin{cases}
			A_{u_1,\ldots,u_k},&0\le u_d<n_d\text{ for }1\le d\le k,\\
			B_{u_1-n_1,\ldots,u_k-n_k},&n_d\le u_d<n_d+m_d\text{ for }1\le d\le k,\\
			0_{\mathcal R},&\text{otherwise},
		\end{cases}
		\]
		for $0\le u_d<n_d+m_d$ and $1\le d\le k$. Their \emph{Kronecker product} $A\Kprod B$ has dimension $(n_1m_1,\ldots,n_km_k)$ and entries
		\[
		(A\Kprod B)_{i_1m_1+j_1,\ldots,i_km_k+j_k}
		=A_{i_1,\ldots,i_k}\cdot_{\mathcal R}B_{j_1,\ldots,j_k}
		\]
		for $0\le i_d<n_d$, $0\le j_d<m_d$, and $1\le d\le k$.
		For every positive integer $r$, write
		\[
		A^{\Ksum r}=\bigKsum_{h=1}^{r}A,
		\quad
		A^{\Kprod r}=\bigKprod_{h=1}^{r}A
		\]
		for the repeated direct sum and repeated Kronecker product, respectively. Repeated operations on distinct arrays use their displayed order.
	\end{definition}
	
	\begin{definition}[tensor concatenation]
		\label{def:concatenation}
		Let $p,q$ be positive integers. For arrays $A\in\Arr_p(\mathcal R)$ and $B\in\Arr_q(\mathcal R)$ with dimensions $(n_1,\ldots,n_p)$ and $(m_1,\ldots,m_q)$, define their \emph{tensor concatenation} $A\Cprod B\in\Arr_{p+q}(\mathcal R)$ by
		\[
		(A\Cprod B)_{i_1,\ldots,i_p,j_1,\ldots,j_q}
		=A_{i_1,\ldots,i_p}\cdot_{\mathcal R}B_{j_1,\ldots,j_q}
		\]
		for $0\le i_d<n_d$ and $1\le d\le p$, and for $0\le j_d<m_d$ and $1\le d\le q$.
		For every positive integer $q$ and $d\in[q]_0$, let $e_d^{(q)}$ be the first-order tensor array of dimension $(q)$ with entries $(e_d^{(q)})_i=1_{\mathcal R}$ if $i=d$ and $0_{\mathcal R}$ otherwise, for $i\in[q]_0$. For positive integers $k,q$, define the \emph{diagonal tensor array}
		\[
		\Id_k(q)=\sum_{d=0}^{q-1}\bigCprod_{h=1}^{k}e_d^{(q)}.
		\]
		Thus $\dim(\Id_k(q))=(q,\ldots,q)$ and $\Id_k(q)=\Id_k(1)^{\Ksum q}$ as arrays. The notation $\Id_k(0)$ will be defined after we introduce the tensor semiring.
	\end{definition}
	
	\begin{definition}[local maps and restriction]
		\label{def:restriction}
		Let $k$ be a positive integer. Let $T,S\in\Arr_k(\mathcal R)$ have dimensions $(n_1,\ldots,n_k)$ and $(m_1,\ldots,m_k)$, respectively. Given matrices $L_d\in\mathcal R^{m_d\times n_d}$ for $1\le d\le k$, put $L=(L_1,\ldots,L_k)$ and define
		\[
		(L\cdot T)_{j_1,\ldots,j_k}
		=\sum_{\substack{0\le i_d<n_d\\1\le d\le k}}
		T_{i_1,\ldots,i_k}\prod_{d=1}^{k}(L_d)_{j_d,i_d}
		\]
		for $0\le j_d<m_d$ and $1\le d\le k$. We say that $S$ is a \emph{restriction} of $T$, and write $S\preceq T$ or $T\succeq S$, if $S=L\cdot T$ for some such matrices. We write $S\cong T$ if their dimensions agree and the matrices can be chosen invertible. Finally, write $S\sim T$ if $S\preceq T$ and $T\preceq S$.
	\end{definition}
	
	\begin{definition}[coordinate-permutation equivalence]
		\label{def:coordinate-permutations}
		Let $k$ be a positive integer, and let $A,B\in\Arr_k(\mathcal R)$. Write $A\equiv_{\mathrm{perm}}B$ if $\dim(A)=\dim(B)=(n_1,\ldots,n_k)$ and there are permutations $\sigma_d:[n_d]_0\to[n_d]_0$, for $1\le d\le k$, such that
		\[
		B_{i_1,\ldots,i_k}=A_{\sigma_1(i_1),\ldots,\sigma_k(i_k)}
		\]
		for $0\le i_d<n_d$ and $1\le d\le k$. The permutations act independently within the legs; they do not exchange the legs. Identity permutations, inverse permutations, and composition show that $\equiv_{\mathrm{perm}}$ is an equivalence relation.
	\end{definition}
	
	\begin{definition}[tensors and the tensor semiring]
		\label{def:semiring}
		Let $k$ be a positive integer. For this definition, assume that $\mathcal R$ has no zero divisors: for all $x,y\in\mathcal R$, the equality $xy=0_{\mathcal R}$ implies $x=0_{\mathcal R}$ or $y=0_{\mathcal R}$. Let $\Arr_k^{\times}(\mathcal R)$ be the set of nonzero arrays in $\Arr_k(\mathcal R)$, and put
		\[
		\Tens_k(\mathcal R)
		=\bigl(\Arr_k^{\times}(\mathcal R)/{\equiv_{\mathrm{perm}}}\bigr)
		\sqcup\{0_k\},
		\]
		where $0_k$ is a new formal element. An element of $\Tens_k(\mathcal R)$ is a \emph{$k$-th order tensor over $\mathcal R$}; when $\mathcal R=\F$, it is a \emph{$k$-th order tensor over $\F$}. For nonzero classes, define
		\[
		[A]\Ksum[B]=[A\Ksum B],\qquad
		[A]\Kprod[B]=[A\Kprod B].
		\]
		Extend these operations by $T\Ksum0_k=0_k\Ksum T=T$ and $T\Kprod0_k=0_k\Kprod T=0_k$ for every $T\in\Tens_k(\mathcal R)$. The unit is the class of $\Id_k(1)$, and we set $\Id_k(0)=0_k$. We use the same notation for a nonzero array and its class, and write $0$ for $0_k$ when the order is clear. For every tensor $T$, set $T^{\Ksum0}=0_k$ and $T^{\Kprod0}=\Id_k(1)$.
	\end{definition}
	
	The extra assumption on $\mathcal R$ is used only to construct this semiring and holds in particular for $\F$. It ensures that a product of two nonzero arrays is nonzero. The array definitions themselves apply to every commutative coefficient semiring. Our use of coordinate-permutation classes follows Wigderson and Zuiddam~\cite[Example~2.20]{wz2026}; we explicitly adjoin the zero element, as permitted by their semiring convention~\cite[Definition~2.1]{wz2026}.
	
	We verify the semiring properties. Permutations on $A$ and $B$ combine into a block permutation on each leg of $A\Ksum B$. For $A\Kprod B$, the permutations on leg $d$ act on the paired index by $(i,j)\mapsto(\sigma_d(i),\tau_d(j))$. Thus both operations are independent of the representatives. Associativity follows by concatenating three coordinate blocks for direct sums and by identifying the three paired indices for products. Swapping the two blocks on every leg proves commutativity of $\Ksum$. The permutation $im_d+j\mapsto jn_d+i$ on leg $d$ proves commutativity of $\Kprod$.
	
	For distributivity, the coordinates on each leg of $A\Kprod(B\Ksum C)$ have the form $(i,j)$ with $j$ in either the $B$ block or the $C$ block. List the pairs from the $B$ block first and those from the $C$ block second. This gives the coordinate order of $(A\Kprod B)\Ksum(A\Kprod C)$. Terms choosing different blocks on different legs have coefficient zero on both sides, and the remaining coefficients agree by the definition of $\Kprod$. Hence distributivity also holds. The class of $\Id_k(1)$ is a multiplicative unit. Finally, all identities involving $0_k$ follow from its defining rules. Therefore $\Tens_k(\mathcal R)$ is a commutative semiring.
	
	\begin{remark}[coordinates and tensor relations]
		\label{rem:tensor-coordinates}
		For a nonzero tensor, its dimension vector is the common dimension vector of its representatives. The formal zero tensor has no fixed positive dimension vector. In a coordinate identity, it is written as the zero coefficient table of the stated ambient dimensions. This convention is used for entrywise sums and local maps, not to replace the defining rules for $\Ksum$ and $\Kprod$ involving $0_k$.
		
		Tensor concatenation $\Cprod$ also descends to nonzero classes, because coordinate permutations of its factors act on disjoint groups of legs. For $S\in\Tens_p(\mathcal R)$ and $T\in\Tens_q(\mathcal R)$, set $0_p\Cprod T=S\Cprod0_q=0_{p+q}$. Scalar multiplication descends in the same way, with scalar zero giving the formal zero tensor.
		
		The relations $\preceq$, $\cong$, and $\sim$ descend to nonzero tensors: changing input and output representatives only composes their local maps with permutation matrices. Extend them by $0_k\preceq T$ for every $T$, by $T\preceq0_k$ only for $T=0_k$, and by $0_k\cong0_k$. Thus $\sim$ still means mutual restriction and is not the equivalence used to define tensors. In particular, adding zero coordinates can change a tensor, although it gives a mutually restricting tensor.
		
		Entrywise addition and the action of a specified local matrix require chosen, aligned coordinates. For example, $(1,0)\equiv_{\mathrm{perm}}(0,1)$, but $(1,0)+(1,0)=(2,0)$ and $(1,0)+(0,1)=(1,1)$ are not coordinate-permutation equivalent over $\F$. Throughout the paper, a formula for tensor entries, an entrywise sum, or a local substitution uses the displayed coordinate representatives. The labels in $e_d^{(q)}$ refer to those coordinates. With this convention, coordinate calculations describe tensors without identifying them under arbitrary invertible local maps or under mutual restriction.
	\end{remark}
	
	\begin{remark}[tensor concatenation and Kronecker product]
		The operation $\Cprod$ concatenates the legs, whereas $\Kprod$ pairs corresponding legs and retains their number. For example, over $\F$, the tensor $e_0^{(2)}\Cprod e_0^{(2)}+e_1^{(2)}\Cprod e_1^{(2)}$ is a second-order tensor of rank $2$, whereas
		\[
		e_0^{(2)}\Kprod e_0^{(2)}+e_1^{(2)}\Kprod e_1^{(2)}=(1,0,0,1)
		\]
		is a nonzero first-order tensor of rank $1$. Tensor rank, defined in Definition~\ref{def:rank}, is invariant under invertible changes of coordinates within each fixed leg. It need not be invariant under regrouping the legs. Our notation records this distinction.
	\end{remark}
	
	The separation construction uses a partial direct sum: two legs have separate coordinates for each block, while the other legs use the same coordinates. The following definition makes these shared coordinates part of the data.
	
	\begin{definition}[partial direct sum in aligned coordinates]
		\label{def:partial-direct-sum}
		Let $k\ge2$ and $M\ge1$ be integers, and let $e=(u,v)$ with $1\le u<v\le k$. For $0\le h<M$, let $T_h$ be a $k$-th order tensor in specified coordinates of positive ambient dimensions $(n_{1,h},\ldots,n_{k,h})$. Assume that $n_{d,h}=n_d$ for every $d\notin\{u,v\}$ and every $0\le h<M$, and identify these coordinates in the displayed order. Put
		\[
		\beta_h=\sum_{t=0}^{h-1}n_{u,t},\qquad
		\gamma_h=\sum_{t=0}^{h-1}n_{v,t}\;(0\le h\le M),
		\]
		where the sums are empty when $h=0$. Define $m_u=\beta_M$, $m_v=\gamma_M$, and $m_d=n_d$ for $d\notin\{u,v\}$. For $1\le d\le k$ and $0\le h<M$, let
		\[
		\eta_{d,h}=\begin{cases}
			\beta_h,&d=u,\\
			\gamma_h,&d=v,\\
			0,&d\notin\{u,v\}.
		\end{cases}
		\]
		The \emph{partial direct sum} $A=T_0\boxplus_e\cdots\boxplus_e T_{M-1}$ is specified in ambient dimensions $(m_1,\ldots,m_k)$ by
		\[
		A_{i_1+\eta_{1,h},\ldots,i_k+\eta_{k,h}}
		=(T_h)_{i_1,\ldots,i_k}
		\]
		for every $0\le h<M$ and every $(i_1,\ldots,i_k)$ with $0\le i_d<n_{d,h}$ for $1\le d\le k$. All entries not assigned by this formula are zero. The sectors assigned to different $h$ are disjoint: the $u$-th coordinate lies in $[\beta_h,\beta_{h+1})$ and the $v$-th coordinate lies in $[\gamma_h,\gamma_{h+1})$. In particular, entries using different sectors on the two legs are zero. A zero block uses the zero coefficient table in its specified ambient dimensions. If every block is zero, $A$ is the formal zero tensor written in the stated output coordinates. For $M=2$, the notation is $T_0\boxplus_e T_1$. Iterating this binary construction in the displayed order gives the same coordinates as the formula above.
	\end{definition}
	
	This construction requires the aligned representatives, as in Remark~\ref{rem:tensor-coordinates}; independent changes of coordinates on a shared leg of different blocks need not give the same tensor. Thus $\boxplus_e$ is not an operation on tensor classes alone, and it is not one of the tensor semiring operations. In particular, a specified zero block retains its coordinate sector. Section~\ref{sec:separation} uses these partial sums to state the separation problem. The polynomial degenerations below retain the same shared coordinates throughout, so their partial sums use precisely this convention.
	
	We will often describe local maps through substitutions in formal multilinear expressions.
	\begin{definition}[formal multilinear notation]
		\label{def:multilinear-notation}
		Let $k,n_1,\ldots,n_k$ be positive integers, and let $T$ be a $k$-th order tensor in specified coordinates of dimensions $(n_1,\ldots,n_k)$. Introduce distinct commuting indeterminates $x_i^{(d)}$ for $1\le d\le k$ and $0\le i<n_d$. The \emph{formal multilinear expression} of $T$ in these coordinates is
		\[
		\sum_{\substack{0\le i_d<n_d\\1\le d\le k}}
		T_{i_1,\ldots,i_k}\prod_{d=1}^{k}x^{(d)}_{i_d}.
		\]
		For positive integers $m_1,\ldots,m_k$ and matrices $L_d\in\F^{m_d\times n_d}$ for $1\le d\le k$, put $L=(L_1,\ldots,L_k)$. The basis-vector images and the corresponding substitutions are
		\[
		\begin{aligned}
			L_d(e_i^{(n_d)})&=\sum_{j=0}^{m_d-1}(L_d)_{j,i}e_j^{(m_d)},\\
			x_i^{(d)}&\longmapsto\sum_{j=0}^{m_d-1}(L_d)_{j,i}y_j^{(d)}
		\end{aligned}
		\qquad(1\le d\le k,\ 0\le i<n_d),
		\]
		where the $y_j^{(d)}$ are distinct commuting output indeterminates. The first line determines each $L_d$ by linear extension; substituting the second line into the expression of $T$ gives the formal multilinear expression of $L\cdot T$ from Definition~\ref{def:restriction}. In particular, a one-dimensional output leg still has its basis vector $e_0^{(1)}$ and its coordinate symbol $y_0^{(d)}$.
	\end{definition}
	
	A product of symbols records one basis tensor and keeps the legs separate. We use this notation throughout the local restriction and degeneration arguments below.
	
	\begin{definition}[permutations of tensor legs]
		\label{def:leg-permutation}
		For every positive integer $k$, let $S_k$ be the set of permutations of $\{1,\ldots,k\}$. For $\pi\in S_k$ and a nonzero $k$-th order tensor $T$ of dimension $(n_1,\ldots,n_k)$, define $\pi T$ by the dimensions $(n_{\pi^{-1}(1)},\ldots,n_{\pi^{-1}(k)})$ and the coordinate formula
		\[
		(\pi T)_{j_1,\ldots,j_k}=T_{j_{\pi(1)},\ldots,j_{\pi(k)}}
		\]
		for $0\le j_d<n_{\pi^{-1}(d)}$ and $1\le d\le k$, and set $\pi0_k=0_k$.
	\end{definition}
	
	A change of representative for $T$ induces a change of representative for $\pi T$, so this action is well-defined. Thus old leg $d$ becomes new leg $\pi(d)$. This operation permutes the participants, whereas $\equiv_{\mathrm{perm}}$ only permutes coordinates within each fixed leg.
	
	\subsection{Asymptotic rank and the spectrum}
	\label{subsec:asymptotic-spectrum}
	We now use local maps to define rank and subrank, and then describe their behavior under large Kronecker powers. Strassen's spectral duality will provide the link between these tensor parameters and the numerical inequalities in our proof.
	
	\begin{definition}[rank and subrank]
		\label{def:rank}
		For a positive integer $k$ and a $k$-th order tensor $T$, its \emph{rank} $\R_{\F}(T)$ is the least nonnegative integer $r$ such that $T\preceq\Id_k(r)$. For $k\ge2$, its \emph{subrank} $\Q_{\F}(T)$ is the greatest nonnegative integer $q$ such that $\Id_k(q)\preceq T$. Equivalently, fix coordinates of dimensions $(n_1,\ldots,n_k)$ for $T$. Then $\R_{\F}(T)$ is the smallest nonnegative integer $r$ for which
		\[
		T=\sum_{h=1}^{r}\bigCprod_{d=1}^{k}v_{d,h},
		\quad v_{d,h}\in\F^{n_d}\;(1\le d\le k,\ 1\le h\le r).
		\]
		An empty sum denotes the zero coefficient table in these coordinates. A tensor of the form $\bigCprod_{d=1}^{k}v_d$, with $v_d\in\F^{n_d}$ for $1\le d\le k$, is called a \emph{simple tensor}. The rank of the zero tensor is $0$; for $k\ge2$, its subrank is also $0$.
	\end{definition}
	
	For $r\ge1$, the two definitions of rank agree by taking the vectors $v_{d,h}$ as the columns of the local maps from $\Id_k(r)$; the case $r=0$ follows from the conventions for the zero tensor. The standard basis expansion is a finite sum of simple tensors, so every tensor has finite rank. A nonzero tensor restricts to $\Id_k(1)$ by selecting one nonzero entry and rescaling it. For $k\ge2$ and nonzero $T$, we also have $\Q_{\F}(T)\le\min_{1\le d\le k}(n_d)$: grouping leg $d$ against all other legs gives a matrix of rank at most $n_d$, while the corresponding matrix of $\Id_k(q)$ has rank $q$. Local maps cannot increase this matrix rank. Such \emph{flattenings} are defined in Definition~\ref{def:flattening-rank}; Lemma~\ref{lem:flattening-properties} proves their rank properties. Hence the maximum defining subrank exists. Composing restrictions also shows that rank and subrank are unchanged under $\sim$.
	
	\begin{definition}[asymptotic rank and subrank]
		\label{def:asymptotic-rank}
		For every integer $k\ge2$ and every nonzero $k$-th order tensor $T$, define
		\[
		\aR_{\F}(T)=\lim_{N\to\infty}\R_{\F}(T^{\Kprod N})^{1/N},
		\quad
		\aQ_{\F}(T)=\lim_{N\to\infty}\Q_{\F}(T^{\Kprod N})^{1/N}.
		\]
		Set $\aR_{\F}(0)=\aQ_{\F}(0)=0$.
	\end{definition}
	
	These limits exist by Fekete's lemma~\cite[no.~98]{polyaszego1998}; see also~\cite[Theorem~2.2]{schrijver2003} and its application to asymptotic rank and subrank in~\cite[Lemma~2.10 and Corollaries~2.11--2.12]{wz2026}. Indeed, taking products of rank decompositions and of subrank restrictions gives $\R_{\F}(A\Kprod B)\le\R_{\F}(A)\R_{\F}(B)$ and $\Q_{\F}(A\Kprod B)\ge\Q_{\F}(A)\Q_{\F}(B)$ for all tensors $A,B$ of the same order $k\ge2$ over $\F$. Fix a nonzero tensor $T$. The sequence $\log\R_{\F}(T^{\Kprod N})$ is subadditive, so Fekete's lemma gives the first equality below. For the second equality, apply the same lemma to the nonnegative subadditive sequence $N\log\R_{\F}(T)-\log\Q_{\F}(T^{\Kprod N})$. We obtain
	\[
	\aR_{\F}(T)=\inf_{N\ge1}\bigl(\R_{\F}(T^{\Kprod N})^{1/N}\bigr),
	\quad
	\aQ_{\F}(T)=\sup_{N\ge1}\bigl(\Q_{\F}(T^{\Kprod N})^{1/N}\bigr).
	\]
	Both limits are finite because $1\le\Q_{\F}(T^{\Kprod N})\le\R_{\F}(T^{\Kprod N})\le\R_{\F}(T)^N$ for every positive integer $N$ and every nonzero $T$.
	
	We next recall the preordered semiring framework underlying Strassen's theory. Its spectral points will turn asymptotic tensor comparisons into inequalities between nonnegative real numbers~\cite{strassen1988,wz2026}.
	
	\begin{definition}[Strassen preorder and asymptotic restriction]
		\label{def:strassen-preorder}
		Let \(\mathcal S\) be a commutative semiring with identities \(0_{\mathcal S}\) and \(1_{\mathcal S}\). For each \(r\in\mathbb N\), let \(r_{\mathcal S}\) denote the sum of \(r\) copies of \(1_{\mathcal S}\) in \(\mathcal S\) (so \(0_{\mathcal S}\) is the empty sum and \(1_{\mathcal S}\) is the sum of one copy). A preorder \(\preceq\) on \(\mathcal S\) is called a \emph{Strassen preorder} if the following properties hold:
		\begin{enumerate}[label=(\arabic*)]
			\item For all \(r,s\in\mathbb N\), we have \(r_{\mathcal S}\preceq s_{\mathcal S}\) if and only if \(r\le s\).
			\item For all \(a,b,c\in\mathcal S\), \(a\preceq b\) implies \(a+c\preceq b+c\) and \(ac\preceq bc\).
			\item For every \(a,b\in\mathcal S\) with \(b\ne0_{\mathcal S}\), there exists a positive integer \(r\) such that \(a\preceq r_{\mathcal S}b\).
		\end{enumerate}
		For tensors $S,T$ of a fixed order $k\ge2$ over $\F$, write $S\lesssim T$ if there are positive integers $h_N$, indexed by positive integers $N$, with $\log h_N=o(N)$ such that
		\[
		S^{\Kprod N}\preceq T^{\Kprod N}\Kprod\Id_k(h_N)
		\]
		for every sufficiently large $N$. This is the \emph{asymptotic restriction preorder}.
	\end{definition}

	For $\mathcal S=\Tens_k(\F)$, the element $r_{\mathcal S}$ is $\Id_k(r)$, with $\Id_k(0)=0_k$; it is not entrywise scalar multiplication by $r$. Restriction is a Strassen preorder on this semiring for $k\ge2$. For the first property, a nontrivial flattening of $\Id_k(r)$ has matrix rank $r$, so a restriction $\Id_k(r)\preceq\Id_k(s)$ requires $r\le s$. Conversely, coordinate projection gives such a restriction whenever $r\le s$. For the second property, a restriction with nonzero target is preserved under direct sum by block-diagonal local matrices and under Kronecker product by the corresponding Kronecker products of local matrices. If the target is zero, a direct sum can be restricted to its other summand by deleting the source summand, and a product restricts to zero. Thus the same property holds in every case. For the third property, if $a$ and $b$ are nonzero classes, then $\Id_k(1)\preceq b$ and $a\preceq\Id_k(\R_{\F}(a))$. Taking $\R_{\F}(a)$ direct-sum copies of the first restriction and composing gives $a\preceq b^{\Ksum\R_{\F}(a)}$. If $a$ is the formal zero tensor, we may take $r=1$.
	
	\begin{definition}[degeneration]
		\label{def:degeneration}
		Let $k\ge2$ be an integer, and let $S,T$ be nonzero $k$-th order tensors. Fix coordinate representatives of dimensions $(m_1,\ldots,m_k)$ and $(n_1,\ldots,n_k)$, respectively. We write $S\unlhd T$, and say that $T$ \emph{degenerates to} $S$, if there are matrices $L_d(\varepsilon)\in\F[\varepsilon]^{m_d\times n_d}$ for $1\le d\le k$, integers $a,b$ with $0\le a\le b$, and $k$-th order tensors $S_j$, written in the same output coordinates for $a+1\le j\le b$, such that
		\[
		L(\varepsilon)\cdot T
		=\varepsilon^a S+\sum_{j=a+1}^{b}\varepsilon^jS_j.
		\]
		Here $L(\varepsilon)=(L_1(\varepsilon),\ldots,L_k(\varepsilon))$, and the equality is coefficientwise in the fixed output coordinates; the $S_j$ may be zero. Changing representatives only composes the local matrices with constant permutation matrices, so this relation is independent of the chosen representatives. Define $0_k\unlhd T$ for every tensor $T$, and let $S\unlhd0_k$ hold only for $S=0_k$. Equivalently, write $T\rightsquigarrow S$, with the source on the left. An empty sum is allowed. Local matrices with Laurent-polynomial entries give the same notion: multiplying each matrix by a sufficiently large power of $\varepsilon$ clears its negative exponents and increases all exponents in the identity by the same amount.
	\end{definition}
	
	\begin{remark}[degenerations by weights]
		\label{rem:weights}
		Use the formal multilinear notation of Definition~\ref{def:multilinear-notation}. After any fixed local maps, assign an integer weight $w(x)$ to each coordinate symbol $x$ and make the coordinate substitutions $x\mapsto\varepsilon^{w(x)}x$. The weight of a supported term is the sum of the weights of its coordinate symbols. If every such total is nonnegative, the resulting formal multilinear expression is a polynomial in $\varepsilon$, and its constant coefficient is exactly the sum of the terms of total weight zero. The tensor represented by this coefficient is therefore a degeneration of the original tensor. Individual weights may be negative, since Laurent-polynomial local maps are allowed by Definition~\ref{def:degeneration}. We call this procedure a \emph{weight degeneration}; it is used in Sections~\ref{sec:separation} and~\ref{sec:polynomial-inequalities} and Appendix~\ref{sec:subrank}.
	\end{remark}

	\begin{lemma}[coefficient extraction; cf. Bini~\cite{bini1980} and~\cite{abz2026}]
		\label{lem:degeneration-asymptotic}
		For every field $\F$, every integer $k\ge2$, and all $k$-th order tensors $S,T$ over $\F$, the relation $S\unlhd T$ implies $S\lesssim T$. More precisely, if the leading degree in Definition~\ref{def:degeneration} is $a$, then
		\[
		S^{\Kprod N}\preceq T^{\Kprod N}\Kprod\Id_k(h_N),
		\qquad h_N=\binom{aN+k-1}{k-1}\quad(N\ge1).
		\]
	\end{lemma}
	\begin{proof}
		The claim is immediate for $S=0$. Otherwise take polynomial local maps $L_1(\varepsilon),\ldots,L_k(\varepsilon)$ as in Definition~\ref{def:degeneration}. Let $D$ be the largest degree of an entry among these local matrices. After taking the $N$-th Kronecker power, expand the local map on leg $d$ as
		\[
		L_d(\varepsilon)^{\otimes N}=\sum_{j=0}^{DN}\varepsilon^jL_{d,N,j},
		\]
		where we set $L_{d,N,j}=0$ for every $j>DN$ and whenever the corresponding degree is absent. Here the power on the left is the ordinary Kronecker power of a matrix. The coefficient of $\varepsilon^{aN}$ in the resulting tensor is exactly $S^{\Kprod N}$, since every tensor term before taking the power has degree at least $a$. Therefore
		\[
		S^{\Kprod N}
		=\sum_{\substack{0\le j_1,\ldots,j_k\le aN\\j_1+\cdots+j_k=aN}}
		(L_{1,N,j_1},\ldots,L_{k,N,j_k})\cdot T^{\Kprod N}.
		\]
		There are $h_N$ possible degree tuples. Give each tuple its own direct-sum copy of $T^{\Kprod N}$, and concatenate its local matrices on each leg. The image is precisely the displayed entrywise sum, so it is a restriction of $h_N$ copies. Since $k$ and $a$ are fixed, $h_N$ is polynomial in $N$, and $\log h_N=o(N)$. This proves the explicit restriction and hence $S\lesssim T$ over every field, without interpolation or division.
	\end{proof}
	
	The following elementary specialization observation lets us use interpolation without assuming that the coefficient field contains sufficiently many elements.
	\begin{lemma}[rational interpolation as a degeneration]
		\label{lem:rational-specialization}
		Let $k\ge2$, let $T$ be a $k$-th order tensor over $\F$, let $z$ be an indeterminate, and let $r\ge1$. If its scalar extension to $\F(z)$ has rank at most $r$, then $T\unlhd\Id_k(r)$ over $\F$.
	\end{lemma}
	\begin{proof}
		The assertion is immediate if $T=0$. Otherwise choose local matrices $L_d(z)$ over $\F(z)$ whose image of $\Id_k(r)$ is $T$. For each leg choose a nonzero polynomial $f_d(z)\in\F[z]$ clearing every denominator in $L_d(z)$, and put $B_d(z)=f_d(z)L_d(z)$. Then
		\[
		(B_1(z),\ldots,B_k(z))\cdot\Id_k(r)=f(z)T,
		\qquad f(z)=\prod_{d=1}^{k}f_d(z)\ne0.
		\]
		Write $f(z)=cz^a+\sum_{j=a+1}^{b}c_jz^j$, where $b=\deg f$ and $c\in\F\setminus\{0\}$. Multiply $B_1$ by $c^{-1}$ and use $z$ as the degeneration parameter. The resulting polynomial local maps send $\Id_k(r)$ to
		\[
		z^aT+\sum_{j=a+1}^{b}c^{-1}c_jz^jT.
		\]
		Its leading coefficient is $T$, so Definition~\ref{def:degeneration} gives $T\unlhd\Id_k(r)$ over $\F$.
	\end{proof}
	
	Degeneration is compatible with direct sum and Kronecker product. For a direct sum, first align the two leading degrees by multiplying one local map in the lower-degree degeneration by a suitable power of the parameter, and then take block-diagonal maps. For a Kronecker product, take the Kronecker products of the local maps; their leading degrees add. It is also transitive. To see the latter explicitly, clear Laurent exponents in two successive degenerations, use parameters $z$ and $\varepsilon$, and replace $z$ by $\varepsilon^M$. Taking $M$ larger than the leading degree in the second degeneration places every error term of the first strictly above the desired leading coefficient. Thus composing the local maps again gives a degeneration over the same field.
	
	\begin{definition}[asymptotic spectrum]
		\label{def:spectrum}
		For every integer $k\ge2$, the \emph{asymptotic spectrum of $k$-th order tensors over $\F$}, denoted $\Phi_{k,\F}$, consists of the maps $\phi:\Tens_k(\F)\to\mathbb R_{\ge0}$ satisfying the following conditions for all $A,B\in\Tens_k(\F)$:
		\begin{enumerate}[label=(\arabic*),itemsep=.35em]
			\item \emph{Normalization:} $\phi(0)=0$ and $\phi(\Id_k(1))=1$.
			\item \emph{Additivity:} $\phi(A\Ksum B)=\phi(A)+\phi(B)$.
			\item \emph{Multiplicativity:} $\phi(A\Kprod B)=\phi(A)\phi(B)$.
			\item \emph{Monotonicity:} if $A\preceq B$, then $\phi(A)\le\phi(B)$.
		\end{enumerate}
		An element of $\Phi_{k,\F}$ is called a \emph{spectral point}.
	\end{definition}

	For every $\phi\in\Phi_{k,\F}$ and every positive integer $q$, additivity gives $\phi(\Id_k(q))=q$. Thus $1\le\phi(T)\le\R_{\F}(T)$ for every nonzero $T$. Spectral points are also monotone under asymptotic restriction. Indeed, $S\lesssim T$ gives $\phi(S)^N\le h_N\phi(T)^N$ for all sufficiently large $N$; taking $N$-th roots and using $h_N^{1/N}\to1$ gives $\phi(S)\le\phi(T)$. In particular, Lemma~\ref{lem:degeneration-asymptotic} shows that spectral points are monotone under degeneration. The following theorem characterizes asymptotic restriction through spectral points and expresses asymptotic rank as their maximum.
	
	\begin{theorem}[Strassen's spectral duality~\cite{strassen1988,cvz2023,wz2026}]
		\label{thm:spectral-duality}
		For every integer $k\ge2$ and every $k$-th order tensor $T$,
		\[
		\aR_{\F}(T)=\max_{\phi\in\Phi_{k,\F}}\bigl(\phi(T)\bigr).
		\]
		Moreover, for $k$-th order tensors $S,T$, we have $S\lesssim T$ if and only if $\phi(S)\le\phi(T)$ for every $\phi\in\Phi_{k,\F}$.
	\end{theorem}
	
	\subsection{Weighted graph tensors, spectral weights, and cuts}
	\label{subsec:graph-tensors}
	\label{subsec:weighted-tensors}
	For weighted graph tensors, spectral values are determined by weights on pairs of tensor legs. We first define weighted graph tensors in coordinates, allowing independent edge dimensions as in Christandl, Gesmundo, Micha\l ek, and Zuiddam~\cite[Section~2.3]{cgmz2021}. Throughout the paper, graphs have at least two vertices, labeled by $\{1,\ldots,n\}$, and an undirected edge is written $(i,j)$ with $i<j$. This notation records its endpoints without orienting the edge.
	
	\begin{definition}[weighted graph tensors and graph tensors]
		\label{def:graph-tensor}
		\label{def:weighted-tensor}
		Let $n\ge2$ and $q\ge1$ be integers. For $1\le i<j\le n$, define the $n$-th order \emph{edge tensor}
		\[
		E^{(q)}_{i,j}=\sum_{a=0}^{q-1}\bigCprod_{d=1}^{n}v_{d,a},
		\qquad
		v_{d,a}=\begin{cases}
			e_a^{(q)},&d\in\{i,j\},\\
			e_0^{(1)},&d\notin\{i,j\},
		\end{cases}
		\]
		for $1\le d\le n$ and $0\le a<q$. We also write $E^{(q)}_{j,i}=E^{(q)}_{i,j}$ when the endpoints occur in reverse order. A symmetric positive integer matrix $Q$ of order $n$ with diagonal entries one is called a \emph{dimension matrix}. Its \emph{weighted graph tensor over $\F$} is
		\[
		\WT(Q)=\bigKprod_{1\le i<j\le n}E^{(Q_{i,j})}_{i,j},
		\]
		with the factors ordered lexicographically. Its \emph{logarithmic edge-weight matrix} is the real symmetric matrix $W$ whose entries are $W_{i,j}=\log_2 Q_{i,j}$ for $1\le i,j\le n$. Thus $W$ is nonnegative and has zero diagonal, and dimension one represents an absent edge.
		
		For a coordinate description, put
		\[
		\Lambda(Q)=\prod_{1\le i<j\le n}[Q_{i,j}]_0,
		\qquad D_v=\prod_{j=1}^{n}Q_{v,j}\quad(1\le v\le n).
		\]
		Extend each labeling $\lambda\in\Lambda(Q)$ symmetrically by $\lambda_{j,i}=\lambda_{i,j}$ and $\lambda_{i,i}=0$. Its mixed-radix coordinate at vertex $v$ is
		\[
		I_v^Q(\lambda)=\sum_{j=1}^{n}\lambda_{v,j}
		\prod_{h=j+1}^{n}Q_{v,h}.
		\]
		Then
		\[
		\WT(Q)=\sum_{\lambda\in\Lambda(Q)}\bigCprod_{v=1}^{n}
		e_{I_v^Q(\lambda)}^{(D_v)}.
		\]
		An empty product equals one. In particular, a vertex incident only with dimension-one edges has a leg of dimension one.
		
		For a finite simple graph $G=(V,E)$ with $V=\{1,\ldots,n\}$, the \emph{graph tensor} with edge dimension $q$ is the special case
		\[
		\GT_q(G)=\bigKprod_{(i,j)\in E}E^{(q)}_{i,j}=\WT(Q),
		\qquad
		Q_{i,j}=\begin{cases}q,&(i,j)\in E,\\1,&(i,j)\notin E,\end{cases}
		\quad(1\le i<j\le n),
		\]
		with $Q$ completed symmetrically and given diagonal one. If $d_v=\deg_G(v)$ and $e_{v,1},\ldots,e_{v,d_v}$ are the incident edges in lexicographic order, the coordinate of an edge labeling $\lambda:E\to[q]_0$ is
		\[
		I_v(\lambda)=\sum_{t=1}^{d_v}\lambda(e_{v,t})q^{d_v-t},
		\qquad
		\GT_q(G)=\sum_{\lambda:E\to[q]_0}\bigCprod_{v=1}^{n}
		e_{I_v(\lambda)}^{(q^{d_v})}.
		\]
		The tensor of an edgeless graph is $\Id_n(1)$.
	\end{definition}
	
	\begin{lemma}[pairing edge labels; cf.~\cite{cvz2019}]
		\label{lem:graph-multiplicative}
		For dimension matrices $Q,Q'$ of the same order, their entrywise product $Q\circ Q'$ satisfies
		\begin{equation}
			\label{eq:weighted-pairing}
			\WT(Q)\Kprod\WT(Q')\cong\WT(Q\circ Q').
		\end{equation}
		If $Q_{i,j}\le Q'_{i,j}$ for every $1\le i<j\le n$, then $\WT(Q)\preceq\WT(Q')$. In particular, for positive integers $p,q,N$,
		\[
		\GT_p(G)\Kprod\GT_q(G)\cong\GT_{pq}(G),
		\qquad
		\GT_q(G)^{\Kprod N}\cong\GT_{q^N}(G).
		\]
	\end{lemma}
	\begin{proof}
		For each pair $(i,j)$, associate labels $a\in[Q_{i,j}]_0$ and $b\in[Q'_{i,j}]_0$ with $Q'_{i,j}a+b\in[Q_{i,j}Q'_{i,j}]_0$. This is a bijection by division with remainder. Pairing all incident labels edge by edge gives a permutation of the coordinates on each vertex leg, and these independent permutations send each supported term on the left of~\eqref{eq:weighted-pairing} to its corresponding term on the right. For monotonicity, retain only the first $Q_{i,j}$ labels on the pair $(i,j)$, making the same projection at both endpoints. The remaining assertions follow by taking constant dimensions on the edges and iterating the pairing identity.
	\end{proof}
	
	For the triangle, a choice of coordinates gives
	\[
	\GT_q(K_3)\cong\sum_{0\le a,b,c\le q-1}
	e_{aq+b}^{(q^2)}\Cprod e_{bq+c}^{(q^2)}\Cprod e_{cq+a}^{(q^2)}.
	\]
	Its contraction with three $q\times q$ matrices $A,B,C$ is $\sum_{0\le a,b,c\le q-1}A_{a,b}B_{b,c}C_{c,a}=\operatorname{tr}(ABC)$. Hence it is the matrix multiplication tensor. Lemma~\ref{lem:graph-multiplicative} gives $\GT_2(K_3)^{\Kprod N}\cong\GT_{2^N}(K_3)$. The classical rank characterization of the matrix multiplication exponent is valid over every field~\cite[Proposition~15.1]{bcs1997}; applying it along powers of two, with projection and padding between successive powers, gives $\omega_{\F}=\log_2\aR_{\F}(\GT_2(K_3))$.
	
	The edge dimension is a normalization parameter for both asymptotic quantities. The next observation justifies stating unweighted bounds with edge dimension two.
	\begin{corollary}[independence of the edge-dimension base]
		\label{cor:base-independence}
		For every field $\F$, every finite simple graph $G$ with at least two vertices, and every integer $q\ge2$,
		\[
		\begin{aligned}
			\log_q\aR_{\F}(\GT_q(G))&=\log_2\aR_{\F}(\GT_2(G)),\\
			\log_q\aQ_{\F}(\GT_q(G))&=\log_2\aQ_{\F}(\GT_2(G)).
		\end{aligned}
		\]
	\end{corollary}
	\begin{proof}
		Let $d_N=\lfloor N\log_2q\rfloor$. Pairing and projection of edge labels from Lemma~\ref{lem:graph-multiplicative} give
		\[
		\GT_2(G)^{\Kprod d_N}\preceq\GT_q(G)^{\Kprod N}
		\preceq\GT_2(G)^{\Kprod(d_N+1)}.
		\]
		Both asymptotic rank and asymptotic subrank are monotone under restriction and satisfy $f(T^{\Kprod m})=f(T)^m$ for every nonnegative integer $m$. The latter identity follows by taking the corresponding subsequence in their defining limits. All graph tensors are nonzero, so the logarithms of both parameters are defined. Applying either parameter $f$ to the two restrictions gives
		\[
		\frac{d_N}{N}\log_2f(\GT_2(G))
		\le\log_2f(\GT_q(G))
		\le\frac{d_N+1}{N}\log_2f(\GT_2(G)).
		\]
		Since $d_N/N\to\log_2q$, taking the limit and dividing by $\log_2q$ yields
		\[
		\log_qf(\GT_q(G))=\log_2f(\GT_2(G)).
		\]
		Taking $f=\aR_{\F}$ and $f=\aQ_{\F}$ proves the two stated identities.
	\end{proof}
	
	Each spectral point assigns a real weight to every pair of tensor legs. These weights provide the linear coordinates for the weighted graph exponent.
	
	\begin{definition}[spectral edge weights]
		\label{def:spectral-edge-weights}
		Let $n\geq2$ be an integer and let $\phi\in\Phi_{n,\F}$. Define its symmetric zero-diagonal matrix of \emph{spectral edge weights} $P^{(\phi)}$ by
		\[
		P_{i,i}^{(\phi)}=0,\qquad P_{i,j}^{(\phi)}=P_{j,i}^{(\phi)}=\log_2\phi(E^{(2)}_{i,j})\quad(1\le i<j\le n).
		\]
		Here $E^{(2)}_{i,j}$ is the $n$-th order edge tensor from Definition~\ref{def:graph-tensor}: only legs $i$ and $j$ have dimension greater than one. We call $P_{i,j}^{(\phi)}$ its \emph{spectral edge weight under $\phi$}.
	\end{definition}
	
	The next lemma extends the single-edge calculation in~\cite[Section~2]{openai2026}. The single-edge power law determines the value of a spectral point on every weighted graph tensor, because the edge tensors multiply independently.
	
	\begin{lemma}[spectral values of weighted graph tensors]
		\label{lem:edge-weights}
		Let $n\ge2$, let $\phi\in\Phi_{n,\F}$, and let $P^{(\phi)}$ be its spectral edge-weight matrix. For $1\le i<j\le n$ and every positive integer $q$,
		\[
		0\le P_{i,j}^{(\phi)}\le1,
		\qquad \phi(E^{(q)}_{i,j})=q^{P_{i,j}^{(\phi)}}.
		\]
		For every dimension matrix $Q$ with logarithmic edge-weight matrix $W$,
		\[
		\phi(\WT(Q))
		=\prod_{1\le i<j\le n}Q_{i,j}^{P_{i,j}^{(\phi)}}
		=2^{\sum_{1\le i<j\le n}W_{i,j}P_{i,j}^{(\phi)}}.
		\]
		Consequently, for every simple graph $G$ on $[n]$ and every positive integer $q$,
		\[
		\phi(\GT_q(G))=q^{\sum_{(i,j)\in E(G)}P_{i,j}^{(\phi)}}.
		\]
		Moreover, $\sum_{j=2}^{n}P_{1,j}^{(\phi)}\ge1$.
	\end{lemma}
	\begin{proof}
		Fix integers $1\leq i<j\leq n$, and set $f(q)=\phi(E^{(q)}_{i,j})$ for every positive integer $q$. If $q\leq r$, projecting the two nontrivial legs onto their first $q$ coordinates restricts $E^{(r)}_{i,j}$ to $E^{(q)}_{i,j}$. Thus $f$ is nondecreasing. For all positive integers $q,r$, the local pairing of the edge labels gives $E^{(q)}_{i,j}\Kprod E^{(r)}_{i,j}\cong E^{(qr)}_{i,j}$, so $f(qr)=f(q)f(r)$. The tensor $E^{(q)}_{i,j}$ is nonzero and has rank $q$: its defining sum has $q$ terms, and its flattening across leg $i$ has rank $q$. Hence $1\leq f(q)\leq q$, which in particular proves $0\leq P_{i,j}^{(\phi)}=\log_2f(2)\leq1$.
		
		Let $q\geq2$ be an integer. For every positive integer $N$, put $\ell_N=\lfloor N\log_2q\rfloor$. Since $2^{\ell_N}\leq q^N<2^{\ell_N+1}$, monotonicity and multiplicativity give
		\[
		f(2)^{\ell_N}\leq f(q)^N\leq f(2)^{\ell_N+1}.
		\]
		Taking logarithms and dividing by $N$ yields
		\[
		\frac{\ell_N}{N}P_{i,j}^{(\phi)}
		\leq\log_2f(q)
		\leq\frac{\ell_N+1}{N}P_{i,j}^{(\phi)}.
		\]
		Both outer expressions tend to $P_{i,j}^{(\phi)}\log_2q$. Therefore $f(q)=q^{P_{i,j}^{(\phi)}}$. For $q=1$, the formula follows from $E^{(1)}_{i,j}=\Id_n(1)$.
		
		Consider the star tensor $T=\bigKprod_{j=2}^{n}E^{(2)}_{1,j}$. Use the formal multilinear notation of Definition~\ref{def:multilinear-notation}. On the first leg, the edge labels $b_2,\ldots,b_n\in\{0,1\}$ have coordinate $r=\sum_{j=2}^{n}2^{n-j}b_j$. Make the coordinate substitutions
		\[
		x_r^{(1)}\longmapsto
		\begin{cases}
			y_0^{(1)},&r=0,\\
			y_1^{(1)},&r=2^{n-1}-1,\\
			0,&\text{otherwise},
		\end{cases}
		\qquad(0\le r<2^{n-1}),
		\]
		leaving the coordinate symbols on the other legs unchanged. Here $y_0^{(1)},y_1^{(1)}$ are the new output symbols on the first leg. The surviving terms are precisely the two constant edge labelings, each with coefficient one. The resulting formal multilinear expression is
		\[
		\sum_{b=0}^{1}y_b^{(1)}\prod_{j=2}^{n}x_b^{(j)},
		\]
		which represents $\Id_n(2)$. Thus $\Id_n(2)\preceq T$. Restriction monotonicity and multiplicativity now imply
		\[
		2=\phi(\Id_n(2))\leq\phi(T)
		=\prod_{j=2}^n\phi(E^{(2)}_{1,j})
		=2^{\sum_{j=2}^n P_{1,j}^{(\phi)}}.
		\]
		Taking base-two logarithms gives the inequality
		\[
		\sum_{j=2}^{n}P_{1,j}^{(\phi)}\ge1
		\]
		stated at the end of the lemma.
		
		It remains to compute the spectral values of weighted graph tensors. By Definition~\ref{def:weighted-tensor}, $\WT(Q)$ is the Kronecker product of its edge tensors. Multiplicativity and the single-edge formula proved above therefore give
		\[
		\begin{aligned}
			\phi(\WT(Q))
			&=\prod_{1\le i<j\le n}\phi(E^{(Q_{i,j})}_{i,j})\\
			&=\prod_{1\le i<j\le n}Q_{i,j}^{P_{i,j}^{(\phi)}}
			=2^{\sum_{1\le i<j\le n}W_{i,j}P_{i,j}^{(\phi)}},
		\end{aligned}
		\]
		where the last equality uses $W_{i,j}=\log_2 Q_{i,j}$. For $\GT_q(G)$, the dimensions are $q$ on edges and $1$ on nonedges, so this identity specializes to
		\[
		\phi(\GT_q(G))
		=\prod_{(i,j)\in E(G)}q^{P_{i,j}^{(\phi)}}
		=q^{\sum_{(i,j)\in E(G)}P_{i,j}^{(\phi)}}.
		\]
		This proves both spectral value formulas, together with the edge-weight bounds already established.
	\end{proof}

	\begin{definition}[weighted graph and edge exponents]
		\label{def:weighted-exponent}
		\label{def:edge-exponent}
		For a nonnegative symmetric real matrix $W$ of order $n\ge2$ with zero diagonal, its \emph{support graph} has vertex set $\{1,\ldots,n\}$ and edges $(i,j)$ with $1\le i<j\le n$ and $W_{i,j}>0$. Define its \emph{weighted graph exponent} by
		\[
		\Gamma_{\F}(W)=\max_{\phi\in\Phi_{n,\F}}
		\left(\sum_{1\le i<j\le n}W_{i,j}P_{i,j}^{(\phi)}\right).
		\]
		When $\sum_{1\le i<j\le n}W_{i,j}>0$, its \emph{weighted edge exponent} is
		\[
		\tau_{\F}(W)=\frac{\Gamma_{\F}(W)}{\sum_{1\le i<j\le n}W_{i,j}}.
		\]
		For a simple graph $G$ with adjacency matrix $A_G$, put $\Gamma_{\F}(G)=\Gamma_{\F}(A_G)$ and, when $E(G)\ne\varnothing$, $\tau_{\F}(G)=\Gamma_{\F}(G)/|E(G)|$. These are its \emph{graph exponent} and \emph{edge exponent}, respectively.
	\end{definition}
	
	The maximum defining $\Gamma_{\F}$ is attained. The spectrum is a closed subset of the compact product $\prod_{T\in\Tens_n(\F)}[0,\R_{\F}(T)]$: normalization, additivity, multiplicativity, and monotonicity are closed conditions. Its projection onto the finitely many spectral edge weights is therefore compact~\cite{strassen1988,wz2026}. After recalling weighted cuts, we identify this definition with logarithmic asymptotic rank whenever the weights arise from tensor dimensions.
	
	\begin{definition}[weighted cuts]
		\label{def:cuts}
		For a symmetric nonnegative real matrix $W$ of order $n\ge2$ with zero diagonal and a nonempty proper subset $S\subset\{1,\ldots,n\}$, define
		\[
		\operatorname{cut}_W(S)=
		\sum_{\substack{1\le i<j\le n\\|\{i,j\}\cap S|=1}}W_{i,j}.
		\]
		The \emph{maximum cut weight} and \emph{minimum cut weight} are
		\[
		\MaxCut(W)=\max_{\varnothing\ne S\subsetneq\{1,\ldots,n\}}
		\bigl(\operatorname{cut}_W(S)\bigr),\qquad
		\MinCut(W)=\min_{\varnothing\ne S\subsetneq\{1,\ldots,n\}}
		\bigl(\operatorname{cut}_W(S)\bigr).
		\]
		For a simple graph $G=(V,E)$, write
		$\delta_G(S)=\{(i,j)\in E:|\{i,j\}\cap S|=1\}$.
		Then $\MaxCut(G)=\MaxCut(A_G)$ and $\MinCut(G)=\MinCut(A_G)$ are the largest and smallest values of $|\delta_G(S)|$ over nonempty proper $S\subset V$.
	\end{definition}
	
	\begin{lemma}[weighted spectral formula]
		\label{lem:weighted-spectral-formula}
		For every dimension matrix $Q$ with logarithmic edge-weight matrix $W$,
		\begin{equation}
			\label{eq:weighted-spectral-formula}
			\Gamma_{\F}(W)=\log_2\aR_{\F}(\WT(Q)).
		\end{equation}
		In particular, $\Gamma_{\F}(G)=\log_2\aR_{\F}(\GT_2(G))$. On nonnegative weight matrices, $\Gamma_{\F}$ is monotone, convex, positively homogeneous, and subadditive. Moreover,
		\begin{equation}
			\label{eq:weighted-lipschitz}
			|\Gamma_{\F}(W)-\Gamma_{\F}(W')|
			\le\sum_{1\le i<j\le n}|W_{i,j}-W'_{i,j}|.
		\end{equation}
	\end{lemma}
	
	\begin{proof}
		The weighted identity in Lemma~\ref{lem:edge-weights} says that $\log_2\phi(\WT(Q))$ is the linear functional $\sum_{1\le i<j\le n}W_{i,j}P_{i,j}^{(\phi)}$. Maximizing over $\phi\in\Phi_{n,\F}$ and applying Theorem~\ref{thm:spectral-duality} identifies this maximum with both $\Gamma_{\F}(W)$ and $\log_2\aR_{\F}(\WT(Q))$, proving~\eqref{eq:weighted-spectral-formula}. Taking $W=A_G$ gives the graph identity. The maximum of these linear functions is convex and positively homogeneous; their nonnegative coefficients also make it monotone. For subadditivity, the value of each functional at $W+W'$ is at most $\Gamma_{\F}(W)+\Gamma_{\F}(W')$, so the same holds for their maximum. Finally, Lemma~\ref{lem:edge-weights} bounds every coefficient by one. Thus, for every $\phi\in\Phi_{n,\F}$,
		\[
		\sum_{1\le i<j\le n}W_{i,j}P_{i,j}^{(\phi)}
		\le\Gamma_{\F}(W')+\sum_{1\le i<j\le n}|W_{i,j}-W'_{i,j}|.
		\]
		Maximizing over $\phi$ and then interchanging $W,W'$ proves~\eqref{eq:weighted-lipschitz}.
	\end{proof}
	
	\begin{remark}[real edge weights]
		For an arbitrary nonnegative weight matrix $W$, define a dimension matrix $Q^{(s)}$ by $Q^{(s)}_{i,j}=2^{\lfloor sW_{i,j}\rfloor}$ for $1\le i,j\le n$ and every positive integer $s$. Lemma~\ref{lem:weighted-spectral-formula}, positive homogeneity, and~\eqref{eq:weighted-lipschitz} give
		\[
		\Gamma_{\F}(W)
		=\lim_{s\to\infty}\frac{1}{s}\log_2\aR_{\F}(\WT(Q^{(s)})).
		\]
		Indeed, each entry of the logarithmic edge-weight matrix of $Q^{(s)}$, divided by $s$, differs from the corresponding entry of $W$ by less than $1/s$. Thus $\Gamma_{\F}$ extends normalized logarithmic asymptotic rank continuously to all nonnegative real edge weights. The normalization $\tau_{\F}(W)$ measures this exponent per unit of total edge weight.
	\end{remark}
	
	The two asymptotic tensor parameters have complementary cut bounds. For every dimension matrix $Q$ with logarithmic edge-weight matrix $W$,
	\[
	\MaxCut(W)\le\log_2\aR_{\F}(\WT(Q)),\qquad
	\log_2\aQ_{\F}(\WT(Q))=\MinCut(W).
	\]
	Proposition~\ref{prop:cut} proves the first inequality by flattening across a maximum cut. Vrana and Christandl~\cite[Theorem~1 and Section~3]{vc2017} proved the minimum cut equality in the setting of entanglement conversion, including nonuniform edge dimensions. Appendix~\ref{sec:subrank} gives an explicit tensor proof of their weighted formula over every field, including finite fields. For ordinary graphs these statements specialize to the maximum cut lower bound and minimum cut equality with binary edge dimensions. In particular,
	\[
	\MaxCut(K_n)=\left\lfloor\frac{n^2}{4}\right\rfloor,
	\qquad \MinCut(K_n)=n-1,
	\]
	since a cut with $s$ vertices on one side has $s(n-s)$ edges.

	\subsection{Linear conic optimization}
	\label{subsec:conic-preliminaries}
	The weighted spectral formula expresses the graph exponent as the maximum of a linear function of spectral edge weights. Section~\ref{sec:weighted-cp} will bound this maximum by linear optimization over cones of real symmetric matrices. We recall the conic framework and its duality theorem here, then specify the cones and covering problems needed for that argument.
	
	\begin{definition}[linear conic programs and dual cones]
		\label{def:conic-program}
		Let $V$ be a finite-dimensional real inner-product space and let $\mathcal K\subseteq V$ be a closed convex cone. Its \emph{dual cone} is
		\[
		\mathcal K^*=\{y\in V\mid \langle y,x\rangle\ge0\text{ for every }x\in\mathcal K\}.
		\]
		We write $\operatorname{int}(\mathcal K)$ for the interior of $\mathcal K$ in $V$: a point $x$ belongs to this interior if some open ball about $x$, for the norm induced by the inner product, is contained in $\mathcal K$.
		Given a linear map $\mathcal A:V\to\mathbb R^r$, vectors $b\in\mathbb R^r$ and $c\in V$, and the adjoint map $\mathcal A^*$, the associated primal and dual linear conic programs are
		\[
		\begin{aligned}
			\text{minimize }&\langle c,x\rangle
			&\quad\text{subject to }&\mathcal A(x)=b,\quad x\in\mathcal K,\\
			\text{maximize }&\langle b,y\rangle
			&\quad\text{subject to }&c-\mathcal A^*(y)\in\mathcal K^*.
		\end{aligned}
		\]
		Linear inequalities can be included by adjoining nonnegative slack variables to the cone.
	\end{definition}

	For every feasible pair, weak duality follows from
	\[
	\langle c,x\rangle-\langle b,y\rangle
	=\langle c-\mathcal A^*(y),x\rangle\ge0.
	\]
	\emph{Slater's condition} requires a feasible point in the interior of the relevant cone. It is a sufficient condition for equality of the primal and dual optimal values. Strict-feasibility conditions for convex programming originate in Slater's work~\cite{slater1950}; Duffin~\cite{duffin1956} developed an early duality theory for linear programs over general cones. We use the following finite-dimensional formulation; see also~\cite[Sections~5.2 and~5.9]{boyd2004}.
	
	\begin{theorem}[conic strong duality under Slater's condition]
		\label{thm:conic-strong-duality}
		Consider the primal--dual pair of Definition~\ref{def:conic-program}, and suppose that the primal is feasible and has finite infimum. Each of the following conditions implies equality of the primal infimum and the dual supremum:
		\begin{enumerate}
			\item \emph{Primal strict feasibility:} there is an $x\in\operatorname{int}(\mathcal K)$ with $\mathcal A(x)=b$. In this case the dual supremum is attained.
			\item \emph{Dual strict feasibility:} there is a $y\in\mathbb R^r$ with $c-\mathcal A^*(y)\in\operatorname{int}(\mathcal K^*)$. In this case the primal infimum is attained.
		\end{enumerate}
		If both conditions hold, both optimal values are attained and equal.
	\end{theorem}
	
	For a formulation with linear inequalities, strict feasibility is understood after adjoining their slack variables: it requires interior membership in the original cone and strict satisfaction of the inequalities. Equality constraints remain exact.
	
	We use $\mathbb S^n$ for the real symmetric $n\times n$ matrices, with inner product $\langle X,Z\rangle=\operatorname{tr}(XZ)$, and $I_n$ for the identity matrix. We write $\mathbf e_1,\ldots,\mathbf e_n$ for the standard coordinate vectors of $\mathbb R^n$. A matrix $X$ is \emph{positive semidefinite}, written $X\succeq0$, if $a^{\mathsf T}Xa\ge0$ for every $a\in\mathbb R^n$. More generally, $X\succeq Y$ means that $X-Y$ is positive semidefinite. This order on symmetric matrices is distinct from the restriction preorder $T\succeq S$ on tensors: the type of the arguments specifies which relation is intended.
	
	\begin{definition}[semidefinite, completely positive, and copositive cones]
		\label{def:matrix-cones}
		Define
		\[
		\begin{aligned}
			\mathbb S_+^n
			&=\{Z\in\mathbb S^n\mid Z\succeq0\},\\
			\CP_n
			&=\left\{\sum_{\nu=1}^t b^{(\nu)}(b^{(\nu)})^{\mathsf T}
			\;\middle|\;t\in\mathbb N,\ b^{(\nu)}\in\mathbb R_{\ge0}^n\ (1\le\nu\le t)\right\},\\
			\COP_n
			&=\{X\in\mathbb S^n\mid a^{\mathsf T}Xa\ge0
			\text{ for every }a\in\mathbb R_{\ge0}^n\}.
		\end{aligned}
		\]
		These are the positive semidefinite, completely positive, and copositive cones. A representation in the definition of $\CP_n$ is a \emph{completely positive factorization}.
	\end{definition}
	
	All three cones are closed and convex~\cite{berman2003,dur2010}, and
	\[
	(\mathbb S_+^n)^*=\mathbb S_+^n,
	\qquad \CP_n^*=\COP_n,
	\qquad \COP_n^*=\CP_n.
	\]
	For example, $\langle X,bb^{\mathsf T}\rangle=b^{\mathsf T}Xb$ proves the middle identity, and the bipolar theorem gives the last. Optimization over $\mathbb S_+^n$ is \emph{semidefinite programming} (SDP); optimization over $\CP_n$ is \emph{completely positive programming} (CP). Every completely positive matrix is positive semidefinite and entrywise nonnegative. Conversely, every entrywise nonnegative matrix in $\mathbb S_+^n$ is completely positive when $n\le4$, whereas this converse fails for every $n\ge5$~\cite{maxfieldminc1962,diananda1962}. Thus the low-dimensional equality concerns positive semidefinite matrices with nonnegative entries, rather than the whole positive semidefinite cone.
	
	The particular conic programs in this paper cover edge weights. Their general form allows a symmetric cost matrix $C$ and additional linear constraints, which makes their common structure explicit.
	
	\begin{definition}[conic matrix covers]
		\label{def:general-conic-cover}
		Let $W\in\mathbb S^n$ be nonnegative with zero diagonal, let $C\in\mathbb S^n$, let $\mathcal A:\mathbb S^n\to\mathbb R^r$ be linear, and let $b\in\mathbb R^r$. For $\mathcal K=\CP_n$ or $\mathcal K=\mathbb S_+^n$, define
		\[
		U_{\mathcal K}(W;C,\mathcal A,b)
		=\frac14\inf_{\substack{Z\in\mathcal K,\ \mathcal A(Z)=b\\
				Z_{i,j}\ge W_{i,j}\ (1\le i<j\le n)}}\bigl(\langle C,Z\rangle\bigr).
		\]
		The corresponding dual is
		\[
		\sup_{\substack{P=P^{\mathsf T},\ P_{i,i}=0,\ P_{i,j}\ge0\\
				y\in\mathbb R^r,\ C-2P-\mathcal A^*(y)\in\mathcal K^*}}
		\left(\sum_{1\le i<j\le n}W_{i,j}P_{i,j}
		+\frac14\langle b,y\rangle\right).
		\]
		Here the entrywise conditions on $P$ range over $1\le i,j\le n$. When there are no additional linear constraints, the map $\mathcal A$, the vector $b$, and the dual variable $y$ are omitted.
	\end{definition}
	
	Indeed, using multipliers $P_{i,j}\ge0$ for the edge constraints and $y/4$ for $\mathcal A(Z)=b$, the Lagrangian is
	\[
	\sum_{1\le i<j\le n}W_{i,j}P_{i,j}
	+\frac14\langle b,y\rangle
	+\frac14\langle C-2P-\mathcal A^*(y),Z\rangle.
	\]
	Its infimum over $Z\in\mathcal K$ gives the stated dual. Whenever the hypotheses of Theorem~\ref{thm:conic-strong-duality} hold, the two optimal values agree. Section~\ref{sec:weighted-cp} specializes these covers to the all-one cost matrix and establishes the relevant attainment conditions directly.
	
	\section{Separating Several Shared Tensor Legs}
	\label{sec:separation}
	
	We begin the complete-graph argument by extending shared-leg separation to higher order. The tensors under consideration have blocks with disjoint coordinates on two legs and may share all their remaining legs. These matched blocks form a partial direct sum in the sense of Definition~\ref{def:partial-direct-sum}. We show that a linear number of full direct-sum copies suffices to separate all the shared legs at once, with each output block retaining an edge tensor. We then apply this construction to blocks of a fixed type to obtain the entropy inequality used in the polynomial argument. The construction extends the single-shared-leg separation of OpenAI~\cite[Proposition~3.1 and Corollary~3.2]{openai2026}, using the polynomial-resource formulation of~\cite{abz2026} to work over an arbitrary field.

	\subsection{An interpolation resource}
	\label{subsec:interpolation-resource}
	We begin with a polynomial tensor whose asymptotic cost is linear in the sum of its input lengths. This is the higher-order analogue of the resource used in~\cite{openai2026,abz2026}. Its support supplies the linear relation needed for separation over every characteristic.
	
	\begin{definition}[polynomial multiplication tensor]
		\label{def:polynomial}
		Let $n\geq2$ be an integer, and let $a_1,\ldots,a_{n-1}$ be positive integers. Put $D=1+\sum_{j=1}^{n-1}(a_j-1)$. The \emph{polynomial multiplication tensor} is the $n$-th order tensor
		\[
		\Poly_n(a_1,\ldots,a_{n-1})
		=\sum_{i_1=0}^{a_1-1}\cdots\sum_{i_{n-1}=0}^{a_{n-1}-1}
		e_{i_1}^{(a_1)}\Cprod\cdots\Cprod
		e_{i_{n-1}}^{(a_{n-1})}\Cprod
		e_{i_1+\cdots+i_{n-1}}^{(D)}.
		\]
		Its first $n-1$ legs are the input legs, and its last leg is the output leg. Thus its dimension vector is $(a_1,\ldots,a_{n-1},D)$.
	\end{definition}
	
	For $1\leq j\leq n-1$, write an input polynomial as $f_j(x)=\sum_{i=0}^{a_j-1}u_{j,i}x^i$. Contracting the first $n-1$ tensor legs against these coefficient vectors gives the coefficient vector of $\prod_{j=1}^{n-1}f_j(x)$.
	
	The next lemma is the many-input version of the interpolation and output-flattening argument in~\cite[Equation~(4.2)]{openai2026}.
	
	\begin{lemma}[asymptotic rank of polynomial multiplication]
		\label{lem:polynomial-rank}
		For every field $\F$, every integer $n\geq2$, and all positive integers $a_1,\ldots,a_{n-1}$,
		\[
		\aR_{\F}\bigl(\Poly_n(a_1,\ldots,a_{n-1})\bigr)
		=1+\sum_{j=1}^{n-1}(a_j-1).
		\]
		Moreover, $\Poly_n(a_1,\ldots,a_{n-1})\unlhd\Id_n(D)$, where $D=1+\sum_{j=1}^{n-1}(a_j-1)$. If $\F$ contains at least $D$ distinct elements, then the ordinary rank also equals $D$.
	\end{lemma}
	\begin{proof}
		Put $D=1+\sum_{j=1}^{n-1}(a_j-1)$, and choose $D$ distinct elements $\xi_1,\ldots,\xi_D$ of the infinite rational function field $\F(z)$. For every integer $1\leq r\leq D$, let
		\[
		L_r(x)=\prod_{\substack{1\leq s\leq D\\s\neq r}}
		\frac{x-\xi_s}{\xi_r-\xi_s},
		\]
		and let $c_r\in\F(z)^D$ be its coefficient vector, in increasing order of degree. Lagrange interpolation states that every polynomial $f$ of degree at most $D-1$ satisfies $f(x)=\sum_{r=1}^D f(\xi_r)L_r(x)$. Apply this identity to $f=\prod_{j=1}^{n-1}f_j$, whose degree is at most $D-1$. For $1\leq j\leq n-1$ and $1\leq r\leq D$, define the evaluation vector
		\[
		v_{j,r}=\sum_{i=0}^{a_j-1}\xi_r^i e_i^{(a_j)}.
		\]
		The interpolation formula is an identity in the formal input coefficients over $\F(z)$. Equating their coefficients therefore gives the tensor identity
		\[
		\Poly_n(a_1,\ldots,a_{n-1})
		=\sum_{r=1}^D v_{1,r}\Cprod\cdots\Cprod v_{n-1,r}\Cprod c_r.
		\]
		This is a sum of $D$ simple tensors over $\F(z)$. Lemma~\ref{lem:rational-specialization} gives $\Poly_n(a_1,\ldots,a_{n-1})\unlhd\Id_n(D)$ over $\F$, and Lemma~\ref{lem:degeneration-asymptotic} gives the asymptotic rank upper bound $D$. If $\F$ itself has $D$ distinct elements, choose the interpolation points there instead; the same identity gives ordinary rank at most $D$.
		
		For the lower bound, put $S=\{1,\ldots,n-1\}$ and use the flattening from Definition~\ref{def:flattening-rank}. Its rows are input tuples $\boldsymbol i=(i_1,\ldots,i_{n-1})$, and its columns are output indices $0\le s<D$, with
		\[
		\operatorname{Flat}_S(\Poly_n(a_1,\ldots,a_{n-1}))_{\boldsymbol i,s}
		=\begin{cases}1,&i_1+\cdots+i_{n-1}=s,\\0,&\text{otherwise}.\end{cases}
		\]
		Every value $0\le s<D$ is a sum of allowed input indices, so every column is nonzero. The column supports are pairwise disjoint, since an input tuple has a unique sum. Thus the $D$ columns are linearly independent over $\F$ and the flattening rank is $D$. Lemma~\ref{lem:flattening-properties} gives ordinary rank at least $D$ and flattening rank $D^N$ for the $N$-th Kronecker power. Taking $N$-th roots gives asymptotic rank at least $D$. Together with interpolation, this proves both the asserted asymptotic equality and the ordinary rank equality when $\F$ has enough interpolation points.
	\end{proof}

	\begin{remark}[finite fields]
		If the field has fewer than $D$ elements, interpolation at $D$ distinct field elements is unavailable, and the ordinary rank need not equal $D$. The preceding degeneration and coefficient-extraction argument still gives asymptotic rank $D$ over the original field.
	\end{remark}

	\subsection{A finite separation construction}
	\label{subsec:finite-separation}
	We extend the shared-leg separation construction of OpenAI~\cite[Section~3]{openai2026} to simultaneous separation on all $k-2$ shared legs. The auxiliary polynomial tensor records a sum of tentative block labels. Integer weights then retain exactly the entries for which all shared legs choose the same block. Since these weights are exponents of a formal variable, the construction works unchanged in positive characteristic.
	
	\begin{theorem}[simultaneous separation]
		\label{thm:separation}
		Fix a field $\F$, integers $k\ge2$ and $M\ge1$, and $e=(k-1,k)$. Let $A_0,\ldots,A_{M-1}$ be $k$-th order tensors in aligned coordinates, with common dimensions on the first $k-2$ legs, and put
		\[
		A=A_0\boxplus_e\cdots\boxplus_e A_{M-1},
		\qquad D_{k,M}=1+(k-1)(M-1).
		\]
		Then
		\[
		A^{\Ksum D_{k,M}}\rightsquigarrow
		\bigKsum_{h=0}^{M-1}\bigl(A_h\Kprod E^{(M)}_{k-1,k}\bigr).
		\]
		Zero blocks are permitted with their specified ambient dimensions. The sum on the right separates all $k$ legs.
	\end{theorem}
	\begin{proof}
		We first construct a degeneration that makes the block label $h$, already visible on the last two legs, visible on every shared leg as well.
		
		Put $\ell=k-2$. More generally, for an arbitrary positive integer $t$, we prove
		\begin{equation}
			\bigKsum_{h=0}^{M-1}\bigl(A_h\Kprod E^{(t)}_{k-1,k}\bigr)
			\unlhd A\Kprod\Poly_k(\underbrace{M,\ldots,M}_{\ell\text{ inputs}},t).
			\label{eq:resource-separation}
		\end{equation}
		Here $t$ is the dimension of the edge retained in each block.
		
		For $k=2$ the auxiliary tensor is $\Poly_2(t)=\Id_2(t)$, and the list of length-$M$ inputs is empty.
		
		Let $d_1,\ldots,d_\ell$ be the common shared dimensions, and let $b_h,c_h$ be the last two dimensions of $A_h$. Define
		\[
		\mathcal I_h=\left(\prod_{j=1}^{\ell}[d_j]_0\right)\times[b_h]_0\times[c_h]_0.
		\]
		Use the formal multilinear notation of Definition~\ref{def:multilinear-notation}. For $(a_1,\ldots,a_\ell,b,c)\in\mathcal I_h$, write a term of $A_h$ as
		\[
		\alpha_{h,a_1,\ldots,a_\ell,b,c}
		\left(\prod_{j=1}^{\ell}x_{j,a_j}\right)y_{h,b}z_{h,c},
		\]
		where $a_j$ runs through the common coordinates of shared leg $j$, and $b,c$ run through the local coordinates of block $h$ on the two separate legs. To record the block label on the shared legs, use the auxiliary indices $g_1,\ldots,g_\ell$ as tentative labels. The remaining auxiliary indices are $u,v$, with
		\[
		0\le g_j<M,\qquad 0\le u<t,\qquad
		v=g_1+\cdots+g_\ell+u.
		\]
		The output index lies between $0$ and $\ell(M-1)+t-1$. Pairing the corresponding legs in the Kronecker product gives coordinate symbols
		$X_{j,a_j,g_j}$, $Y_{h,b,u}$, and $Z_{h,c,v}$.
		We want the total weight to be $\sum_{j=1}^{\ell}(g_j-h)^2$. A shared leg does not know $h$, so it cannot contribute the cross term $-2hg_j$ directly. The last leg does know $h$ and the sum $v$; using $v=\sum_{j=1}^{\ell}g_j+u$ moves all cross terms onto the two separate legs. This motivates the following local integer weights:
		\[
		\begin{aligned}
			w(X_{j,a,g})&=g^2
			&&(1\le j\le\ell,\ 0\le a<d_j,\ 0\le g<M),\\
			w(Y_{h,b,u})&=2hu+\ell h^2
			&&(0\le h<M,\ 0\le b<b_h,\ 0\le u<t),\\
			w(Z_{h,c,v})&=-2hv
			&&(0\le h<M,\ 0\le c<c_h,\ 0\le v<\ell(M-1)+t).
		\end{aligned}
		\]
		On every supported term the sum of the weights is
		\[
		\sum_{j=1}^{\ell}g_j^2+2hu+\ell h^2-2hv
		=\sum_{j=1}^{\ell}(g_j-h)^2.
		\]
		This is nonnegative and vanishes precisely when $g_j=h$ for every shared leg. The argument concerns integers indexing powers of the formal degeneration parameter; it never reduces the squares modulo the characteristic of $\F$.
		
		Make the coordinate substitutions $x\mapsto\varepsilon^{w(x)}x$ from Remark~\ref{rem:weights}. The preceding nonnegativity calculation shows that the resulting expression is polynomial in $\varepsilon$. Its constant coefficient is
		\[
		\sum_{h=0}^{M-1}\sum_{(a_1,\ldots,a_\ell,b,c)\in\mathcal I_h}\sum_{u=0}^{t-1}
		\alpha_{h,a_1,\ldots,a_\ell,b,c}
		\left(\prod_{j=1}^{\ell}X_{j,a_j,h}\right)
		Y_{h,b,u}Z_{h,c,\ell h+u}.
		\]
		On the last leg, make the further coordinate substitutions $Z_{h,c,\ell h+u}\mapsto\widehat Z_{h,c,u}$ for $0\le h<M$, $0\le c<c_h$, and $0\le u<t$, where the $\widehat Z_{h,c,u}$ are new output symbols, and set all unused last-leg symbols to zero. The block label $h$ is already part of each last-leg coordinate, so these substitutions are well defined. In sector $h$, the resulting expression represents $A_h\Kprod E^{(t)}_{k-1,k}$. Distinct $h$ now use separate coordinates on every shared leg as well; the last two legs were already separate. Project away unused auxiliary coordinates, and omit zero summands in the full direct sum. These are restrictions; any zero coordinates retained in a chosen target representative can be restored by local embeddings. Thus the original coefficient table of each $A_h$ is unchanged, the label $u$ contributes exactly the retained edge, and $h$ now distinguishes all legs. Remark~\ref{rem:weights}, followed by these coordinate projections, proves~\eqref{eq:resource-separation}, including zero blocks with their specified coordinates.
		
		It remains to obtain this auxiliary tensor from the stated number of source copies. Take $t=M$. Lemma~\ref{lem:polynomial-rank} gives
		\[
		\Poly_k(M,\ldots,M)\unlhd\Id_k(D_{k,M}).
		\]
		Taking its Kronecker product with $A$ and composing with~\eqref{eq:resource-separation} gives
		\[
		\begin{aligned}
			\bigKsum_{h=0}^{M-1}\bigl(A_h\Kprod E^{(M)}_{k-1,k}\bigr)
			&\unlhd A\Kprod\Poly_k(M,\ldots,M)\\
			&\unlhd A\Kprod\Id_k(D_{k,M})
			\cong A^{\Ksum D_{k,M}}.
		\end{aligned}
		\]
		The final isomorphism identifies multiplication by the diagonal tensor with a full direct sum of $D_{k,M}$ copies of $A$. Thus these copies degenerate to the stated full direct sum, with an $M$-dimensional edge retained in every block.
	\end{proof}
	
	\subsection{The entropy inequality}
	\label{subsec:entropy}
	
	We now extend the entropy inequality of OpenAI~\cite[Section~3]{openai2026} to $k$-th order tensors using the simultaneous separation theorem. The construction has two effects on spectral values: it produces $M$ independent blocks and leaves an $M$-dimensional edge in each block. Its cost is only $D_{k,M}\le(k-1)M$ source copies, so the linear cost cancels the number of blocks asymptotically, leaving the contribution of the retained edge. To quantify this gain, we apply the construction to Kronecker powers and retain sector words with prescribed frequencies. The corresponding tensor blocks have the same spectral value, while entropy controls the number of words. This gives the following inequality.
	
	\begin{corollary}[entropy inequality]
		\label{cor:entropy}
		Let $k\ge2$ and $s\ge1$ be integers. Let $A=A_1\boxplus_e\cdots\boxplus_e A_s$, where $e=(k-1,k)$ and the nonzero tensors $A_1,\ldots,A_s$ have aligned shared coordinates as in Definition~\ref{def:partial-direct-sum}. Let $\phi\in\Phi_{k,\F}$ and let $P_{k-1,k}^{(\phi)}$ be its spectral edge weight from Definition~\ref{def:spectral-edge-weights}. For every probability vector $\theta=(\theta_1,\ldots,\theta_s)$,
		\[
		\phi(A)\ge
		\exp\bigl(P_{k-1,k}^{(\phi)}H(\theta_1,\ldots,\theta_s)\bigr)
		\prod_{i=1}^{s}\phi(A_i)^{\theta_i},
		\]
		where $H(\theta_1,\ldots,\theta_s)=-\sum_{i=1}^{s}\theta_i\ln\theta_i$ and $0\ln0=0$.
	\end{corollary}
	\begin{proof}
		Put $p=P_{k-1,k}^{(\phi)}$. Lemma~\ref{lem:edge-weights} gives $\phi(E^{(q)}_{k-1,k})=q^p$ for every positive integer $q$.
		
		First suppose that all entries of $\theta$ are rational, and let $N$ range through positive multiples of their common denominator. Define the set of sector words
		\[
		\mathcal W_N=
		\bigl\{w\in\{1,\ldots,s\}^N:
		|\{r\in[N]\mid w_r=i\}|=N\theta_i
		\text{ for every }1\le i\le s\bigr\}.
		\]
		Put
		\[
		M_N=|\mathcal W_N|=\frac{N!}{\prod_{i=1}^{s}(N\theta_i)!},
		\qquad T_w=\bigKprod_{r=1}^{N}A_{w_r}\quad(w\in\mathcal W_N),
		\]
		where the factors in $T_w$ retain their order $r=1,\ldots,N$. Let $d_1,\ldots,d_{k-2}$ be the common dimensions of the shared legs of the $A_i$. Every $T_w$ has dimension $d_j^N$ on shared leg $j$, with the same product coordinate basis. Enumerate $\mathcal W_N$ in lexicographic order as $w^{(0)},\ldots,w^{(M_N-1)}$, and define the aligned partial direct sum
		\[
		B_N=T_{w^{(0)}}\boxplus_e\cdots\boxplus_e T_{w^{(M_N-1)}}.
		\]
		Thus $B_N$ has $M_N$ blocks, separated on the last two legs and aligned on the first $k-2$ legs, exactly as in Definition~\ref{def:partial-direct-sum}.
		
		To obtain $B_N$ from $A^{\Kprod N}$, on each of the last two legs project onto the coordinate sectors whose words belong to $\mathcal W_N$, and order the retained sectors as above. Use the identity on each shared leg. Every supported term has the same sector word on the last two legs, because each factor of $A$ has matched sectors. The block at word $w$ is precisely $T_w$, so these local maps prove $B_N\preceq A^{\Kprod N}$.
		
		Each $T_w$ contains exactly $N\theta_i$ factors equal to $A_i$ for every $1\le i\le s$. Multiplicativity of $\phi$ therefore gives the common block value
		\[
		\phi(T_w)=C_N,\qquad
		C_N=\prod_{i=1}^{s}\phi(A_i)^{N\theta_i}\quad(w\in\mathcal W_N).
		\]
		Apply Theorem~\ref{thm:separation} with $M=M_N$ and with its block $A_h$ equal to $T_{w^{(h)}}$ for $0\le h<M_N$. The displayed partial direct sum for $B_N$ satisfies its alignment hypothesis, and the theorem gives
		\[
		B_N^{\Ksum D_{k,M_N}}\rightsquigarrow
		\bigKsum_{h=0}^{M_N-1}
		\bigl(T_{w^{(h)}}\Kprod E^{(M_N)}_{k-1,k}\bigr).
		\]
		Here $D_{k,M_N}=1+(k-1)(M_N-1)$. Using $D_{k,M_N}\le(k-1)M_N$, restriction monotonicity, and degeneration monotonicity from Lemma~\ref{lem:degeneration-asymptotic}, followed by additivity and multiplicativity, yields
		\begin{align*}
			(k-1)M_N\phi(A)^N
			&\ge D_{k,M_N}\phi(B_N)\\
			&\ge\sum_{h=0}^{M_N-1}\phi(T_{w^{(h)}})\phi(E^{(M_N)}_{k-1,k})\\
			&=M_N^{1+p}C_N.
		\end{align*}
		Canceling $M_N$ and taking $N$-th roots yields
		\[
		\phi(A)\ge
		(k-1)^{-1/N}M_N^{p/N}
		\prod_{i=1}^{s}\phi(A_i)^{\theta_i}.
		\]
		
		Stirling's formula, with $s$ and $\theta$ fixed, gives
		\[
		\frac{\ln M_N}{N}
		=-\sum_{i=1}^{s}\theta_i\ln\theta_i
		+O_s\left(\frac{\ln(N+1)}{N}\right)
		\longrightarrow H(\theta_1,\ldots,\theta_s).
		\]
		The prefactor $(k-1)^{-1/N}$ tends to one. Thus the growth in the number of sectors contributes exactly $\exp(pH(\theta_1,\ldots,\theta_s))$, proving the claimed inequality for rational $\theta$. Spectral monotonicity is applied to each fixed degeneration before taking this limit; Lemma~\ref{lem:degeneration-asymptotic} allows its degree to depend on $N$.
		
		Finally, every $A_i$ is nonzero, so $\phi(A_i)\ge1$. Both entropy and the product $\prod_{i=1}^{s}\phi(A_i)^{\theta_i}$ are continuous on the probability simplex. Approximating an arbitrary probability vector by rational ones therefore gives
		\[
		\phi(A)\ge
		\exp\bigl(P_{k-1,k}^{(\phi)}H(\theta_1,\ldots,\theta_s)\bigr)
		\prod_{i=1}^{s}\phi(A_i)^{\theta_i}
		\]
		for every probability vector $\theta$, as asserted.
	\end{proof}
	
	\section{Polynomial Multiplication and the Complete-Graph Bound}
	\label{sec:polynomial-method}
	
	The entropy inequality of Corollary~\ref{cor:entropy} is the bridge from the separation construction to the complete-graph bound. The proof compares two estimates for a symmetric spectral profile of polynomial multiplication: interpolation gives an upper bound, while two degenerations force growth as the input lengths increase. Comparing these estimates constrains the spectral values of the complete graph tensor.
	
	\subsection{Polynomial tensors and spectral profiles}
	\label{sec:profile}
	
	For $n\ge3$, the tensor $\Poly_n$ from Definition~\ref{def:polynomial} has $n-1$ independent polynomial inputs, each retained as a separate leg. Lemma~\ref{lem:polynomial-rank} gives its linear interpolation bound over every field. We now symmetrize its spectral values over all $n$ legs, extending the profile used by OpenAI~\cite[Section~4]{openai2026}. The normalization uses the spectral edge weights from Definition~\ref{def:spectral-edge-weights}.
	
	The following definition extends the six-factor symmetrization used in~\cite[Equation~(4.4)]{openai2026} to all permutations of the $n$ tensor legs.
	
	\begin{definition}[symmetric spectral profile]
		\label{def:profile}
		Let $n\geq3$ be an integer, let $\phi\in\Phi_{n,\F}$, and let $P_{i,j}^{(\phi)}$ be its spectral edge weights. Define their average by
		\[
		m=\binom n2,
		\quad t_\phi=\frac1m\sum_{1\leq i<j\leq n}P_{i,j}^{(\phi)}.
		\]
		For every permutation $\pi\in S_n$, put $\phi_\pi(T)=\phi(\pi T)$, where the action on tensor legs is defined in Definition~\ref{def:leg-permutation}. For all positive integers $a_1,\ldots,a_{n-1}$, define the \emph{symmetric spectral profile}
		\[
		P_\phi(a_1,\ldots,a_{n-1})
		=\left(\prod_{\pi\in S_n}
		\phi_\pi\bigl(\Poly_n(a_1,\ldots,a_{n-1})\bigr)
		\right)^{1/(n!t_\phi)}.
		\]
	\end{definition}
	
	Lemma~\ref{lem:edge-weights} gives $t_\phi\geq1/m>0$, so the definition uses a positive exponent. Every map $\phi_\pi$ is a spectral point: permutation of the legs preserves restriction, the diagonal tensors, and both semiring operations.
	
	For integers $1\leq u<v\leq n$ and $\pi\in S_n$, Definition~\ref{def:spectral-edge-weights} gives $P_{u,v}^{(\phi_\pi)}=\log_2\phi_\pi(E^{(2)}_{u,v})$. The image of the pair $(u,v)$ under $\pi$ is the pair $\left(\min(\pi(u),\pi(v)),\max(\pi(u),\pi(v))\right)$. For any fixed pair $1\leq i<j\leq n$, exactly $2(n-2)!$ permutations send $(u,v)$ to $(i,j)$ in this sense: there are two choices for the images of $u,v$, and $(n-2)!$ choices for the remaining images. Consequently,
	\begin{equation}
		\sum_{\pi\in S_n}P_{u,v}^{(\phi_\pi)}
		=2(n-2)!\sum_{1\leq i<j\leq n}P_{i,j}^{(\phi)}
		=n!t_\phi
		\; (1\leq u<v\leq n).
		\label{eq:edge-average}
	\end{equation}
	
	\begin{lemma}[basic properties of the profile]
		\label{lem:profile-basic}
		For every integer $n\geq3$ and every spectral point $\phi\in\Phi_{n,\F}$, the profile in Definition~\ref{def:profile} is positive and symmetric in its $n-1$ arguments. For every positive integer $b$ and all positive integers $a_1,\ldots,a_{n-1}$, it satisfies
		\[
		P_\phi(1,\ldots,1,b)=b,
		\quad
		P_\phi(a_1,\ldots,a_{n-1})
		\leq\left(1+\sum_{j=1}^{n-1}(a_j-1)\right)^{1/t_\phi}.
		\]
	\end{lemma}
	\begin{proof}
		Every polynomial multiplication tensor is nonzero, and a spectral point has value at least one on a nonzero tensor. Thus each factor in the definition of $P_\phi$ is positive, and so is $P_\phi$.
		
		Let $\sigma$ permute the first $n-1$ tensor legs and fix leg $n$. By the defining coefficient formula, permuting the input lengths of $\Poly_n$ by $\sigma$ gives $\sigma\Poly_n$. The map $\pi\mapsto\pi\sigma$ is a bijection of $S_n$. Hence replacing $\Poly_n$ by $\sigma\Poly_n$ reindexes the factors defining $P_\phi$ and does not change their product. This proves symmetry.
		
		The normalization comes from a single edge. In $\Poly_n(1,\ldots,1,b)$, the first $n-2$ indices are zero, so its coefficient formula reduces to $E^{(b)}_{n-1,n}$. Lemma~\ref{lem:edge-weights}, applied to each spectral point $\phi_\pi$, therefore gives
		\[
		\phi_\pi\bigl(\Poly_n(1,\ldots,1,b)\bigr)
		=b^{P_{n-1,n}^{(\phi_\pi)}}.
		\]
		The permutation average~\eqref{eq:edge-average} says that the sum of these exponents is $n!t_\phi$. Substituting into Definition~\ref{def:profile} yields
		\[
		P_\phi(1,\ldots,1,b)
		=\left(b^{\sum_{\pi\in S_n}P_{n-1,n}^{(\phi_\pi)}}\right)^{1/(n!t_\phi)}
		=\left(b^{n!t_\phi}\right)^{1/(n!t_\phi)}=b.
		\]
		For the upper bound, put $D=1+\sum_{j=1}^{n-1}(a_j-1)$. Theorem~\ref{thm:spectral-duality} bounds the value of every spectral point by asymptotic rank, and Lemma~\ref{lem:polynomial-rank} computes that rank. Thus, for every $\pi\in S_n$,
		\[
		\phi_\pi\bigl(\Poly_n(a_1,\ldots,a_{n-1})\bigr)
		\le\aR_{\F}\bigl(\Poly_n(a_1,\ldots,a_{n-1})\bigr)=D.
		\]
		Multiplying these $n!$ inequalities and taking the power $1/(n!t_\phi)$ gives
		\[
		P_\phi(a_1,\ldots,a_{n-1})
		\le D^{1/t_\phi}
		=\left(1+\sum_{j=1}^{n-1}(a_j-1)\right)^{1/t_\phi},
		\]
		which proves the remaining bound on the profile.
	\end{proof}
	
	\subsection{Discrete concavity}
	\label{sec:polynomial-inequalities}
	\label{subsec:concavity}
	
	To compare input lengths $b-1$, $b$, and $b+1$, we add an edge that doubles the length-$b$ input and its output. The additional coordinates allow us to separate multiplication blocks of lengths $b+1$ and $b-1$. This extends the determinant construction of~\cite[Lemma~4.1]{openai2026}, while retaining every shared input as a separate leg.
	
	Fix an integer $n\ge3$, positive integers $a_1,\ldots,a_{n-2}$, and an integer $b\ge2$, and write
	\[
	S=\sum_{j=1}^{n-2}(a_j-1).
	\]
	In the formal multilinear notation of Definition~\ref{def:multilinear-notation}, the tensor $\Poly_n(a_1,\ldots,a_{n-2},b)\Kprod E^{(2)}_{n-1,n}$ is represented by
	\[
	\sum_{\substack{0\le i_j<a_j\ (1\le j\le n-2)\\0\le i<b,\ \alpha\in\{0,1\}}}
	\left(\prod_{j=1}^{n-2}X^{(j)}_{i_j}\right)
	Y_{i,\alpha}Z_{i+\sum_{j=1}^{n-2}i_j,\alpha}.
	\]
	Here the two values of $\alpha$ distinguish the two coordinates of the extra edge. In particular, the same $\alpha$ occurs on the last input and the output.
	
	To expose the two desired blocks, it is convenient to package these coordinates as polynomial coefficients. Let $V_d$ be the space of homogeneous polynomials of degree $d$ in $u,v$, with basis $(u^{d-i}v^i)_{i=0}^d$, and let
	\[
	U_d=\{Fs+Gw\mid F,G\in V_d\}.
	\]
	The symbols $u,v,s,w$ are auxiliary polynomial variables. Encode the input coordinates by
	\[
	f_j=\sum_{i=0}^{a_j-1}X^{(j)}_iu^{a_j-1-i}v^i,
	\qquad
	F=\sum_{i=0}^{b-1}u^{b-1-i}v^i(Y_{i,0}s+Y_{i,1}w).
	\]
	If $g_{r,0},g_{r,1}$ are the coefficients of $u^{S+b-1-r}v^rs$ and $u^{S+b-1-r}v^rw$ in $(\prod_{j=1}^{n-2}f_j)F$, then the preceding formal expression is exactly
	\[
	\sum_{r=0}^{S+b-1}\bigl(g_{r,0}Z_{r,0}+g_{r,1}Z_{r,1}\bigr).
	\]
	Indeed, each output coefficient is obtained by summing the products of input coordinates whose degree indices add to $r$. Thus the tensor is the coefficient array of the map
	\[
	\begin{aligned}
		V_{a_1-1}\times\cdots\times V_{a_{n-2}-1}\times U_{b-1}
		&\longrightarrow U_{S+b-1},\\
		(f_1,\ldots,f_{n-2},F)&\longmapsto
		\left(\prod_{j=1}^{n-2}f_j\right)F.
	\end{aligned}
	\]
	This map is multilinear, meaning linear in each argument when the others are fixed. Recording its output coordinates and pairing them with independent output symbols gives precisely the formal multilinear expression of its coefficient tensor. We will use this equivalence to describe local changes of bases without merging the shared input legs.
	
	The two blocks arise from replacing $s,w$ by $u,v$. This substitution forgets multiples of $D=uw-vs$. The vectors that remain visible will give the larger multiplication block, and the multiples of $D$ will give the smaller one. The following elementary basis calculation makes this decomposition explicit.
	
	\begin{lemma}[coordinates for the two multiplication blocks]
		\label{lem:binary-splitting}
		For every integer $d\ge1$, put $D=uw-vs$ and define
		\[
		\begin{aligned}
			q_{d,i}&=u^{d-i}v^is &&(0\le i\le d),\\
			q_{d,d+1}&=v^dw,\\
			k_{d,j}&=(uw-vs)u^{d-1-j}v^j &&(0\le j<d).
		\end{aligned}
		\]
		These vectors form an ordered basis of $U_d$, with the $q$-vectors first and the $k$-vectors second. The linear map
		\[
		\rho_d:U_d\longrightarrow V_{d+1},\qquad
		\rho_d(Fs+Gw)=uF+vG
		\]
		sends the $q$-vectors to the monomial basis of $V_{d+1}$ and has kernel spanned by the $k$-vectors. Under the identification $g\mapsto Dg$, this kernel is a copy of $V_{d-1}$.
	\end{lemma}
	\begin{proof}
		Substitution gives $\rho_d(q_{d,i})=u^{d+1-i}v^i$ for $0\le i\le d+1$, and $\rho_d(k_{d,j})=0$ because $uw-vs$ becomes $uv-vu=0$. To check independence, apply $\rho_d$ to a relation among the displayed vectors. The monomials force every coefficient of a $q$-vector to vanish. The remaining relation has the form
		$D\sum_{j=0}^{d-1}\lambda_ju^{d-1-j}v^j=0$.
		Since the polynomial ring over $\F$ has no zero divisors and $D\ne0$, all $\lambda_j$ vanish. There are $d+2+d=2d+2=\dim U_d$ independent vectors, so they form a basis. In this basis, a vector maps to zero precisely when all its $q$-coordinates vanish. Hence $\ker\rho_d$ is exactly the copy of $V_{d-1}$ spanned by the $k$-vectors.
	\end{proof}
	
	The basis separates what survives the substitution from what it kills. Multiplication by any shared polynomial product preserves the killed part: $f(Dg)=D(fg)$. Consequently a kernel input never produces a visible output. This one-way interaction is exactly what makes a weight degeneration possible.
	
	\begin{lemma}[two-block polynomial degeneration]
		\label{lem:concavity-degeneration}
		Let $n\ge3$, let $a_1,\ldots,a_{n-2}$ be positive integers, let $b\ge2$ be an integer, and put $e=(n-1,n)$. In aligned shared monomial coordinates,
		\[
		\Poly_n(a_1,\ldots,a_{n-2},b)\Kprod E^{(2)}_{n-1,n}
		\rightsquigarrow
		\Poly_n(a_1,\ldots,a_{n-2},b+1)\boxplus_e\Poly_n(a_1,\ldots,a_{n-2},b-1).
		\]
	\end{lemma}
	\begin{proof}
		Use Lemma~\ref{lem:binary-splitting} on the last input with $d=b-1$ and on the output with $d=S+b-1$. Both integers are positive. The bases depend only on these degrees; in particular they are the same for every choice of the shared inputs.
		
		Write their product as
		\[
		f=\prod_{j=1}^{n-2}f_j=\sum_{t=0}^{S}c_tu^{S-t}v^t,
		\qquad c_t=0\quad(t\in\mathbb Z\setminus\{0,\ldots,S\}).
		\]
		The identity $\rho_{S+b-1}(fF)=f\rho_{b-1}(F)$ says that the visible output coefficients are those of ordinary multiplication by $f$. The identity $f(Dg)=D(fg)$ says the same for the kernel coefficients. Explicitly,
		\begin{equation}
			\label{eq:quotient-kernel-coordinates}
			\begin{aligned}
				fq_{b-1,i}
				&=\sum_{r=0}^{S+b}c_{r-i}q_{S+b-1,r}
				+\sum_{r=0}^{S+b-2}\eta_{r,i}k_{S+b-1,r}
				&&(0\le i\le b),\\
				fk_{b-1,j}
				&=\sum_{r=0}^{S+b-2}c_{r-j}k_{S+b-1,r}
				&&(0\le j\le b-2),
			\end{aligned}
		\end{equation}
		where $\eta_{r,i}=0$ for $i<b$ and $\eta_{r,b}=c_{r-b+1}$.
		For $i<b$, multiplication of $q_{b-1,i}$ contains only $s$-terms and gives the first sum directly. For $i=b$, its $w$-terms are recovered using
		$k_{S+b-1,r}=u^{S+b-1-r}v^rw-u^{S+b-2-r}v^{r+1}s$;
		the accompanying $s$-terms are canceled by the first sum. The second line follows directly from $f(Dg)=D(fg)$. These identities exhibit the two multiplication blocks and the sole extra block $\eta$.
		
		Return to the formal multilinear notation of Definition~\ref{def:multilinear-notation}. Let $\widehat Y_i,\widetilde Y_j$ be the last-input coordinate symbols in the $q$- and $k$-bases, and let $\widehat Z_r,\widetilde Z_r$ be the corresponding output symbols. In these coordinates, equation~\eqref{eq:quotient-kernel-coordinates} gives the expression
		\[
		\begin{aligned}
			&\sum_{i=0}^{b}\sum_{r=0}^{S+b}c_{r-i}\widehat Y_i\widehat Z_r
			+\sum_{i=0}^{b}\sum_{r=0}^{S+b-2}\eta_{r,i}\widehat Y_i\widetilde Z_r\\
			&\hspace{35mm}
			+\sum_{j=0}^{b-2}\sum_{r=0}^{S+b-2}c_{r-j}\widetilde Y_j\widetilde Z_r.
		\end{aligned}
		\]
		Here $c_t$ and $\eta_{r,i}$ are multilinear polynomials in the separate shared coordinate groups $X^{(1)},\ldots,X^{(n-2)}$. More explicitly, if $B$ and $C$ have the new input and output basis vectors as columns in the old bases, then an input coordinate vector changes from $y$ to $\widetilde y$ by $y=B\widetilde y$, while the output coordinate vector changes by $g\mapsto C^{-1}g$. Hence the coefficient tensor is transformed by $B^{\mathsf T}$ on the last input leg and $C^{-1}$ on the output leg. Both maps are invertible. Thus the displayed expression represents a tensor locally isomorphic to the original one, with every shared coordinate group unchanged.
		
		We now remove only the middle sum. Assign the local integer weights
		\[
		\begin{aligned}
			w(X^{(j)}_i)&=0 &&(1\le j\le n-2,\ 0\le i<a_j),\\
			w(\widehat Y_i)&=0,\qquad w(\widetilde Y_j)=-1
			&&(0\le i\le b,\ 0\le j\le b-2),\\
			w(\widehat Z_r)&=0,\qquad w(\widetilde Z_t)=1
			&&(0\le r\le S+b,\ 0\le t\le S+b-2).
		\end{aligned}
		\]
		Make the coordinate substitutions $x\mapsto\varepsilon^{w(x)}x$ from Remark~\ref{rem:weights}. The three sums acquire factors $1,\varepsilon,1$, respectively. There is no term $\widetilde Y_j\widehat Z_r$, since a kernel input has no visible output, so the resulting expression is polynomial in $\varepsilon$. Its constant coefficient represents the tensor $T_0$ with formal multilinear expression
		\[
		\sum_{i=0}^{b}\sum_{r=0}^{S+b}c_{r-i}\widehat Y_i\widehat Z_r
		+\sum_{j=0}^{b-2}\sum_{r=0}^{S+b-2}c_{r-j}\widetilde Y_j\widetilde Z_r.
		\]
		In the first sum, the last-input and output lengths are $b+1$ and $S+b+1$; in the second, they are $b-1$ and $S+b-1$. The coefficients $c_{r-i}$ and $c_{r-j}$ are the corresponding polynomial multiplication coefficients. The last two coordinate sectors are disjoint and the shared coordinates remain aligned, so Definition~\ref{def:partial-direct-sum} identifies
		\[
		T_0=\Poly_n(a_1,\ldots,a_{n-2},b+1)
		\boxplus_e\Poly_n(a_1,\ldots,a_{n-2},b-1).
		\]
		Remark~\ref{rem:weights} gives the degeneration to this constant-coefficient tensor, proving the claim.
	\end{proof}
	
	It remains to turn this tensor comparison into an additive inequality. The entropy inequality supplies a weighted geometric mean of the two blocks; choosing its weights in proportion to their profile values converts that mean into their sum.
	
	\begin{lemma}[discrete concavity]
		\label{lem:concavity}
		For every integer $n\ge3$, every $\phi\in\Phi_{n,\F}$, all positive integers $a_1,\ldots,a_{n-2}$, and every integer $b\ge2$,
		\[
		2P_\phi(a_1,\ldots,a_{n-2},b)
		\ge P_\phi(a_1,\ldots,a_{n-2},b+1)
		+P_\phi(a_1,\ldots,a_{n-2},b-1).
		\]
		The same inequality holds in every input coordinate.
	\end{lemma}
	\begin{proof}
		Write $C=\Poly_n(a_1,\ldots,a_{n-2},b)$ and $C_\pm=\Poly_n(a_1,\ldots,a_{n-2},b\pm1)$, and put $e=(n-1,n)$. Lemma~\ref{lem:concavity-degeneration} gives the degeneration $C\Kprod E^{(2)}_{n-1,n}\rightsquigarrow C_+\boxplus_e C_-$ with aligned shared coordinates. For every $\pi\in S_n$ and every real number $0\leq\theta\leq1$, apply Corollary~\ref{cor:entropy} to this partial sum with spectral point $\phi_\pi$ and probability vector $(\theta,1-\theta)$. Combining its conclusion with degeneration monotonicity from Lemma~\ref{lem:degeneration-asymptotic} gives
		\[
		2^{P_{n-1,n}^{(\phi_\pi)}}\phi_\pi(C)
		\geq
		\exp\bigl(P_{n-1,n}^{(\phi_\pi)}H(\theta,1-\theta)\bigr)
		\phi_\pi(C_+)^\theta\phi_\pi(C_-)^{1-\theta}.
		\]
		The factor on the left is the spectral value of the additional edge tensor $E^{(2)}_{n-1,n}$. The entropy factor on the right uses the same pair of legs, because these are precisely the two unshared legs of the retained blocks.
		
		Use one common value of $\theta$ for all permutations. Multiply over all $\pi\in S_n$, take the power $1/(n!t_\phi)$, and apply~\eqref{eq:edge-average}. By the definition of $P_\phi$, the resulting inequality is
		\[
		2P_\phi(a_1,\ldots,a_{n-2},b)
		\geq e^{H(\theta,1-\theta)}
		P_\phi(a_1,\ldots,a_{n-2},b+1)^\theta
		P_\phi(a_1,\ldots,a_{n-2},b-1)^{1-\theta}.
		\]
		Put $A=P_\phi(a_1,\ldots,a_{n-2},b+1)$ and $B=P_\phi(a_1,\ldots,a_{n-2},b-1)$. Both are positive by Lemma~\ref{lem:profile-basic}, so we may choose $\theta=A/(A+B)$. Then $A/\theta=B/(1-\theta)=A+B$, and
		\[
		e^{H(\theta,1-\theta)}A^\theta B^{1-\theta}
		=\left(\frac A\theta\right)^\theta
		\left(\frac B{1-\theta}\right)^{1-\theta}
		=A+B.
		\]
		Substituting this identity and the definitions of $A,B$ gives
		\[
		\begin{aligned}
			2P_\phi(a_1,\ldots,a_{n-2},b)
			&\ge P_\phi(a_1,\ldots,a_{n-2},b+1)\\
			&\quad+P_\phi(a_1,\ldots,a_{n-2},b-1).
		\end{aligned}
		\]
		This is the asserted discrete concavity in the last input. By the symmetry proved in Lemma~\ref{lem:profile-basic}, permuting any chosen input into that position gives the same inequality in every input coordinate.
	\end{proof}
	
	\subsection{Shifted tripling}
	\label{subsec:tripling}
	
	We next compare an input length $h$ with $3h+S$, where $S$ is the sum of the other input degrees. The resulting factor of three provides a lower bound on a long increment of the profile; concavity will then give a lower bound on a single increment. The construction extends~\cite[Lemma~4.2]{openai2026}.

	For the formal comparison, it is useful to name the middle block without changing its shared coordinates.
	\begin{definition}[the middle polynomial block]
		\label{def:middle-polynomial-block}
		Let $n\ge3$, let $a_1,\ldots,a_{n-2},h$ be positive integers, and put $S=\sum_{j=1}^{n-2}(a_j-1)$. Define $B_n(a_1,\ldots,a_{n-2},h)$ in dimensions $(a_1,\ldots,a_{n-2},h+S,h)$ by
		\[
		\bigl(B_n(a_1,\ldots,a_{n-2},h)\bigr)_{i_1,\ldots,i_{n-2},y,r}
		=\begin{cases}
			1,&y+\sum_{j=1}^{n-2}i_j=S+r,\\
			0,&\text{otherwise}.
		\end{cases}
		\]
		Here $0\le i_j<a_j$ for $1\le j\le n-2$, $0\le y<h+S$, and $0\le r<h$. The shared indices $i_j$ use the same monomial coordinates as $\Poly_n$.
	\end{definition}
	
	\begin{lemma}[three-sector polynomial degeneration]
		\label{lem:tripling-degeneration}
		Let $n\ge3$ and let $a_1,\ldots,a_{n-2},h$ be positive integers. Put $S=\sum_{j=1}^{n-2}(a_j-1)$, $e=(n-1,n)$, and $C=\Poly_n(a_1,\ldots,a_{n-2},h)$. With the aligned middle block from Definition~\ref{def:middle-polynomial-block},
		\[
		\Poly_n(a_1,\ldots,a_{n-2},3h+S)
		\rightsquigarrow C\boxplus_e B_n(a_1,\ldots,a_{n-2},h)\boxplus_e C.
		\]
		As an individual tensor, $B_n(a_1,\ldots,a_{n-2},h)$ is locally isomorphic to $\sigma C$, where $\sigma$ exchanges the last two legs.
	\end{lemma}
	\begin{proof}
		We first choose three pairs of coordinate sectors for the retained blocks.
		
		Set $B=3h+S$. Index a supported entry of $\Poly_n(a_1,\ldots,a_{n-2},B)$ by its shared input indices $i_1,\ldots,i_{n-2}$, its last-input index $y$, and its output index $z$. The defining conditions are
		\[
		0\leq i_j<a_j\; (1\leq j\leq n-2),
		\quad 0\leq y<B,
		\quad z=y+\sum_{j=1}^{n-2}i_j.
		\]
		In particular, $0\leq\sum_{j=1}^{n-2}i_j\leq S$ and $0\leq z<B+S$. Introduce three last-input lengths $b_0=b_2=h$ and $b_1=h+S$, and three output lengths $c_0=c_2=h+S$ and $c_1=h$. For $r\in\{0,1,2\}$, put
		\[
		\beta_r=\sum_{t=0}^{r-1}b_t,
		\quad \gamma_r=\sum_{t=0}^{r-1}c_t,
		\quad Y_r=\{\beta_r+i\mid 0\le i<b_r\},
		\quad Z_r=\{\gamma_r+j\mid 0\le j<c_r\}.
		\]
		The sums for $r=0$ are empty. These are disjoint partitions of $[B]_0$ and $[B+S]_0$, since $\sum_{r=0}^2b_r=B$ and $\sum_{r=0}^2c_r=B+S$. Their offsets satisfy $\beta_0=\gamma_0=0$, $\beta_1=h$, $\gamma_1=h+S$, and $\beta_2=\gamma_2=2h+S$.
		
		The support equation bounds the output index between the last-input index and that index plus $S$. With the chosen offsets, an outer input sector can reach only its matching output sector: if $y\in Y_0$, then $0\le z\le h-1+S$, so $z\in Z_0$; if $y\in Y_2$, then $2h+S\le z<B+S$, so $z\in Z_2$. Consequently, $z\in Z_1$ requires $y\in Y_1$. The middle input sector may also reach the outer output sectors, and we remove precisely these interactions by local weights.
		
		Use coordinate symbols $x^{(j)}_i$ on shared leg $j$, $x^{(n-1)}_y$ on the last input, and $x^{(n)}_z$ on the output, as in Definition~\ref{def:multilinear-notation}. Assign
		\[
		\begin{aligned}
			w(x^{(j)}_i)&=0 &&(1\le j\le n-2,\ 0\le i<a_j),\\
			w(x^{(n-1)}_y)&=\begin{cases}1,&y\in Y_1,\\0,&y\in Y_0\cup Y_2,\end{cases}
			&& (0\le y<B),\\
			w(x^{(n)}_z)&=\begin{cases}-1,&z\in Z_1,\\0,&z\in Z_0\cup Z_2,\end{cases}
			&& (0\le z<B+S).
		\end{aligned}
		\]
		On every supported term, the total weight is
		\[
		\sum_{j=1}^{n-2}w(x^{(j)}_{i_j})+w(x^{(n-1)}_y)+w(x^{(n)}_z)
		=\begin{cases}
			0,&(y,z)\in\displaystyle\bigcup_{r=0}^{2}(Y_r\times Z_r),\\
			1,&y\in Y_1,\ z\in Z_0\cup Z_2.
		\end{cases}
		\]
		Make the coordinate substitutions $x\mapsto\varepsilon^{w(x)}x$ from Remark~\ref{rem:weights}. The preceding nonnegativity calculation shows that the resulting expression is polynomial in $\varepsilon$. Its constant coefficient represents the tensor $T_0$ with formal multilinear expression
		\[
		\sum_{r=0}^{2}
		\sum_{\substack{0\le i_j<a_j\ (1\le j\le n-2)\\
				y\in Y_r,\ z\in Z_r\\z=y+\sum_{j=1}^{n-2}i_j}}
		\left(\prod_{j=1}^{n-2}x^{(j)}_{i_j}\right)x^{(n-1)}_yx^{(n)}_z.
		\]
		We now identify its three matched blocks in their original shared coordinates.
		
		In sector $r=0$, write $y=y'$ and $z=z'$. In sector $r=2$, write $y=2h+S+y'$ and $z=2h+S+z'$. In either case, the support condition becomes
		\[
		0\le y'<h,
		\quad 0\le z'<h+S,
		\quad z'=y'+\sum_{j=1}^{n-2}i_j.
		\]
		Every choice of the shared input indices and $y'\in[h]_0$ gives exactly one allowed $z'$. Thus, after these translations of the last two coordinates, each outer block equals $C$ in the original shared coordinates.
		
		In the middle sector, put $y'=y-h$ and $r=z-(h+S)$. The support equation becomes
		\[
		y'+\sum_{j=1}^{n-2}i_j=S+r,
		\qquad 0\le y'<h+S,\quad 0\le r<h.
		\]
		Thus this block is exactly $B_n(a_1,\ldots,a_{n-2},h)$ in its original shared coordinates, by Definition~\ref{def:middle-polynomial-block}. To identify its individual isomorphism type, define $\bar i_j=a_j-1-i_j$ for $1\leq j\leq n-2$. Replacing $i_j$ by $a_j-1-\bar i_j$ in the support equation gives
		\[
		y=h+r+\sum_{j=1}^{n-2}\bar i_j,
		\quad z=h+S+r.
		\]
		Conversely, choose arbitrary integers $0\leq\bar i_j<a_j$ for $1\leq j\leq n-2$ and $0\leq r<h$. These equations give $h\leq y\leq2h+S-1$ and $h+S\leq z\leq2h+S-1$, so they specify a supported entry in the middle block. Therefore the displayed parametrization describes all and only its supported entries, each with coefficient one.
		
		On the middle block alone, relabel each shared coordinate $i_j$ by $\bar i_j$, retaining the translated last-input and output coordinates. The equation becomes $y'=r+\sum_{j=1}^{n-2}\bar i_j$. This is the defining support equation of $C$ with its last two legs exchanged: $r$ is its length-$h$ input, while $y'$ is its length-$(h+S)$ output. Let $\sigma\in S_n$ denote that exchange. Thus $B_n(a_1,\ldots,a_{n-2},h)\cong\sigma C$. We use this isomorphism to evaluate the middle block as an individual tensor; the three-block partial sum retains its original aligned shared coordinates.
		
		After the indicated translations on the last two legs, the constant-coefficient tensor is
		\[
		T_0=C\boxplus_e B_n(a_1,\ldots,a_{n-2},h)\boxplus_e C,
		\qquad e=(n-1,n),
		\]
		with aligned shared coordinates. Remark~\ref{rem:weights} gives the degeneration from $\Poly_n(a_1,\ldots,a_{n-2},3h+S)$ to $T_0$. Together with the individual isomorphism of the middle block proved above, this establishes both assertions.
	\end{proof}
	
	The three blocks need not have the same value under one spectral point: the middle block exchanges the last two legs. Averaging over all leg permutations removes that distinction. The uniform entropy weight $1/3$ then yields the factor of three.
	
	\begin{lemma}[shifted tripling inequality]
		\label{lem:tripling}
		For every integer $n\ge3$, every $\phi\in\Phi_{n,\F}$, and all positive integers $a_1,\ldots,a_{n-2},h$, put $S=\sum_{j=1}^{n-2}(a_j-1)$. Then
		\[
		P_\phi(a_1,\ldots,a_{n-2},3h+S)
		\ge3P_\phi(a_1,\ldots,a_{n-2},h).
		\]
		By symmetry, the analogous inequality holds in every input coordinate, with $S$ the sum of the other input degrees.
	\end{lemma}
	\begin{proof}
		Put $B=3h+S$, $C=\Poly_n(a_1,\ldots,a_{n-2},h)$, and $e=(n-1,n)$, and let $\sigma$ exchange the last two legs. Lemma~\ref{lem:tripling-degeneration} degenerates $\Poly_n(a_1,\ldots,a_{n-2},B)$ to the aligned partial sum $C\boxplus_e B_n(a_1,\ldots,a_{n-2},h)\boxplus_e C$. For every $\pi\in S_n$, apply Corollary~\ref{cor:entropy} to this partial sum with spectral point $\phi_\pi$ and probability vector $(1/3,1/3,1/3)$, whose entropy is $\ln3$. The same degeneration lemma identifies the middle block as individually locally isomorphic to $\sigma C$, so its spectral value is $\phi_\pi(\sigma C)$. Combining the entropy inequality with degeneration monotonicity from Lemma~\ref{lem:degeneration-asymptotic} gives
		\[
		\phi_\pi\bigl(\Poly_n(a_1,\ldots,a_{n-2},B)\bigr)
		\geq3^{P_{n-1,n}^{(\phi_\pi)}}
		\phi_\pi(C)^{2/3}\phi_\pi(\sigma C)^{1/3}.
		\]
		Since $\pi\mapsto\pi\sigma$ is a bijection of $S_n$,
		\[
		\prod_{\pi\in S_n}\phi_\pi(\sigma C)
		=\prod_{\pi\in S_n}\phi_\pi(C).
		\]
		Multiplying the spectral inequalities over all $\pi\in S_n$, and then taking the power $1/(n!t_\phi)$, therefore gives
		\[
		P_\phi(a_1,\ldots,a_{n-2},B)
		\geq
		3^{\frac{1}{n!t_\phi}\sum_{\pi\in S_n}P_{n-1,n}^{(\phi_\pi)}}
		\left(\prod_{\pi\in S_n}\phi_\pi(C)\right)^{1/(n!t_\phi)}
		=3P_\phi(a_1,\ldots,a_{n-2},h),
		\]
		where the final equality uses~\eqref{eq:edge-average}. Since $B=3h+S$, this proves the shifted tripling inequality in the last input. Symmetry from Lemma~\ref{lem:profile-basic} proves the assertion for every other input coordinate.
	\end{proof}
	
	\subsection{Diagonal growth}
	\label{sec:growth}
	
	We now extend the diagonal-growth step of OpenAI~\cite[Section~5 and Lemma~5.1]{openai2026} to $n-1$ input coordinates. Discrete concavity bounds a long increment by its first step, while shifted tripling gives a lower bound on that long increment. We apply the resulting estimate along a path that increases each coordinate once. Keeping track of all $n-1$ steps gives the exponent $2(n-1)/n$. Comparing this growth with the interpolation upper bound for the profile from Lemma~\ref{lem:profile-basic} then bounds the average spectral edge weight.
	
	\begin{lemma}[diagonal growth]
		\label{lem:diagonal-growth}
		Let $n\geq 3$ be an integer. Suppose that $F:\N_{>0}^{n-1}\to\mathbb R_{>0}$ is symmetric, that $F(1,\ldots,1)=1$, and that it satisfies the concavity and shifted tripling inequalities in Lemmas~\ref{lem:concavity} and~\ref{lem:tripling}, with $F$ in place of $P_\phi$. Then, for every positive integer $a$,
		\[
		F(\underbrace{a,\ldots,a}_{n-1\text{ coordinates}})
		\geq a^{2(n-1)/n}.
		\]
	\end{lemma}
	\begin{proof}
		For every positive integer $a$, write $D_a=F(a,\ldots,a)$. We first compare $D_{a+1}$ with $D_a$. Start from $(a,\ldots,a)$ and increase its coordinates one at a time until reaching $(a+1,\ldots,a+1)$.
		
		Fix an integer $j$ with $1\leq j\leq n-1$, and consider the step that increases coordinate $j$. The earlier coordinates already equal $a+1$, and the later coordinates still equal $a$. Define
		\[
		F_j(b)=F(\underbrace{a+1,\ldots,a+1}_{j-1\text{ coordinates}},
		b,\underbrace{a,\ldots,a}_{n-1-j\text{ coordinates}})
		\quad(b\in\N_{>0}).
		\]
		The total degree in the other coordinates is
		\[
		S_j=\sum_{i=1}^{j-1}a+\sum_{i=j+1}^{n-1}(a-1)
		=(j-1)a+(n-1-j)(a-1)\leq(n-2)a.
		\]
		By symmetry, the concavity and shifted tripling inequalities apply to coordinate $j$. Thus $F_j(3a+S_j)\geq3F_j(a)$, and concavity gives
		$F_j(h+1)-F_j(h)\leq F_j(h)-F_j(h-1)$ for every integer $h\geq2$.
		Thus every forward difference starting at an integer $h\geq a$ is at most
		$F_j(a+1)-F_j(a)$. In particular,
		\begin{align*}
			2F_j(a)
			&\leq F_j(3a+S_j)-F_j(a)\\
			&=\sum_{h=a}^{3a+S_j-1}\bigl(F_j(h+1)-F_j(h)\bigr)\\
			&\leq(2a+S_j)\bigl(F_j(a+1)-F_j(a)\bigr).
		\end{align*}
		Since $2a+S_j\leq na$ and $F_j(a)>0$, division by $2a+S_j$ gives
		\[
		F_j(a+1)\geq\left(1+\frac{2}{2a+S_j}\right)F_j(a)
		\geq\left(1+\frac{2}{na}\right)F_j(a).
		\]
		The values at successive steps agree: $F_1(a)=D_a$,
		$F_j(a+1)=F_{j+1}(a)$ for $1\leq j<n-1$, and
		$F_{n-1}(a+1)=D_{a+1}$. Multiplying the $n-1$ inequalities therefore gives
		\begin{equation}
			D_{a+1}\geq\left(1+\frac{2}{na}\right)^{n-1}D_a.
			\label{eq:diagonal-step}
		\end{equation}
		This argument also applies when $a=1$, since concavity at input length $2$ compares the first two forward differences.
		
		For every positive real number $x$ and every real $r\geq 1$, Bernoulli's inequality gives $(1+x)^r\geq1+rx$. With $x=2/(na)$ and $r=n/2$, it follows that
		\[
		\left(1+\frac{2}{na}\right)^{n/2}\geq 1+\frac1a.
		\]
		Raise this inequality to the positive power $2(n-1)/n$ and combine it with~\eqref{eq:diagonal-step}. We find
		\[
		\frac{D_{a+1}}{D_a}
		\geq\left(\frac{a+1}{a}\right)^{2(n-1)/n}.
		\]
		Finally, $D_1=1$. Multiplying the last inequality for $a=1,\ldots,A-1$ gives
		\[
		F(\underbrace{A,\ldots,A}_{n-1\text{ coordinates}})=D_A\geq
		\prod_{a=1}^{A-1}\left(\frac{a+1}{a}\right)^{2(n-1)/n}
		=A^{2(n-1)/n}
		\]
		for every integer $A\geq 2$. The statement for $A=1$ follows directly from the normalization.
	\end{proof}
	
	\subsection{The complete-graph bound and its consequences}
	\label{subsec:complete-graph-bound}
	
	We apply the preceding growth estimate to the symmetric spectral profile $P_\phi$. Its lower growth rate, compared with the interpolation upper bound from Subsection~\ref{sec:profile}, constrains the average edge weight $t_\phi$. Since a complete graph contains every pair of vertices as an edge, this average determines $\phi(\GT_2(K_n))$. Spectral duality then gives the desired asymptotic rank upper bound, while a maximum cut gives the lower bound.
	\begin{theorem}[complete-graph bound]
		\label{thm:complete-graph-bound}
		For every field $\F$ and every integer $n\geq2$,
		\[
		\left\lfloor\frac{n^2}{4}\right\rfloor
		\leq\log_2\aR_{\F}\bigl(\GT_2(K_n)\bigr)
		\leq\frac{n^2}{4}.
		\]
	\end{theorem}
	\begin{proof}
		First suppose that $n\geq 3$, and fix an arbitrary spectral point $\phi\in\Phi_{n,\F}$. Let $P_\phi$ and $t_\phi$ be as in Definition~\ref{def:profile}. Lemma~\ref{lem:profile-basic} gives positivity, symmetry, and $P_\phi(1,\ldots,1)=1$. Lemmas~\ref{lem:concavity} and~\ref{lem:tripling} give the two inequalities required by Lemma~\ref{lem:diagonal-growth}. Therefore, for every positive integer $a$,
		\[
		a^{2(n-1)/n}
		\leq P_\phi(a,\ldots,a)
		\leq\bigl((n-1)(a-1)+1\bigr)^{1/t_\phi}.
		\]
		For $a\geq2$, take logarithms and divide by $\log a$. The expression
		$((n-1)(a-1)+1)/a$ tends to the positive constant $n-1$. Hence
		\[
		\frac{2(n-1)}{n}
		\leq\frac{1}{t_\phi}\,
		\frac{\log((n-1)(a-1)+1)}{\log a}
		\longrightarrow\frac{1}{t_\phi}.
		\]
		It follows that $t_\phi\leq n/(2(n-1))$. By the definition of the average spectral edge weight $t_\phi$ and Lemma~\ref{lem:edge-weights},
		\[
		\phi(\GT_2(K_n))
		=2^{\sum_{1\leq i<j\leq n}P_{i,j}^{(\phi)}}
		=2^{\binom n2\,t_\phi}
		\leq 2^{n^2/4}.
		\]
		This holds for every $\phi\in\Phi_{n,\F}$. Taking the maximum over $\Phi_{n,\F}$ and applying Theorem~\ref{thm:spectral-duality} gives $\aR_{\F}(\GT_2(K_n))\leq2^{n^2/4}$.
		
		For the lower bound, Proposition~\ref{prop:cut} gives
		\[
		\log_2\aR_{\F}(\GT_2(K_n))\geq\MaxCut(K_n)
		=\left\lfloor\frac{n^2}{4}\right\rfloor,
		\]
		where a balanced cut attains the last value. This proves both inequalities for $n\geq3$.
		
		Finally, $\GT_2(K_2)$ is the $2\times2$ identity matrix viewed as a second-order tensor. Its $N$-th Kronecker power has matrix rank $2^N$, so its asymptotic rank is $2$. Thus both bounds are equal to $1$ when $n=2$, completing the proof.
	\end{proof}
	
	The next two consequences identify the questions from~\cite{cvz2019} answered by the theorem. The edge exponent $\tau_{\F}(G)$ is defined in Definition~\ref{def:edge-exponent}.
	
	\begin{corollary}[maximum cut equality for even complete graphs]
		\label{cor:even-cliques}
		For every positive integer $m$,
		\[
		\log_2\aR_{\F}\bigl(\GT_2(K_{2m})\bigr)
		=\MaxCut(K_{2m})=m^2,
		\quad \tau_{\F}(K_{2m})=\frac{m}{2m-1}.
		\]
		Thus the equality asked for in~\cite[Question~1.3.2]{cvz2019} holds for every even complete graph.
	\end{corollary}
	\begin{proof}
		Apply Theorem~\ref{thm:complete-graph-bound} with $n=2m$. Its lower and upper bounds both equal $m^2$. A balanced cut of $K_{2m}$ has $m^2$ edges and is a maximum cut. Finally, $K_{2m}$ has $\binom{2m}{2}=m(2m-1)$ edges, so dividing the exponent $m^2$ by this number gives the formula for $\tau_{\F}(K_{2m})$.
	\end{proof}
	
	\begin{corollary}[edge exponents below two thirds]
		\label{cor:two-thirds}
		For every integer $n\geq5$,
		\[
		\tau_{\F}(K_n)\leq\frac{n}{2(n-1)}<\frac23.
		\]
		In particular, the existence question in~\cite[Question~1.3.1]{cvz2019} has an affirmative answer for every $n\geq5$. Moreover,
		$\lim_{n\to\infty}\tau_{\F}(K_n)=1/2$.
	\end{corollary}
	\begin{proof}
		Dividing the upper bound in Theorem~\ref{thm:complete-graph-bound} by $|E(K_n)|=n(n-1)/2$ gives $\tau_{\F}(K_n)\le n/(2(n-1))<2/3$ for $n\ge5$.
		For the limit, both bounds in the theorem give
		\[
		\frac{\lfloor n^2/4\rfloor}{\binom n2}
		\leq\tau_{\F}(K_n)\leq\frac{n}{2(n-1)}.
		\]
		The two bounding quantities tend to $1/2$, so $\lim_{n\to\infty}\tau_{\F}(K_n)=1/2$, proving the remaining assertion.
	\end{proof}
	
	For odd $n$, the upper bound exceeds the maximum cut lower bound by $1/4$.
	
	\section{From Complete Graphs to Weighted Graph Tensors}
	\label{sec:weighted-cp}\label{sec:cp}
	The complete-graph bound can be used even when the target graph has unequal edge weights and little symmetry. Grouping vertices of a complete graph produces quadratic inequalities for arbitrary nonnegative vertex weights. Combining these inequalities through completely positive matrices gives an upper bound for every weighted graph tensor; elementary path restrictions strengthen it further.
	
	Throughout this section, $n\ge2$, the field $\F$ is arbitrary, and $J=\mathbf1\mathbf1^{\mathsf T}$ is the all-one matrix. We first specialize the generalized covers of Definition~\ref{def:general-conic-cover} to this cost matrix and the edge-covering constraints, and then add path inequalities to their dual formulations. The weighted tensors, logarithmic edge-weight matrices, and exponents are those of Definitions~\ref{def:weighted-tensor} and~\ref{def:weighted-exponent}.
	
	\subsection{A quadratic consequence of the complete-graph bound}
	\label{subsec:quadratic-constraint}
	We now apply the complete-graph bound at a different number of vertices. The key operation is to group several vertices into one tensor leg. A simple tensor remains simple after grouping, so this operation cannot increase rank, either before or after taking Kronecker powers.
	
	\begin{lemma}[quadratic spectral constraint]
		\label{lem:quadratic-spectral-constraint}
		For every $\phi\in\Phi_{n,\F}$ and every $a=(a_1,\ldots,a_n)\in\mathbb R_{\ge0}^n$,
		\begin{equation}
			\label{eq:quadratic-spectral-constraint}
			4\sum_{1\le i<j\le n}a_i a_jP_{i,j}^{(\phi)}
			\le\left(\sum_{i=1}^n a_i\right)^2.
		\end{equation}
	\end{lemma}
	\begin{proof}
		The products $a_i a_j$ arise by counting edges between groups of vertices. First suppose that the $a_i$ are nonnegative integers, and put $A=\sum_{i=1}^n a_i$. When $A\le1$, every product $a_i a_j$ with $i\ne j$ vanishes, so the claim is immediate. Assume $A\ge2$, partition the vertices of $K_A$ into groups of sizes $a_1,\ldots,a_n$, and group the tensor legs within each nonempty group. Add a one-dimensional leg for every empty group.
		
		We remove the internal edges by coordinate substitutions in the notation of Definition~\ref{def:multilinear-notation}. For each pair of distinct groups $i,j$, collect the $a_i a_j$ binary edge labels into one integer label, using the same order at both endpoints. Put $(Q_a)_{i,j}=2^{a_i a_j}$ for $i\ne j$ and $(Q_a)_{i,i}=1$. On grouped leg $i$, place the external-edge labels first, in the order of Definition~\ref{def:weighted-tensor}. They give an integer coordinate $r$ with $0\le r<\prod_{j=1}^n(Q_a)_{i,j}=2^{a_i(A-a_i)}$. The $a_i(a_i-1)$ endpoint bits of the internal edges give a second integer coordinate $s$ with $0\le s<2^{a_i(a_i-1)}$. Thus, after a local coordinate permutation, the coordinate symbols are $X^{(i)}_{r,s}$. Make the substitutions
		\[
		X^{(i)}_{r,s}\longmapsto
		\begin{cases}
			Y^{(i)}_r,&s=0,\\
			0,&s>0,
		\end{cases}
		\qquad(1\le i\le n,\ 0\le r<2^{a_i(A-a_i)},\ 0\le s<2^{a_i(a_i-1)}),
		\]
		where the $Y^{(i)}_r$ are new output coordinate symbols. The condition $s=0$ retains precisely those terms in which every internal edge has label zero. Each assignment to the external edges has one such extension, so every surviving coefficient is one. Using the label set $\Lambda(Q_a)$ and coordinate indices $I_i^{Q_a}$ from Definition~\ref{def:weighted-tensor}, the resulting expression is
		\[
		\sum_{\lambda\in\Lambda(Q_a)}\prod_{i=1}^n
		Y^{(i)}_{I_i^{Q_a}(\lambda)}.
		\]
		This is the formal multilinear expression of $\WT(Q_a)$, so the grouped tensor restricts to $\WT(Q_a)$. Grouping and restriction commute with Kronecker powers, and hence
		\[
		\log_2\aR_{\F}(\WT(Q_a))
		\le\log_2\aR_{\F}(\GT_2(K_A))
		\le\frac{A^2}{4}.
		\]
		The last step is Theorem~\ref{thm:complete-graph-bound}. The logarithmic edge weights of $Q_a$ are $a_i a_j$, so Lemma~\ref{lem:weighted-spectral-formula} gives
		\[
		\sum_{1\le i<j\le n}a_i a_jP_{i,j}^{(\phi)}
		\le\log_2\aR_{\F}(\WT(Q_a))
		\le\frac{A^2}{4}.
		\]
		
		This proves~\eqref{eq:quadratic-spectral-constraint} for integer vectors. For nonnegative rational $a_i$, multiply the vector by a common positive denominator and divide the resulting inequality by its square. Finally, approximate an arbitrary nonnegative real vector by nonnegative rational vectors and use continuity. This proves~\eqref{eq:quadratic-spectral-constraint} for all $a\in\mathbb R_{\ge0}^n$.
	\end{proof}
	
	The complete-graph bound thus controls the off-diagonal edge weights supplied by every nonnegative outer product $aa^{\mathsf T}$. The completely positive cone from Definition~\ref{def:matrix-cones} records nonnegative combinations of these rank-one matrices. For the spectral edge-weight matrix $P^{(\phi)}$, equation~\eqref{eq:quadratic-spectral-constraint} is precisely
	\begin{equation}
		\label{eq:spectral-copositivity}
		J-2P^{(\phi)}\in\COP_n.
	\end{equation}
	Indeed, $a^{\mathsf T}(J-2P^{(\phi)})a=(\sum_{i=1}^n a_i)^2-4\sum_{1\le i<j\le n}a_i a_jP_{i,j}^{(\phi)}$.
	
	\subsection{Completely positive covers}
	\label{subsec:cp-covers}
	The preceding quadratic inequality controls edge weights of the form $a_i a_j$ with $a\ge0$. A completely positive matrix sums such contributions and therefore supplies capacities on all vertex pairs. We require these capacities to cover the prescribed logarithmic edge weights, while charging for the entire matrix, including its diagonal.
	
	\begin{definition}[completely positive covering bound]
		\label{def:cp-cover}
		For a symmetric nonnegative matrix $W$ with zero diagonal, define
		\[
		U_{\CP}(W)=\frac14\min_{\substack{Z\in\CP_n\\
				Z_{i,j}\ge W_{i,j}\ (1\le i<j\le n)}}\bigl(\langle J,Z\rangle\bigr).
		\]
		For a graph $G$ with adjacency matrix $A_G$, write $U_{\CP}(G)=U_{\CP}(A_G)$.
	\end{definition}
	
	This minimum exists. Feasibility follows by taking the outer products of the factors
	$\sqrt{W_{i,j}}(\mathbf e_i+\mathbf e_j)$ over all pairs $1\le i<j\le n$ with positive weight. Completely positive matrices have nonnegative entries, so a bound on $\langle J,Z\rangle$ bounds every entry. A minimizing sequence therefore lies in a compact set, and closedness of $\CP_n$ supplies a minimizer. Writing $Z=\sum_{\nu=1}^t b^{(\nu)}(b^{(\nu)})^{\mathsf T}$ with $b^{(\nu)}=(b_1^{(\nu)},\ldots,b_n^{(\nu)})\in\mathbb R_{\ge0}^n$, the covering problem equivalently minimizes
	\begin{equation}
		\label{eq:cp-factor-cost}
		\frac14\sum_{\nu=1}^t\left(\sum_{i=1}^n b_i^{(\nu)}\right)^2
		\quad\text{subject to}\quad
		\sum_{\nu=1}^t b_i^{(\nu)}b_j^{(\nu)}\ge W_{i,j}\quad(1\le i<j\le n),
	\end{equation}
	with every $b_i^{(\nu)}\ge0$.
	
	Each nonnegative factor has a concrete tensor interpretation. For integer coordinates $b_i$, the grouping construction in Lemma~\ref{lem:quadratic-spectral-constraint} gives a weighted graph tensor with logarithmic edge weights $b_i b_j$ and logarithmic asymptotic rank at most $(\sum_{i=1}^n b_i)^2/4$. The corresponding quadratic spectral inequality extends to real factors by rational scaling and continuity. Thus complete positivity records exactly the nonnegative factors to which the complete-graph argument applies.
	
	\begin{proposition}[covering and copositive formulations]
		\label{prop:cp-duality}
		For every symmetric nonnegative zero-diagonal $W$,
		\begin{equation}
			\label{eq:cp-dual}
			U_{\CP}(W)=\max_{\substack{P=P^{\mathsf T},\ P_{i,i}=0,\ P_{i,j}\ge0\\
					J-2P\in\COP_n}}\left(\sum_{1\le i<j\le n}W_{i,j}P_{i,j}\right).
		\end{equation}
		Here the entrywise conditions range over $1\le i,j\le n$. Both optima are attained. Every feasible $P$ has $0\le P_{i,j}\le1$.
	\end{proposition}
	\begin{proof}
		Specialize the dual in Definition~\ref{def:general-conic-cover} to $\mathcal K=\CP_n$, with the all-one cost matrix and no additional linear constraints. Since $\CP_n^*=\COP_n$, this gives the maximization problem in~\eqref{eq:cp-dual}. We verify that its optimum equals the covering minimum and that both are attained.
		
		There is no duality gap. For example, set $P_{i,i}=0$ for $1\le i\le n$ and $P_{i,j}=\varepsilon$ for $i\ne j$, where $0<\varepsilon<1/2$. Then
		\[
		J-2P=(1-2\varepsilon)J+2\varepsilon I_n
		\]
		is positive definite and hence belongs to the interior of $\COP_n$, and all nonnegativity constraints on the multipliers are strict. After introducing slack variables $s_{i,j}=Z_{i,j}-W_{i,j}\ge0$ for $1\le i<j\le n$, the primal is a conic program over $\CP_n\times\mathbb R_{\ge0}^{\binom n2}$. The chosen $P$ gives an interior dual slack in both factors. Since the primal is feasible with finite value, the dual strict-feasibility case of Theorem~\ref{thm:conic-strong-duality} gives equality of the two values. Primal attainment was proved above. For dual attainment, evaluating copositivity on $\mathbf e_i+\mathbf e_j$ yields $4-4P_{i,j}\ge0$. Consequently every feasible entry lies in $[0,1]$, and the dual feasible set is a closed subset of $[0,1]^{\binom n2}$. It is nonempty because $P=0$ is feasible. Compactness now gives a dual maximizer, proving both attainment statements and~\eqref{eq:cp-dual}.
	\end{proof}
	
	The inequality $\Gamma_{\F}(W)\le U_{\CP}(W)$ can already be read from either formulation. By~\eqref{eq:spectral-copositivity}, every spectral matrix is feasible in~\eqref{eq:cp-dual}. Alternatively, given a feasible factorization $Z=\sum_{\nu=1}^t b^{(\nu)}(b^{(\nu)})^{\mathsf T}$, equation~\eqref{eq:quadratic-spectral-constraint} gives, for every spectral point,
	\begin{align}
		\sum_{1\le i<j\le n}W_{i,j}P_{i,j}^{(\phi)}
		&\le\sum_{\nu=1}^t\sum_{1\le i<j\le n}b_i^{(\nu)}b_j^{(\nu)}P_{i,j}^{(\phi)}\notag\\
		&\le\frac14\sum_{\nu=1}^t\left(\sum_{i=1}^n b_i^{(\nu)}\right)^2
		=\frac14\langle J,Z\rangle.
		\label{eq:cp-tensor-certificate}
	\end{align}
	Maximizing the left side over spectral points gives $\Gamma_{\F}(W)$ by Definition~\ref{def:weighted-exponent}. Thus every feasible completely positive factorization certifies $\Gamma_{\F}(W)\le\langle J,Z\rangle/4$. Minimizing this cost over all feasible covers gives $\Gamma_{\F}(W)\le U_{\CP}(W)$.
	
	\subsection{Path inequalities and the weighted graph theorem}
	\label{subsec:path-constraints}
	The quadratic constraints do not exhaust the relations between edge tensors. A pairwise resource can also be passed through an intermediate vertex. Keeping this elementary restriction gives a useful strengthening, especially for odd cycles.
	
	\begin{lemma}[spectral triangle inequalities]
		\label{lem:spectral-triangle}
		For every three distinct vertices $i,j,k$, every positive integer $q$, and every $\phi\in\Phi_{n,\F}$,
		\[
		E_{i,j}^{(q)}\preceq E_{i,k}^{(q)}\Kprod E_{k,j}^{(q)},
		\qquad
		P_{i,j}^{(\phi)}\le P_{i,k}^{(\phi)}+P_{k,j}^{(\phi)}.
		\]
	\end{lemma}
	\begin{proof}
		Put $T=E_{i,k}^{(q)}\Kprod E_{k,j}^{(q)}$. The label $a$ on the first edge occurs on legs $i,k$, and the label $b$ on the second occurs on legs $k,j$. Thus leg $k$ has dimension $q^2$, with coordinate $qa+b$ recording the pair $(a,b)$, while legs $i,j$ each have dimension $q$ and all other legs have dimension one. In the formal multilinear notation of Definition~\ref{def:multilinear-notation}, $T$ is represented by
		\[
		\left(\sum_{a=0}^{q-1}\sum_{b=0}^{q-1}
		x_a^{(i)}x_{qa+b}^{(k)}x_b^{(j)}\right)
		\prod_{\substack{1\le v\le n\\v\notin\{i,j,k\}}}x_0^{(v)}.
		\]
		To match the two edge labels, make the coordinate substitutions
		\[
		x_{qa+b}^{(k)}\longmapsto\delta_{a,b}y_0^{(k)}
		\qquad(0\le a,b<q),
		\]
		leaving every other coordinate symbol unchanged. Here $y_0^{(k)}$ is the sole output symbol on leg $k$, and $\delta_{a,b}$ equals one when $a=b$ and zero otherwise. The resulting expression is
		\[
		\left(\sum_{a=0}^{q-1}x_a^{(i)}x_a^{(j)}\right)y_0^{(k)}
		\prod_{\substack{1\le v\le n\\v\notin\{i,j,k\}}}x_0^{(v)}.
		\]
		Each label $a$ contributes once, with coefficient one, and the factor $y_0^{(k)}$ records the one-dimensional output at the middle vertex. This is the formal multilinear expression of $E_{i,j}^{(q)}$ from Definition~\ref{def:graph-tensor}. The substitutions therefore prove $E_{i,j}^{(q)}\preceq T$ over every field.
		
		Restriction monotonicity and multiplicativity now give
		\[
		\phi(E_{i,j}^{(2)})\le
		\phi(E_{i,k}^{(2)})\phi(E_{k,j}^{(2)}).
		\]
		Taking base-two logarithms and using Definition~\ref{def:spectral-edge-weights} proves $P_{i,j}^{(\phi)}\le P_{i,k}^{(\phi)}+P_{k,j}^{(\phi)}$, completing both assertions.
	\end{proof}
	
	\begin{definition}[path-strengthened completely positive bound]
		\label{def:path-cp-bound}
		For a symmetric nonnegative zero-diagonal weight matrix $W$, define
		\begin{equation}
			\label{eq:path-cp-bound}
			U^{\Delta}_{\CP}(W)=
			\max_{\substack{P=P^{\mathsf T},\ P_{i,i}=0,\ P_{i,j}\ge0\\
					J-2P\in\COP_n\\
					P_{i,j}\le P_{i,k}+P_{k,j}\ (i,j,k\text{ distinct})}}
			\left(\sum_{1\le i<j\le n}W_{i,j}P_{i,j}\right).
		\end{equation}
		All indices in the constraints range from $1$ to $n$. Write $U^{\Delta}_{\CP}(G)=U^{\Delta}_{\CP}(A_G)$ for a graph $G$.
	\end{definition}
	
	The feasible set is a closed subset of the compact feasible set in Proposition~\ref{prop:cp-duality}, so the maximum exists. The variables belong to every vertex pair, including nonedges: a zero demand $W_{i,j}=0$ leaves $P_{i,j}$ available in the path inequalities.
	
	\begin{theorem}[weighted graph tensor bound]
		\label{thm:weighted-cp}
		Let $\F$ be any field and let $Q$ be a dimension matrix from Definition~\ref{def:weighted-tensor}. Let $W$ be its logarithmic edge-weight matrix. Then
		\begin{equation}
			\label{eq:weighted-cp-main}
			\MaxCut(W)
			\le\log_2\aR_{\F}(\WT(Q))
			\le U^{\Delta}_{\CP}(W).
		\end{equation}
		More generally, the same inequalities hold with $\Gamma_{\F}(W)$ in place of the logarithmic rank for every real symmetric nonnegative zero-diagonal matrix $W$, and
		\[
		U^{\Delta}_{\CP}(W)\le U_{\CP}(W)\le\sum_{1\le i<j\le n}W_{i,j}.
		\]
		In particular, every finite simple graph satisfies
		\begin{equation}
			\label{eq:graph-cp-main}
			\MaxCut(G)\le\log_2\aR_{\F}(\GT_2(G))
			\le U^{\Delta}_{\CP}(G).
		\end{equation}
	\end{theorem}
	\begin{proof}
		We first bound $\log_2\aR_{\F}(\WT(Q))$ from below by $\MaxCut(W)$. Proposition~\ref{prop:cut} shows that flattening across a cut of weight $\operatorname{cut}_W(S)$ gives
		\[
		\log_2\aR_{\F}(\WT(Q))\ge\operatorname{cut}_W(S).
		\]
		Maximizing over the nonempty proper vertex sets $S$ gives the lower bound $\MaxCut(W)$. By Lemma~\ref{lem:weighted-spectral-formula}, this is also a lower bound for $\Gamma_{\F}(W)$ when $W$ is the logarithmic edge-weight matrix of $Q$.
		
		To prove $\Gamma_{\F}(W)\le U^{\Delta}_{\CP}(W)$, let $W$ be any real nonnegative zero-diagonal weight matrix and fix $\phi\in\Phi_{n,\F}$. We show that its spectral edge-weight matrix $P^{(\phi)}$ satisfies every constraint in the program defining $U^{\Delta}_{\CP}(W)$. It is symmetric and zero-diagonal by definition, and its entries are nonnegative by Lemma~\ref{lem:edge-weights}. The quadratic constraint of Lemma~\ref{lem:quadratic-spectral-constraint} gives
		\[
		a^{\mathsf T}(J-2P^{(\phi)})a
		=\left(\sum_{i=1}^n a_i\right)^2
		-4\sum_{1\le i<j\le n}a_i a_jP_{i,j}^{(\phi)}\ge0
		\quad(a\in\mathbb R_{\ge0}^n).
		\]
		Thus $J-2P^{(\phi)}$ is copositive. Lemma~\ref{lem:spectral-triangle} supplies the remaining path constraints,
		$P_{i,j}^{(\phi)}\le P_{i,k}^{(\phi)}+P_{k,j}^{(\phi)}$
		for distinct $i,j,k$. Every spectral matrix is therefore feasible in~\eqref{eq:path-cp-bound}, so
		\[
		\Gamma_{\F}(W)
		=\max_{\phi\in\Phi_{n,\F}}
		\left(\sum_{1\le i<j\le n}W_{i,j}P_{i,j}^{(\phi)}\right)
		\le U^{\Delta}_{\CP}(W).
		\]
		For logarithmic edge weights, Lemma~\ref{lem:weighted-spectral-formula} identifies the left side with $\log_2\aR_{\F}(\WT(Q))$. Together with the cut lower bound, this proves~\eqref{eq:weighted-cp-main}.
		
		Removing the path constraints enlarges the feasible set to the copositive program in Proposition~\ref{prop:cp-duality}, so $U^{\Delta}_{\CP}(W)\le U_{\CP}(W)$. Every feasible matrix in that program has $P_{i,j}\le1$, by the same proposition. Since $W$ is nonnegative, its objective is at most $\sum_{1\le i<j\le n}W_{i,j}$, proving the remaining comparison $U_{\CP}(W)\le\sum_{1\le i<j\le n}W_{i,j}$.
		
		It remains to extend the cut lower bound to arbitrary real weights. For a positive integer $s$, put $W^{(s)}_{i,j}=\lfloor sW_{i,j}\rfloor/s$. The matrix $sW^{(s)}$ is the logarithmic edge-weight matrix of a dimension matrix with entries $2^{\lfloor sW_{i,j}\rfloor}$. Applying the lower bound already proved and dividing by $s$ gives
		\[
		\MaxCut(W^{(s)})\le\Gamma_{\F}(W^{(s)}),
		\]
		using positive homogeneity from Lemma~\ref{lem:weighted-spectral-formula}. As $s\to\infty$, that lemma gives $\Gamma_{\F}(W^{(s)})\to\Gamma_{\F}(W)$. The maximum cut weights also converge, since they are maxima of finitely many linear functions of the edge weights. This proves $\MaxCut(W)\le\Gamma_{\F}(W)$ for every such $W$. Finally, take $Q_{i,j}=2$ when $(i,j)$ is an edge of $G$ and $Q_{i,j}=1$ otherwise. Then $W=A_G$ and $\WT(Q)=\GT_2(G)$, which gives~\eqref{eq:graph-cp-main}.
	\end{proof}
	
	A useful special case of~\eqref{eq:cp-tensor-certificate} uses a single nonnegative vector:
	\[
	\Gamma_{\F}(W)\le\frac14\left(\sum_{i=1}^n b_i\right)^2
	\quad\text{whenever}\quad b_i b_j\ge W_{i,j}\ \ (1\le i<j\le n).
	\]
	Several of the exact graph evaluations in Section~\ref{sec:graph-applications} will follow from one such vector and a matching cut. For other graphs the path constraints carry additional information. Before turning to these applications, we compare the conic bounds with the familiar semidefinite relaxation for maximum cut.

	\subsection{Comparison with semidefinite bounds for maximum cut}
	\label{subsec:sdp-comparison}
	Replacing the completely positive cone in the covering problem by the positive semidefinite cone gives a form of the Goemans--Williamson relaxation. This comparison explains both the value $n^2/4$ on complete graphs and the role of nonnegative factors in the tensor argument.
	
	\begin{definition}[semidefinite covering bound]
		\label{def:sdp-bound}
		For a nonnegative edge-weight matrix $W$, define
		\[
		U_{\SDP}(W)=\frac14\min_{\substack{Z\in\mathbb S_+^n\\
				Z_{i,j}\ge W_{i,j}\ (1\le i<j\le n)}}\bigl(\langle J,Z\rangle\bigr).
		\]
		For a graph $G$, write $U_{\SDP}(G)=U_{\SDP}(A_G)$.
	\end{definition}
	
	The minimum is attained: the covering constraints and positive semidefiniteness make every entry of a feasible matrix nonnegative, so a bounded objective bounds all its entries. Closedness then gives compact sublevel sets.
	
	Because $W$ is nonnegative, every feasible cover has nonnegative off-diagonal entries; its diagonal entries are nonnegative by positive semidefiniteness. Thus the only difference between the two covering problems is the existence of a Gram factorization with nonnegative coordinates. Every completely positive cover is a semidefinite cover. Conversely, the two feasible sets agree for $n\le4$ by the theorem recalled in Definition~\ref{def:matrix-cones}, but can differ for $n\ge5$.
	
	\begin{proposition}[SDP covering and the Goemans--Williamson relaxation]
		\label{prop:sdp-gw}
		For every nonnegative edge-weight matrix $W$,
		\begin{align}
			U_{\SDP}(W)
			&=\max_{\substack{X\succeq0\\X_{i,i}=1}}
			\left(\frac12\sum_{1\le i<j\le n}W_{i,j}(1-X_{i,j})\right) \label{eq:sdp-gw}\\
			&=\max_{\substack{P=P^{\mathsf T},\ P_{i,i}=0\\ J-2P\succeq0}}
			\left(\sum_{1\le i<j\le n}W_{i,j}P_{i,j}\right).\notag
		\end{align}
		In particular,
		\[
		\MaxCut(W)\leq U_{\SDP}(W)\leq U_{\CP}(W),
		\qquad
		U_{\SDP}(W)=U_{\CP}(W)\quad(n\leq4).
		\]
	\end{proposition}
	\begin{proof}
		Specialize the dual in Definition~\ref{def:general-conic-cover} to $\mathcal K=\mathbb S_+^n$, with the all-one cost matrix and no additional linear constraints. Its objective is $\sum_{1\le i<j\le n}W_{i,j}P_{i,j}$, maximized over symmetric zero-diagonal matrices $P\ge0$ with $J-2P\succeq0$.
		The primal has a strictly feasible point $Z=tJ+sI_n$ with $s>0$ and
		$t>\max_{1\le i<j\le n}(W_{i,j})$, so Theorem~\ref{thm:conic-strong-duality} identifies its value $U_{\SDP}(W)$ with this dual maximum. The constraint $P_{i,j}\ge0$ can be omitted: the matrix $X=J-2P$ has diagonal entries one, and positive semidefiniteness of each two-by-two principal submatrix gives $|X_{i,j}|\le1$. Hence $P_{i,j}=(1-X_{i,j})/2\ge0$. Conversely, every positive semidefinite $X$ with diagonal one gives such a feasible $P$. Under this change of variables the objective becomes
		\[
		\sum_{1\le i<j\le n}W_{i,j}P_{i,j}
		=\frac12\sum_{1\le i<j\le n}W_{i,j}(1-X_{i,j}),
		\]
		which proves both formulations in~\eqref{eq:sdp-gw}.
		
		For a cut with sign vector $s\in\{-1,1\}^n$, the matrix $X=ss^{\mathsf T}$ is feasible and its objective is the cut weight. Maximizing over cuts gives $\MaxCut(W)\le U_{\SDP}(W)$. Every completely positive cover is a semidefinite cover, so minimizing over these covers gives $U_{\SDP}(W)\le U_{\CP}(W)$. For $n\le4$, the two covering feasible sets coincide, as noted after Definition~\ref{def:sdp-bound}; their minima therefore satisfy $U_{\SDP}(W)=U_{\CP}(W)$.
	\end{proof}
	
	Equation~\eqref{eq:sdp-gw} is the Goemans--Williamson relaxation~\cite{gw1995}.
	If $X_{i,j}=\langle u_i,u_j\rangle$ for unit vectors $u_i$, a random
	hyperplane separates $u_i$ and $u_j$ with probability
	$\arccos(X_{i,j})/\pi$.  Consequently,
	\begin{equation}
		\label{eq:gw-sandwich}
		\begin{gathered}
			\alpha_{\mathrm{GW}}U_{\SDP}(W)
			\leq\MaxCut(W)\leq U_{\SDP}(W),\\
			\alpha_{\mathrm{GW}}=
			\min_{0<\theta\leq\pi}\left(\frac{2\theta}{\pi(1-\cos\theta)}\right)
			=0.878567\ldots.
		\end{gathered}
	\end{equation}
	For rational weights, the relaxation can be evaluated to prescribed
	additive accuracy in polynomial time.  Together with rounding, this
	gives an $(\alpha_{\mathrm{GW}}-\varepsilon)$-approximation for each fixed
	$\varepsilon>0$.
	
	Khot, Kindler, Mossel, and O'Donnell~\cite{kkmo2007} originally established
	the optimality of this worst-case approximation threshold assuming the
	Unique Games Conjecture.  OpenAI's Unique Games
	Theorem~\cite{openaiugc2026} supplies that hypothesis. Combined with the reduction,
	it rules out a deterministic $(\alpha_{\mathrm{GW}}+\varepsilon)$ approximation
	for maximum cut in polynomial time, for any fixed $\varepsilon>0$, unless
	$\mathrm P=\mathrm{NP}$.
	
	In contrast, completely positive optimization is not a semidefinite
	program: even membership in $\CP_n$ is NP-hard when $n$ varies, as are
	weak membership problems for $\CP_n$ and $\COP_n$~\cite{dickinsongijben2014}.
	For our specific covering objective, the following reduction from CLIQUE establishes the corresponding hardness statement.
	
	\begin{proposition}[complexity of the completely positive bound]
		\label{prop:cp-hardness}
		When the number of vertices is part of the input, evaluating
		$U_{\CP}(W)$ to a prescribed rational additive accuracy is NP-hard under
		polynomial-time oracle reductions.  Here the input consists of a
		nonnegative rational matrix $W$ and a positive rational error tolerance,
		both encoded in binary.
	\end{proposition}
	\begin{proof}
		We first show that an additive-value oracle for $U_{\CP}$ gives weak membership for the dual feasible set in Proposition~\ref{prop:cp-duality}.
		
		Write $m=\binom n2$ and identify a zero-diagonal symmetric matrix with its
		$m$ upper-triangular entries.  The feasible set
		\[
		K=\{P\geq0\mid J-2P\in\COP_n\}
		\]
		is convex, closed, and downward closed.  Moreover,
		$[0,1/2]^m\subseteq K\subseteq[0,1]^m$: the first inclusion follows from
		entrywise nonnegativity of $J-2P$, and the second by testing copositivity
		on $\mathbf e_i+\mathbf e_j$.  Thus the ball of radius $1/4$ about
		$c=(1/4,\ldots,1/4)$ lies in $K$, and $K$ lies in the ball of radius $n$
		about $c$.  For a rational objective $v$, write $h_K(v)=\max_{p\in K}(v\cdot p)$ and $(v_+)_{i,j}=\max(v_{i,j},0)$ for each pair $1\le i<j\le n$. Identify $v_+$ with its symmetric zero-diagonal matrix. Then
		\[
		h_K(v)=h_K(v_+)=U_{\CP}(v_+).
		\]
		Indeed, coordinates with negative coefficients can be set to zero
		without leaving $K$.  An additive-value oracle therefore supplies weak
		validity for this explicitly centered convex body.  The polynomial
		oracle equivalence of weak validity and weak membership
		\cite[Chapter~4]{gls1988} supplies weak membership for $K$.
		
		We show that even this membership problem decides CLIQUE with a
		polynomially bounded gap.  Given a graph with adjacency matrix $A$ and
		$2\leq k<n$, put
		\[
		t=\frac12\left(\frac{k}{k-1}+\frac{k+1}{k}\right),\qquad
		\eta=\frac1{16n^2},\qquad
		P_0=\frac t2 A+\eta(J-I_n).
		\]
		If the clique number is at most $k$, the Motzkin--Straus theorem
		\cite{motzkinstraus1965} gives, for every $x\geq0$ with
		$\mathbf1^{\mathsf T}x=1$,
		\[
		x^{\mathsf T}(J-2P_0)x
		\geq 1-t\frac{k-1}{k}-2\eta
		=\frac1{2k^2}-2\eta\geq\frac3{8n^2}.
		\]
		If there is a clique of size $k+1$, its uniform probability vector
		instead gives
		\[
		x^{\mathsf T}(J-2P_0)x
		\leq1-t\frac{k}{k+1}
		=-\frac1{2(k^2-1)}<-\frac1{2n^2}.
		\]
		Let $\Delta P$ be a symmetric zero-diagonal perturbation and write
		$\|\Delta P\|_{\mathrm{coord}}=(\sum_{1\le i<j\le n}(\Delta P_{i,j})^2)^{1/2}$
		for the Euclidean norm of its $m$ upper-triangular coordinates. If
		$\|\Delta P\|_{\mathrm{coord}}\le\eta/2$, then either quadratic evaluation changes by at most $\eta$, since
		$4\left|\sum_{1\le i<j\le n}\Delta P_{i,j}x_i x_j\right|
		\leq4\|\Delta P\|_{\mathrm{coord}}\sum_{1\le i<j\le n}x_i x_j\leq\eta$. In the first case the
		whole Euclidean ball of radius $\eta/2$ about $P_0$ is therefore in $K$, including its
		nonnegativity constraints; in the second case that ball is disjoint
		from $K$.  Weak membership with tolerance $\eta/4$ distinguishes the
		cases.  The centers, radii, tolerances, and all rational coefficients
		have polynomial encoding length, so the oracle reduction is polynomial. Hence an oracle for the stated additive evaluation problem would decide CLIQUE, proving the claimed NP-hardness.
	\end{proof}

	\begin{remark}[the role of nonnegative factors]
		For $2\le n\le4$, Theorem~\ref{thm:weighted-cp} and Proposition~\ref{prop:sdp-gw} therefore imply
		\[
		\Gamma_{\F}(W)\le U_{\SDP}(W).
		\]
		The extension to larger matrices is the first open problem at the end of the paper. The nonnegative coordinates in a CP factorization have an essential role in the present proof: they become numbers of vertices in the complete graph used for grouping. A semidefinite factorization may contain negative coordinates, for which that construction has no corresponding vertex groups.
	\end{remark}
	
	\subsection{Path inequalities and the value of the comparison}
	\label{subsec:sdp-triangles}
	The preceding comparison concerns the basic covering problems. Our tensor theorem also uses path inequalities, so we now add the same constraints to the Goemans--Williamson formulation. This lets us compare the strengthened bounds directly and determine whether their agreement on at most four vertices survives the added constraints.
	
	\begin{definition}[path-strengthened semidefinite bound]
		\label{def:sdp-triangle}
		Define
		\[
		U^{\Delta}_{\SDP}(W)=
		\max_{\substack{P=P^{\mathsf T},\ P_{i,i}=0,\ J-2P\succeq0\\
				P_{i,j}\le P_{i,k}+P_{k,j}\ (i,j,k\text{ distinct})}}
		\left(\sum_{1\le i<j\le n}W_{i,j}P_{i,j}\right),
		\]
		where all indices range from $1$ to $n$.
	\end{definition}
	
	The feasible set is nonempty because $P=0$ is feasible. It is closed and bounded, since $J-2P$ is a correlation matrix, so the maximum is attained.
	
	Cut matrices satisfy the path inequalities, and every positive semidefinite matrix is copositive. Hence
	\begin{equation}
		\label{eq:sdp-triangle-comparison}
		\MaxCut(W)\le U^{\Delta}_{\SDP}(W)
		\le\min\bigl(U_{\SDP}(W),U^{\Delta}_{\CP}(W)\bigr).
	\end{equation}
	The equality $U_{\CP}(W)=U_{\SDP}(W)$ for $n\le4$ extends to the path-strengthened bounds. The next proof establishes this by a geometric construction.
	
	\begin{proposition}[equality of the strengthened bounds in small order]
		\label{prop:small-order-path-equality}
		For every nonnegative zero-diagonal weight matrix $W$ of order $2\le n\le4$,
		\[
		U^{\Delta}_{\CP}(W)=U^{\Delta}_{\SDP}(W).
		\]
		Consequently, for every field $\F$ and every weighted graph tensor on between two and four vertices, with logarithmic edge-weight matrix $W$,
		\[
		\log_2\aR_{\F}(\WT(Q))\le U^{\Delta}_{\SDP}(W).
		\]
	\end{proposition}
	\begin{proof}
		By~\eqref{eq:sdp-triangle-comparison}, it suffices to prove
		$U^{\Delta}_{\CP}(W)\le U^{\Delta}_{\SDP}(W)$. Fix a matrix $P$ feasible for $U^{\Delta}_{\CP}(W)$. We will construct a matrix $P'$ feasible for $U^{\Delta}_{\SDP}(W)$ with $P'_{i,j}\ge P_{i,j}$ for $i\ne j$.
		
		We first represent the entries of $P$ as squared Euclidean distances. A nonnegative symmetric zero-diagonal matrix satisfying the path inequalities is a \emph{semimetric}. For a vertex subset $S$, its cut semimetric has entries $(\mathbf1_S(i)-\mathbf1_S(j))^2$, where $\mathbf1_S(i)$ is one if $i\in S$ and zero otherwise. Every semimetric on at most four points is a nonnegative combination of cut semimetrics~\cite{dezalaurent1997}. To make this step explicit, suppose $n=4$ and put
		\[
		\begin{gathered}
			s_1=P_{1,2}+P_{3,4},\quad s_2=P_{1,3}+P_{2,4},\quad
			s_3=P_{1,4}+P_{2,3},\quad M=\max(s_1,s_2,s_3),\\
			\alpha_i=\frac12\left(\sum_{\substack{1\le j\le4\\j\ne i}}P_{i,j}-M\right)
			\quad(1\le i\le4),\qquad
			\beta_h=\frac12(M-s_h)\quad(1\le h\le3).
		\end{gathered}
		\]
		These will be the coefficients of the four singleton cuts and the three cuts with two vertices on each side, respectively. The common value $M$ makes all three coefficients $\beta_h$ nonnegative. It also makes the $\alpha_i$ nonnegative: for a fixed vertex $i$, each $s_h$ pairs an edge incident with $i$ with an edge between two other vertices, and the path inequality through $i$ bounds the latter edge by the two remaining edges incident with $i$. Thus
		\[
		M\le\sum_{\substack{1\le j\le4\\j\ne i}}P_{i,j}.
		\]
		Assign coefficient $\alpha_i$ to the cut $\{i\}$ and coefficients $\beta_1,\beta_2,\beta_3$ to the cuts $\{1,2\}$, $\{1,3\}$, $\{1,4\}$. Precisely four of these cuts separate vertices $1$ and $2$, and their total coefficient is
		\[
		\begin{aligned}
			\alpha_1+\alpha_2+\beta_2+\beta_3
			&=\frac12\bigl(2P_{1,2}+P_{1,3}+P_{1,4}
			+P_{2,3}+P_{2,4}-s_2-s_3\bigr)\\
			&=P_{1,2}.
		\end{aligned}
		\]
		Relabeling gives the same identity for every other pair. If $n<4$, first duplicate vertices at distance zero and then restrict the resulting cut decomposition to the original vertices. Discarding cuts whose restriction is empty or full, we obtain in all cases nonnegative coefficients $\lambda_S$ such that
		\[
		P_{i,j}=\sum_{\varnothing\ne S\subsetneq\{1,\ldots,n\}}
		\lambda_S\bigl(\mathbf1_S(i)-\mathbf1_S(j)\bigr)^2.
		\]
		Consequently, the vectors $y_i$ with coordinates
		$2\sqrt{\lambda_S}\mathbf1_S(i)$ satisfy
		$\|y_i-y_j\|^2=4P_{i,j}$.
		
		Next we use copositivity to control the size of this configuration. For every probability vector $a\ge0$ with $\sum_{i=1}^n a_i=1$, the weighted variance identity and $a^{\mathsf T}(J-2P)a\ge0$ give
		\[
		\sum_{i=1}^n a_i\|y_i\|^2
		-\left\|\sum_{i=1}^n a_i y_i\right\|^2
		=\sum_{1\le i<j\le n}a_i a_j\|y_i-y_j\|^2
		=4\sum_{1\le i<j\le n}a_i a_jP_{i,j}\le1.
		\]
		The largest such variance is the squared radius of the smallest enclosing ball. Indeed, if $c$ and $R$ are its center and radius, every variance is at most
		$\sum_{i=1}^n a_i\|y_i-c\|^2\le R^2$.
		For the reverse inequality, $c$ lies in the convex hull of the points at distance $R$ from $c$. Otherwise a separating direction would move $c$ closer to every boundary point; a sufficiently small movement would also leave every other point at distance less than $R$, contradicting minimality of the ball. Choose a probability vector supported on the boundary points with mean $c$. Its variance is exactly $R^2$. The variance bound above therefore gives $R\le1$, so $\|y_i-c\|\le1$ for every $i$.
		
		It remains to put all points on the unit sphere while preserving the path inequalities. Define
		\[
		d_i=\frac{1-\|y_i-c\|^2}{4}\ge0,
		\qquad
		u_i=(y_i-c,\,2\sqrt{d_i}\,\mathbf e_i).
		\]
		The extra coordinate direction for vertex $i$ is orthogonal to those for every other vertex. Thus the $u_i$ are unit vectors, and for $i\ne j$,
		\[
		\|u_i-u_j\|^2
		=\|y_i-y_j\|^2+4d_i+4d_j
		=4(P_{i,j}+d_i+d_j).
		\]
		Let $P'_{i,j}=(1-\langle u_i,u_j\rangle)/2$ for all $i,j$. Then $P'$ is symmetric with zero diagonal, and $J-2P'$ is the Gram matrix of the $u_i$, so it is positive semidefinite. The distance formula shows that
		$P'_{i,j}=P_{i,j}+d_i+d_j\ge P_{i,j}$ for $i\ne j$. Moreover, for distinct $i,j,k$,
		\[
		P'_{i,k}+P'_{k,j}-P'_{i,j}
		=P_{i,k}+P_{k,j}-P_{i,j}+2d_k\ge0.
		\]
		Thus $P'$ is feasible for $U^{\Delta}_{\SDP}(W)$. Since $W$ is nonnegative, its objective is at least that of $P$. Maximizing over feasible $P$ gives $U^{\Delta}_{\CP}(W)\le U^{\Delta}_{\SDP}(W)$. The reverse inequality is~\eqref{eq:sdp-triangle-comparison}, so
		\[
		U^{\Delta}_{\CP}(W)=U^{\Delta}_{\SDP}(W).
		\]
		For the logarithmic edge-weight matrix of a weighted graph tensor, Theorem~\ref{thm:weighted-cp} therefore gives
		\[
		\log_2\aR_{\F}(\WT(Q))
		\le U^{\Delta}_{\CP}(W)=U^{\Delta}_{\SDP}(W),
		\]
		proving the tensor bound as well.
	\end{proof}

	\begin{proposition}[complete graphs and the five-cycle]
		\label{prop:cone-comparison-examples}
		For every $n\geq2$,
		\[
		U^{\Delta}_{\SDP}(K_n)=U_{\SDP}(K_n)
		=U^{\Delta}_{\CP}(K_n)=U_{\CP}(K_n)=\frac{n^2}{4}.
		\]
		For the five-cycle,
		\begin{align*}
			U^{\Delta}_{\SDP}(C_5)=U^{\Delta}_{\CP}(C_5)&=\frac{25}{6},\\
			U_{\SDP}(C_5)&=\frac52\left(1+\cos\frac\pi5\right),\\
			U_{\CP}(C_5)&=5.
		\end{align*}
		Here a graph argument denotes its zero-one edge-weight matrix.
	\end{proposition}
	\begin{proof}
		For $K_n$, the cover $Z=J$ gives $U_{\CP}(K_n)\leq n^2/4$.
		The symmetric zero-diagonal matrix with $P_{i,j}=n/(2(n-1))$ for $i\ne j$ satisfies all path
		inequalities, and
		\[
		J-2P=\frac{n}{n-1}I_n-\frac{1}{n-1}J\succeq0.
		\]
		Its objective gives $U^{\Delta}_{\SDP}(K_n)\ge n^2/4$. By Proposition~\ref{prop:sdp-gw} and~\eqref{eq:sdp-triangle-comparison}, both $U_{\SDP}(K_n)$ and $U^{\Delta}_{\CP}(K_n)$ lie between $U^{\Delta}_{\SDP}(K_n)$ and $U_{\CP}(K_n)$. The cover $Z=J$ therefore forces all four values to equal $n^2/4$.
		
		The two path-strengthened values on $C_5$ are the $m=5$ case of Theorem~\ref{thm:odd-cycle-cp}: it proves $U^{\Delta}_{\CP}(C_5)\le25/6$ and $U^{\Delta}_{\SDP}(C_5)\ge25/6$, which force equality by~\eqref{eq:sdp-triangle-comparison}. We compute the unstrengthened values here to show the effect of the path constraints.
		
		Let $A$ be the adjacency matrix of $C_5$.  Its least eigenvalue is
		$-2\cos(\pi/5)$.  For $X\succeq0$ with $X_{i,i}=1$,
		\[
		\frac12\sum_{\substack{1\le i<j\le5\\ (i,j)\in E(C_5)}}(1-X_{i,j})
		=\frac52-\frac14\langle A,X\rangle
		\leq\frac52\left(1+\cos\frac\pi5\right).
		\]
		Equality is attained by the Gram matrix of the five planar unit vectors
		with angles $4\pi j/5$, $0\leq j\le4$: vectors at adjacent vertices have inner product $-\cos(\pi/5)$. Consequently,~\eqref{eq:sdp-gw} gives
		\[
		U_{\SDP}(C_5)=\frac52\left(1+\cos\frac\pi5\right).
		\]
		
		Finally, $P=A$ is feasible for the unstrengthened $\CP$ dual.
		Indeed, the triangle-free case of the Motzkin--Straus theorem~\cite{motzkinstraus1965}
		gives
		$x^{\mathsf T}Ax\leq(\mathbf1^{\mathsf T}x)^2/2$ for $x\geq0$;
		hence $J-2A$ is copositive.  This yields $U_{\CP}(C_5)\geq5$.
		Theorem~\ref{thm:weighted-cp} gives $U_{\CP}(C_5)\le|E(C_5)|=5$. Together with the feasible dual point $P=A$, this proves $U_{\CP}(C_5)=5$.
	\end{proof}
	
	\begin{remark}[the complete-graph value]
		The complete-graph value has a useful geometric interpretation. In the
		GW formulation,
		\[
		\sum_{1\le i<j\le n}\frac{1-\langle u_i,u_j\rangle}{2}
		=\frac14\left(n^2-\left\|\sum_{i=1}^n u_i\right\|^2\right).
		\]
		Unit vectors can sum to zero for every $n\geq2$, while $n$ signs cannot
		sum to zero when $n$ is odd.  Thus $n^2/4$ is the continuous balancing
		value, and $\lfloor n^2/4\rfloor$ is its discrete cut counterpart.  The value $9/4$ at the triangle is therefore a natural continuous balancing parameter.
	\end{remark}
	
	The five-cycle shows a different phenomenon: path constraints recover
	information that the basic GW relaxation does not impose.  In particular,
	$U^{\Delta}_{\CP}(C_5)<U_{\SDP}(C_5)<U_{\CP}(C_5)$.
	The reverse strict comparison also occurs. If $G$ is triangle-free,
	set $P_{i,i}=0$ for $1\le i\le n$, and for distinct $i,j$ set $P_{i,j}=2/3$ on edges and $P_{i,j}=1/3$ on nonedges. Path inequalities
	hold, and Motzkin--Straus gives
	\[
	J-2P=\tfrac13(J-2A_G)+\tfrac23I_n\in\COP_n,
	\qquad U^{\Delta}_{\CP}(G)\geq\tfrac23|E(G)|.
	\]
	For a $d$-regular graph with least adjacency eigenvalue $\lambda_{\min}$,
	$\langle A_G,X\rangle\geq n\lambda_{\min}$ for every feasible GW matrix,
	so
	\[
	U_{\SDP}(G)\leq\frac{|E(G)|}{2}
	\left(1-\frac{\lambda_{\min}}d\right).
	\]
	The nonbipartite Lubotzky--Phillips--Sarnak graphs of degree $38$ have
	arbitrarily large girth and $\lambda_{\min}\geq-2\sqrt{37}$~\cite[Theorems~3.4 and~4.1]{lps1988}.
	Choosing girth greater than three yields
	$U_{\SDP}(G)\leq(1/2+\sqrt{37}/38)|E(G)|<2|E(G)|/3$.
	Thus $U^{\Delta}_{\CP}$ and the basic $U_{\SDP}$ are incomparable. The effect of adding path constraints is distinct from the choice of matrix cone.
	The next section evaluates the completely positive tensor bound on graph families with explicit conic certificates.
	
	\section{Consequences for Graph Families and Matrix Multiplication}
	\label{sec:graph-applications}\label{sec:applications}
	
	We now use Theorem~\ref{thm:weighted-cp} to obtain explicit bounds for matrix multiplication and graph families. The examples illustrate two ways to bound the graph exponent: completely positive covers certify upper bounds on $U_{\CP}$, while path inequalities can give a smaller upper bound through $U^\Delta_{\CP}$. We first recover the rectangular matrix multiplication bounds through a three-vertex cover, then treat bipartite graphs, odd cycles, complete multipartite graphs, split graphs, and wheels.
	
	Throughout this section, $\F$ is arbitrary. We use the weighted graph exponent $\Gamma_{\F}(W)$ and the graph exponent $\Gamma_{\F}(G)$ from Definition~\ref{def:weighted-exponent}. We write $U_{\CP}(G)$ and $U^{\Delta}_{\CP}(G)$ for the corresponding parameters of the adjacency matrix of $G$, and use the same convention for $U_{\SDP}$ and $U^{\Delta}_{\SDP}$.
	
	\subsection{Rectangular matrix multiplication as a three-vertex case}
	\label{subsec:rectangular-cp}
	
	Weighted triangles return us to matrix multiplication. They show how the general theorem recovers the known rectangular exponent bound through an optimal conic certificate.
	
	\begin{definition}[rectangular matrix multiplication exponent]
		\label{def:rectangular-exponent}
		For $a,b,c\geq0$, the \emph{rectangular matrix multiplication exponent}
		$\omega_{\F}(a,b,c)$ is the infimum of the real numbers $\gamma$ such that the rank over $\F$ of the matrix multiplication tensor with index dimensions
		$\lceil N^a\rceil,\lceil N^b\rceil,\lceil N^c\rceil$ is
		$O_{\varepsilon}(N^{\gamma+\varepsilon})$ for every $\varepsilon>0$~\cite[Section~1]{abz2026}.
		In particular, $\omega_{\F}(1,1,1)=\omega_{\F}$ is the square matrix multiplication exponent. The \emph{dual matrix multiplication exponent} is
		\[
		\alpha_{\F}=\sup_{\substack{k\ge0\\\omega_{\F}(1,k,1)=2}}(k).
		\]
	\end{definition}
	
	The standard identification of matrix multiplication with a three-vertex graph tensor~\cite{bcs1997,cvz2019}, together with Lemma~\ref{lem:weighted-spectral-formula}, gives
	\[
	\omega_{\F}(a,b,c)=\Gamma_{\F}(W),
	\qquad W_{1,2}=a,\quad W_{1,3}=b,\quad W_{2,3}=c.
	\]
	Here $W$ is symmetric with zero diagonal. For integer parameters, the tensor at $N=2^t$ is exactly the $t$-th Kronecker power of the weighted triangle with edge dimensions $2^a,2^b,2^c$. Projection and padding between consecutive powers of two give the same exponent for arbitrary $N$. Clearing denominators proves the rational case. For real parameters, approximate each coordinate from above and below by rational numbers. Both exponents are monotone, and increasing the three coordinates by $\varepsilon_1,\varepsilon_2,\varepsilon_3\ge0$ increases either exponent by at most $\varepsilon_1+\varepsilon_2+\varepsilon_3$: the additional edge resources have rank at most
	$\prod_{i=1}^{3}\lceil N^{\varepsilon_i}\rceil=N^{\varepsilon_1+\varepsilon_2+\varepsilon_3+o(1)}$.
	This proves the real-parameter identity as well.
	
	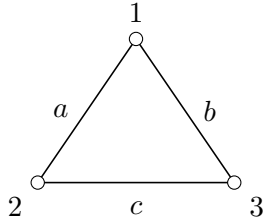
\begin{figure}[H]
		\centering
		\begin{tikzpicture}[every node/.style={font=\normalsize},
			vertex/.style={circle,draw=black,fill=white,inner sep=0pt,minimum size=5pt}]
			\coordinate (v1) at (0,1.15);
			\coordinate (v2) at (-1.3,-.75);
			\coordinate (v3) at (1.3,-.75);
			\draw[semithick] (v1)--node[midway,left=3pt] {$a$}(v2);
			\draw[semithick] (v1)--node[midway,right=3pt] {$b$}(v3);
			\draw[semithick] (v2)--node[midway,below=3pt] {$c$}(v3);
			\node[vertex,label=above:$1$] at (v1) {};
			\node[vertex,label=below left:$2$] at (v2) {};
			\node[vertex,label=below right:$3$] at (v3) {};
		\end{tikzpicture}
		\caption{The weighted triangle for rectangular matrix multiplication. Edge labels are entries of the logarithmic edge-weight matrix; the corresponding graph exponent is $\omega_{\F}(a,b,c)$.}
		\label{fig:weighted-triangle}
	\end{figure}
	
	\begin{proposition}[the rectangular bound of Uzna\'nski and Amanov--van den Berg--Zuiddam]
		\label{prop:rectangular-cp}
		Let $a\geq b\geq c\geq0$. Define
		\[
		\Psi(a,b,c)=
		\begin{cases}
			a+b,&c(a+b)\leq ab,\\[1mm]
			\displaystyle\frac{(ab+ac+bc)^2}{4abc},&c(a+b)>ab.
		\end{cases}
		\]
		The second branch has positive denominator. For the corresponding weighted triangle,
		\[
		U^{\Delta}_{\SDP}(W)
		=U_{\SDP}(W)
		=U^{\Delta}_{\CP}(W)
		=U_{\CP}(W)=\Psi(a,b,c).
		\]
		Consequently $\omega_{\F}(a,b,c)\leq\Psi(a,b,c)$, with equality whenever $c(a+b)\leq ab$.
	\end{proposition}
	\begin{proof}
		If $a+b=0$, all parameters vanish and the claim is immediate. Suppose first that $a+b>0$ and $c(a+b)\leq ab$. The nonnegative vector
		\[
		x=\left(\sqrt{a+b},\frac{a}{\sqrt{a+b}},\frac{b}{\sqrt{a+b}}\right)
		\]
		satisfies $x_1x_2=a$, $x_1x_3=b$, and $x_2x_3=ab/(a+b)\geq c$. Its completely positive covering cost is $(x_1+x_2+x_3)^2/4=a+b$. The cut separating vertex $1$ from vertices $2,3$ has weight $a+b$. The cut gives a lower bound of $a+b$ for each conic value, while the cover gives $U_{\CP}(W)\le a+b$. The comparisons in Proposition~\ref{prop:sdp-gw} and~\eqref{eq:sdp-triangle-comparison} therefore force all four values to equal $a+b$. Theorem~\ref{thm:weighted-cp} then gives $\omega_{\F}(a,b,c)=a+b$.
		
		Suppose next that $c(a+b)>ab$. Set
		\[
		x_1=\sqrt{ab/c},\qquad x_2=\sqrt{ac/b},\qquad x_3=\sqrt{bc/a}.
		\]
		Then $x_ix_j=W_{i,j}$ for $1\le i<j\le3$, and
		\[
		\frac14(x_1+x_2+x_3)^2=\frac{(ab+ac+bc)^2}{4abc}.
		\]
		Thus $U_{\CP}(W)\le\Psi(a,b,c)$. To prove a matching lower bound for $U^\Delta_{\SDP}(W)$, observe that $x_1$ is the largest coordinate and $x_1<x_2+x_3$. Hence there are unit vectors $u_1,u_2,u_3$ in the plane with $\sum_{i=1}^3 x_i u_i=0$: take three directed sides of a triangle with side lengths $x_1,x_2,x_3$. Define $P_{i,j}=(1-\langle u_i,u_j\rangle)/2$. The matrix $J-2P$ is their Gram matrix, and
		\[
		\sum_{1\le i<j\le3}W_{i,j}P_{i,j}
		=\frac14\left(\left(\sum_{i=1}^3 x_i\right)^2-
		\left\|\sum_{i=1}^3 x_iu_i\right\|^2\right)
		=\Psi(a,b,c).
		\]
		It remains to check the path inequalities. Writing $x=x_1$, $y=x_2$, and $z=x_3$, the closed-triangle identity gives
		\[
		P_{1,2}=\frac{(x+y)^2-z^2}{4xy},
		\qquad
		4xyz(P_{1,2}-P_{1,3}-P_{2,3})
		=(x-y-z)(x-y+z)(x+y+z)\leq0.
		\]
		The strict triangle inequalities for $x,y,z$ make the first factor negative and the other two positive. Permuting $x,y,z$ proves the other two path inequalities in the same way. Thus $P$ is feasible for $U^\Delta_{\SDP}(W)$. The comparison~\eqref{eq:sdp-triangle-comparison} and the basic SDP--CP comparison in Proposition~\ref{prop:sdp-gw} place all four values between this feasible value and the CP cover, so they all equal $\Psi(a,b,c)$. Applying Theorem~\ref{thm:weighted-cp} and the identity $\omega_{\F}(a,b,c)=\Gamma_{\F}(W)$ gives $\omega_{\F}(a,b,c)\le\Psi(a,b,c)$.
	\end{proof}
	
	Uzna\'nski~\cite[Theorem~1.1]{uznanski2026} obtained the exponent bound in Proposition~\ref{prop:rectangular-cp} over fields of characteristic zero. Amanov, van den Berg, and Zuiddam~\cite[Theorem~1.2]{abz2026} established the same bound over every field. Expanding the second branch gives
	\[
	\Psi(a,b,c)=\frac{a+b+c}{2}+\frac14\left(\frac{ab}{c}+\frac{ac}{b}+\frac{bc}{a}\right).
	\]
	The contribution here is the direct derivation from the weighted graph theorem and its interpretation as the three-vertex CP covering problem. In particular, their bounds
	\[
	\omega_{\F}(1,k,1)=2\quad(0\leq k\leq\tfrac12),
	\qquad
	\omega_{\F}(1,k,1)\leq1+k+\frac1{4k}\quad(k\geq\tfrac12)
	\]
	follow by substituting edge weights $1,k,1$ into Proposition~\ref{prop:rectangular-cp}, after permuting the three vertices so that the weights are in nonincreasing order~\cite[Corollary~4.17]{abz2026}.
	
	\subsection{Bipartite graphs}
	\label{subsec:bipartite-cp}
	The simplest exact graph exponents arise when one cut contains every positive-weight edge. The maximum cut and total-weight terms in Theorem~\ref{thm:weighted-cp} then agree, determining the graph exponent and both CP values.
	
	\begin{definition}[bipartite support]
		\label{def:bipartite-support}
		A nonnegative symmetric zero-diagonal $n\times n$ weight matrix $W$ has \emph{bipartite support} if its vertices can be partitioned as $L\sqcup R$ so that $W_{i,j}=0$ whenever $i,j$ belong to the same part. A graph is bipartite when its adjacency matrix has bipartite support.
	\end{definition}
	
	\begin{proposition}[bipartite equality]
		\label{prop:bipartite-weighted}
		If $W$ has bipartite support, then
		\[
		\Gamma_{\F}(W)=U^{\Delta}_{\CP}(W)
		=U_{\CP}(W)=\MaxCut(W)=\sum_{1\le i<j\le n}W_{i,j}.
		\]
		In particular, $\Gamma_{\F}(G)=|E(G)|$ for every bipartite graph $G$.
	\end{proposition}
	\begin{proof}
		A support bipartition cuts every positive-weight edge, so $\MaxCut(W)=\sum_{1\le i<j\le n}W_{i,j}$; this identity also holds when $W=0$. Theorem~\ref{thm:weighted-cp} gives
		\[
		\MaxCut(W)\le\Gamma_{\F}(W)
		\le U^\Delta_{\CP}(W)\le U_{\CP}(W)
		\le\sum_{1\le i<j\le n}W_{i,j}=\MaxCut(W).
		\]
		Thus all these quantities are equal. Taking $W=A_G$ gives $\Gamma_{\F}(G)=|E(G)|$.
	\end{proof}

	\subsection{Odd cycles and the role of path inequalities}
	\label{subsec:odd-cycle-cp}
	An odd cycle has no cut containing every edge. Here the path inequalities improve the basic CP bound and bring the upper bound on the graph exponent close to the maximum cut lower bound.
	
	\begin{definition}[cycles]
		\label{def:cycle-family}
		For $m\geq3$, let $C_m$ have vertex set $[m]$ and edges $(i,i+1)$ for $1\le i<m$, together with $(1,m)$. In formulas involving cycle vertices, indices are interpreted cyclically. Even cycles are bipartite, so we concentrate on odd $m$.
	\end{definition}
	
	\begin{figure}[H]
		\centering
		\begin{tikzpicture}[every node/.style={font=\normalsize},
			vertex/.style={circle,draw=black,fill=white,inner sep=0pt,minimum size=5pt}]
			\foreach \i in {1,...,5}\coordinate (v\i) at ({90+72*(\i-1)}:1.1);
			\draw[semithick] (v1)--(v2)--(v3)--(v4)--(v5)--cycle;
			\foreach \i in {1,...,5}{
				\pgfmathsetmacro{\ang}{90+72*(\i-1)}
				\node[vertex,label={\ang:$\i$}] at (v\i) {};
			}
		\end{tikzpicture}
		\caption{The odd cycle $C_5$. Its five edges form a cycle, with no chords.}
		\label{fig:odd-cycle}
	\end{figure}
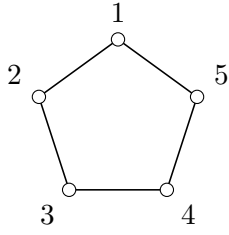

	\begin{theorem}[odd-cycle bound]
		\label{thm:odd-cycle-cp}
		For every odd integer $m\geq3$,
		\[
		m-1\leq \Gamma_{\F}(C_m)
		\leq U^{\Delta}_{\CP}(C_m)
		=U^{\Delta}_{\SDP}(C_m)
		=\frac{m^2}{m+1}
		=m-1+\frac1{m+1}.
		\]
	\end{theorem}
	\begin{proof}
		We first bound $U^{\Delta}_{\CP}(C_m)$ from above by comparing the cycle-edge objective with the total weight on all vertex pairs. The quadratic constraint controls this total, while the path inequalities force a certain amount of weight onto the chords.
		
		Let $P$ be feasible for $U^{\Delta}_{\CP}(C_m)$, and put
		\[
		A=\sum_{i=1}^mP_{i,i+1},\qquad
		B=\sum_{\substack{1\le i<j\le m\\(i,j)\notin E(C_m)}}P_{i,j}.
		\]
		Thus $A$ is the objective and $A+B$ is the total weight on all pairs. Testing $J-2P\in\COP_m$ on $\mathbf1$ gives
		\begin{equation}
			A+B\le\frac{m^2}{4}.
			\label{eq:cycle-total-budget}
		\end{equation}
		To relate $B$ to $A$, write $m=2k+1$. From each vertex $v$, pair its nonneighbors in the order
		$(v+2,v+3),\ldots,(v+2k-2,v+2k-1)$. Each pair is a cycle edge, so
		\[
		P_{v+2r,v+2r+1}\le P_{v,v+2r}+P_{v,v+2r+1}
		\qquad(1\le r\le k-1).
		\]
		Summing over $v$ and $r$, every cycle edge occurs $k-1$ times on the left, and every chord occurs twice on the right, once from each endpoint. Consequently,
		\[
		\frac{m-3}{2}A\le2B,
		\qquad
		\frac{m+1}{4}A\le A+B\le\frac{m^2}{4}.
		\]
		This proves $U^{\Delta}_{\CP}(C_m)\le m^2/(m+1)$. For $m=3$, there are no chords, and the same conclusion follows directly from~\eqref{eq:cycle-total-budget} with $B=0$.
		
		To prove $U^\Delta_{\SDP}(C_m)\ge m^2/(m+1)$, average the $m$ maximum cuts obtained by rotating an alternating coloring of the odd cycle. Their average already satisfies every path inequality. We will multiply it by a factor greater than one; the centered coloring vectors show how much scaling the semidefinite constraint permits. Specifically, let $s^{(0)},\ldots,s^{(m-1)}\in\{1,-1\}^m$ be the cyclic shifts of $(1,-1,\ldots,1)$, and set
		\[
		Y=\frac1m\sum_{h=0}^{m-1}s^{(h)}(s^{(h)})^{\mathsf T},
		\qquad c=\frac{m^2}{m^2-1},
		\qquad P=\frac c2(J-Y).
		\]
		Each matrix $(J-s^{(h)}(s^{(h)})^{\mathsf T})/2$ is a cut semimetric. Hence $P$ is symmetric, nonnegative, has zero diagonal, and satisfies the path inequalities. Each coloring vector has coordinate sum $1$, and their average is $\mathbf1/m$. Centering them gives
		\[
		J-2P
		=\frac cm\sum_{h=0}^{m-1}
		\left(s^{(h)}-\frac1m\mathbf1\right)
		\left(s^{(h)}-\frac1m\mathbf1\right)^{\mathsf T}\succeq0.
		\]
		Thus $P$ is feasible for $U^{\Delta}_{\SDP}(C_m)$. Every cycle edge is cut in exactly $m-1$ of the $m$ colorings, so its entry is $c(m-1)/m=m/(m+1)$. The objective is therefore $m^2/(m+1)$, proving the claimed lower bound on $U^\Delta_{\SDP}(C_m)$. Together with $U^\Delta_{\SDP}(C_m)\le U^\Delta_{\CP}(C_m)$ from~\eqref{eq:sdp-triangle-comparison} and the upper bound already proved for $U^\Delta_{\CP}(C_m)$, this establishes both conic equalities.
		
		Finally, an odd cycle has maximum cut size $m-1$: a cut cannot contain every cycle edge, and alternating colors attains $m-1$. Applying Theorem~\ref{thm:weighted-cp} and substituting the conic values just proved gives
		\[
		m-1\le\Gamma_{\F}(C_m)
		\le U^\Delta_{\CP}(C_m)=U^\Delta_{\SDP}(C_m)
		=\frac{m^2}{m+1}=m-1+\frac1{m+1}.
		\]
	\end{proof}
	
	\begin{remark}[the role of path inequalities]
		\label{rem:cycle-path-inequalities}
		The improvement over the unstrengthened CP bound comes entirely from the path inequalities. Indeed,
		\[
		U_{\CP}(C_m)=m\qquad(m\geq4).
		\]
		For $m\ge4$, the cycle is triangle-free. The Motzkin--Straus inequality~\cite{motzkinstraus1965} gives
		\[
		x^{\mathsf T}A_{C_m}x\le\frac12(\mathbf1^{\mathsf T}x)^2
		\qquad(x\in\mathbb R_{\ge0}^m).
		\]
		Thus $P=A_{C_m}$ is feasible in the copositive formulation of Proposition~\ref{prop:cp-duality}. Its objective is $m$, proving $U_{\CP}(C_m)\ge m$. The total-weight bound in Theorem~\ref{thm:weighted-cp} gives $U_{\CP}(C_m)\le m$, and hence equality.
	\end{remark}
	
	Over $\C$, Christandl and Zuiddam~\cite[Theorem~3.10]{cz2019} proved the odd-cycle bound $m(m-1)/(m-1+\alpha_{\C})$, also recorded by Christandl, Vrana, and Zuiddam~\cite[Section~1.1]{cvz2019}, where $\alpha_{\C}$ is the dual matrix multiplication exponent. Using $\alpha_{\C}\geq1/2$ from~\cite{abz2026} in that formula gives the explicit estimate $m(m-1)/(m-1/2)$. Our conic bound is strictly smaller than this estimate for every odd $m\geq3$, since
	\[
	\frac{m(m-1)}{m-1/2}-\frac{m^2}{m+1}
	=\frac{m(m-2)}{2(m-1/2)(m+1)}>0.
	\]
	For $m=3$, the direct matrix multiplication bound is already $9/4$. For $C_5$ and $C_7$, the new intervals are $[4,25/6]$ and $[6,49/8]$, respectively.

	\subsection{Complete multipartite graphs}
	\label{subsec:multipartite-cp}
	
	Complete multipartite graphs illustrate why nonuniform edge dimensions are useful even for unweighted questions. Merging each independent part into one tensor leg produces a weighted complete graph. A tensor-product construction shows that this grouping preserves the graph exponent.
	
	\begin{definition}[independent blowups and complete multipartite graphs]
		\label{def:independent-blowup}
		Let $k\geq2$, let $H$ be a graph on $[k]$, and let $s_1,\ldots,s_k$ be positive integers. The \emph{independent blowup} $H[s_1,\ldots,s_k]$ replaces vertex $i$ by an independent set $V_i$ of size $s_i$. For each pair $1\le i<j\le k$, it joins every vertex of $V_i$ to every vertex of $V_j$ precisely when $(i,j)\in E(H)$. In particular,
		\[
		K_{s_1,\ldots,s_k}=K_k[s_1,\ldots,s_k]
		\]
		is the \emph{complete multipartite graph} with part sizes $s_1,\ldots,s_k$.
	\end{definition}
	
	\begin{figure}[H]
		\centering
		\begin{tikzpicture}[every node/.style={font=\normalsize},
			vertex/.style={circle,draw=black,fill=white,inner sep=0pt,minimum size=5pt}]
			\foreach \i/\x in {1/-2.4,2/0,3/2.4}{
				\coordinate (v\i a) at (\x,.65);
				\coordinate (v\i b) at (\x,-.65);
				\draw[dashed,rounded corners=3pt] ({\x-.38},-1)--({\x+.38},-1)--({\x+.38},1)--({\x-.38},1)--cycle;
				\node at (\x,-1.25) {$V_{\i}$};
			}
			\foreach \u in {a,b}\foreach \v in {a,b}{
				\draw[semithick] (v1\u)--(v2\v);
				\draw[semithick] (v2\u)--(v3\v);
			}
			\draw[semithick] (v1a).. controls (-1.2,1.8) and (1.2,1.8) ..(v3a);
			\draw[semithick] (v1b).. controls (-1.2,-2.15) and (1.2,-2.15) ..(v3b);
			\draw[semithick] (v1a)--(v3b);
			\draw[semithick] (v1b)--(v3a);
			\foreach \i in {1,2,3}\foreach \u in {a,b}\node[vertex] at (v\i\u) {};
		\end{tikzpicture}
		\caption{The complete tripartite graph $K_{2,2,2}$. Each dashed box is an independent part; every pair of vertices in different parts is joined.}
		\label{fig:multipartite}
	\end{figure}
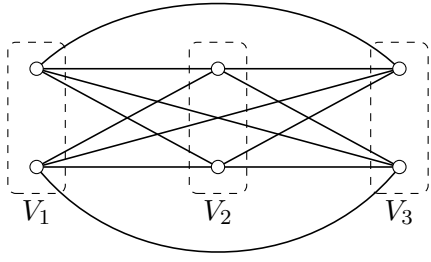
	
	\begin{proposition}[independent-blowup identity]
		\label{prop:independent-blowup}
		Let $W$ be the symmetric zero-diagonal $k\times k$ matrix whose upper-triangular entries, for $1\le i<j\le k$, are $W_{i,j}=s_i s_j$ when $(i,j)\in E(H)$ and $W_{i,j}=0$ otherwise. Then
		\[
		\Gamma_{\F}(H[s_1,\ldots,s_k])=\Gamma_{\F}(W).
		\]
		In particular, if every part has size $t$, then
		$\Gamma_{\F}(H[t,\ldots,t])=t^2\Gamma_{\F}(H)$.
	\end{proposition}
	\begin{proof}
		We first prove $\Gamma_{\F}(W)\le\Gamma_{\F}(H[s_1,\ldots,s_k])$ by grouping each independent part.
		
		Put $N=\sum_{i=1}^k s_i$, $T=\GT_2(H[s_1,\ldots,s_k])$, and let $Q_{i,j}=2^{W_{i,j}}$. Group the $s_i$ legs in $V_i$ into one leg. If $(i,j)\in E(H)$, the $s_i s_j$ binary edge labels between $V_i$ and $V_j$ form one label of dimension $2^{s_i s_j}$; otherwise the dimension is one. There are no internal edges in a part. The grouped tensor is therefore $\WT(Q)$, up to local coordinate permutations. Grouping preserves simple tensors and commutes with Kronecker powers. Applying the asymptotic rank definition and Lemma~\ref{lem:weighted-spectral-formula} yields
		\[
		\Gamma_{\F}(W)=\log_2\aR_{\F}(\WT(Q))
		\le\log_2\aR_{\F}(T)=\Gamma_{\F}(H[s_1,\ldots,s_k]).
		\]
		
		For the reverse inequality, we distribute copies of the grouped tensor over all choices of one vertex from each part. Each original edge will then receive the same total logarithmic dimension. Let $\mathcal T=V_1\times\cdots\times V_k$ and $D=|\mathcal T|=\prod_{i=1}^k s_i$. For each $(v_1,\ldots,v_k)\in\mathcal T$, place $\WT(Q)$ on the selected vertices and add one-dimensional legs at the other vertices; denote this $N$-leg tensor by $T_{v_1,\ldots,v_k}$. Adding these scalar legs preserves rank, in both directions, by adjoining a scalar factor or removing it with a local functional. The same holds for every tensor power, so $\aR_{\F}(T_{v_1,\ldots,v_k})=\aR_{\F}(\WT(Q))$.
		
		Fix an original edge between a vertex of $V_i$ and a vertex of $V_j$. Exactly $D/(s_i s_j)$ choices in $\mathcal T$ contain both endpoints. Each contributes logarithmic edge dimension $s_i s_j$, for a total of $D$. Thus edge-label pairing from Lemma~\ref{lem:graph-multiplicative} gives
		\[
		\bigKprod_{(v_1,\ldots,v_k)\in\mathcal T}T_{v_1,\ldots,v_k}\cong T^{\Kprod D}.
		\]
		Submultiplicativity of asymptotic rank across these $D$ factors, and multiplicativity on powers of the single tensor $T$, imply
		\[
		D\,\Gamma_{\F}(H[s_1,\ldots,s_k])
		=\log_2\aR_{\F}(T^{\Kprod D})
		\le\sum_{(v_1,\ldots,v_k)\in\mathcal T}\log_2\aR_{\F}(T_{v_1,\ldots,v_k})
		=D\,\Gamma_{\F}(W).
		\]
		Dividing by $D$ and combining this with the grouping inequality proves
		\[
		\Gamma_{\F}(H[s_1,\ldots,s_k])=\Gamma_{\F}(W).
		\]
		If every $s_i=t$, then $W=t^2A_H$. Positive homogeneity in Lemma~\ref{lem:weighted-spectral-formula} therefore gives
		\[
		\Gamma_{\F}(H[t,\ldots,t])
		=\Gamma_{\F}(t^2A_H)=t^2\Gamma_{\F}(H),
		\]
		proving the equal-part specialization as well.
	\end{proof}
	
	\begin{theorem}[complete multipartite bounds]
		\label{thm:multipartite-cp}
		Let $k\geq2$, let $s_1,\ldots,s_k$ be positive integers, and put
		\[
		N=\sum_{i=1}^k s_i,\qquad
		b=\max_{1\leq i\leq k}(s_i),\qquad
		\delta(s_1,\ldots,s_k)=\min_{S\subseteq[k]}
		\left(\left|N-2\sum_{\substack{1\leq i\leq k\\i\in S}}s_i\right|\right).
		\]
		For $G=K_{s_1,\ldots,s_k}$,
		\[
		\MaxCut(G)=\frac{N^2-\delta(s_1,\ldots,s_k)^2}{4}
		\leq\Gamma_{\F}(G)
		\leq U_{\CP}(G)
		=\begin{cases}
			b(N-b),&b\geq N/2,\\[1mm]
			N^2/4,&b\leq N/2.
		\end{cases}
		\]
		The graph exponent equals the maximum cut size whenever $b\geq N/2$ or $\delta(s_1,\ldots,s_k)=0$.
	\end{theorem}
	\begin{proof}
		We first compute $\MaxCut(G)$ by showing that some maximum cut keeps every part intact.
		
		For any cut, let $x_i$ be the number of vertices of $V_i$ on one side. With the other coordinates fixed, the number of crossing edges is affine in $x_i$, so replacing $x_i$ by either $0$ or $s_i$ does not decrease it. Applying this to all parts produces a maximum cut that splits no part. If the selected parts have total size $u$, the cut has
		\[
		u(N-u)=\frac{N^2-(N-2u)^2}{4}
		\]
		edges. Minimizing $|N-2u|$ over subsets of the parts gives $\MaxCut(G)=(N^2-\delta(s_1,\ldots,s_k)^2)/4$. The lower bound on $\Gamma_{\F}(G)$ follows from Theorem~\ref{thm:weighted-cp}.
		
		To evaluate $U_{\CP}(G)$, compress each part into one coordinate and let $W_{i,j}=s_i s_j$ for distinct $i,j$, with zero diagonal. We claim that $U_{\CP}(G)=U_{\CP}(W)$: both directions of compression preserve the cost of each nonnegative factor. In a CP cover of $G$, replace a factor $a$ by the vector with coordinates $\sum_{v\in V_i}a_v$. The capacity between compressed coordinates $i,j$ is the sum of capacities on all $s_i s_j$ edges between these parts, and hence is at least $s_i s_j$. The coordinate sum, and therefore the factor cost in~\eqref{eq:cp-factor-cost}, is unchanged. Conversely, in a cover of $W$, replace coordinate $i$ of every factor by $s_i$ copies of that coordinate divided by $s_i$. Each original edge receives at least unit capacity, again at the same cost. These two constructions prove
		\[
		U_{\CP}(G)=U_{\CP}(W).
		\]
		
		Suppose first that $b\ge N/2$. We cover the original graph with one factor, assigning coordinate $\sqrt{(N-b)/b}$ to vertices in a largest part and $\sqrt{b/(N-b)}$ to all other vertices. Every edge incident with the largest part receives capacity one; every other edge receives capacity $b/(N-b)\ge1$. The coordinate sum is $2\sqrt{b(N-b)}$, so this cover costs $b(N-b)$. The cut separating the largest part from the rest has exactly that many edges. Theorem~\ref{thm:weighted-cp} places $\MaxCut(G)$, $\Gamma_{\F}(G)$, and $U_{\CP}(G)$ between this cut weight and cover cost, so they all equal $b(N-b)$.
		
		Suppose now that $b\le N/2$. The factor $(s_1,\ldots,s_k)$ covers $W$ at cost $N^2/4$. To prove that this cover is optimal, we seek a feasible point of the basic SDP program whose objective also equals $N^2/4$. The Gram formulation in Proposition~\ref{prop:sdp-gw} suggests choosing unit vectors $u_i$ with zero weighted sum. Such vectors exist precisely because no part is longer than all the others combined: the lengths $s_1,\ldots,s_k$ form a closed, possibly degenerate polygon. To see the geometric condition directly, remove a longest side of length $b$. Aligning the remaining sides gives a resultant of length $N-b\ge b$; adding them successively opposite the current resultant keeps its length at most $b$. By continuity, an intermediate choice of directions gives resultant length exactly $b$. Orient the longest side oppositely to obtain $\sum_{i=1}^k s_i u_i=0$.
		
		Set $P_{i,j}=(1-\langle u_i,u_j\rangle)/2$. This is a symmetric nonnegative zero-diagonal matrix, and $J-2P$ is a Gram matrix, so $P$ is feasible for $U_{\SDP}(W)$. Its objective is
		\[
		\sum_{1\le i<j\le k}s_i s_jP_{i,j}
		=\frac14\left(N^2-\left\|\sum_{i=1}^k s_i u_i\right\|^2\right)
		=\frac{N^2}{4}.
		\]
		The feasible SDP value, Proposition~\ref{prop:sdp-gw}, the compression identity, and the cover give, respectively,
		\[
		\frac{N^2}{4}\le U_{\SDP}(W)\le U_{\CP}(W)
		=U_{\CP}(G)\le\frac{N^2}{4}.
		\]
		Thus $U_{\CP}(G)=N^2/4$ when $b\le N/2$. Together with the cut formula and Theorem~\ref{thm:weighted-cp}, this gives
		\[
		\frac{N^2-\delta(s_1,\ldots,s_k)^2}{4}
		=\MaxCut(G)\le\Gamma_{\F}(G)\le\frac{N^2}{4}.
		\]
		If $\delta(s_1,\ldots,s_k)=0$, the two endpoints coincide, so $\Gamma_{\F}(G)=\MaxCut(G)$. The case $b\ge N/2$ already gave the same equality with value $b(N-b)$, proving both exactness assertions.
	\end{proof}
	
	These families give further cases in which the maximum cut equality in Question~1.3.2 of~\cite{cvz2019} holds. For example,
	\[
	\Gamma_{\F}(K_{2,3,4,5})=49,\qquad
	\Gamma_{\F}(K_{3,3,6})=36,\qquad
	\Gamma_{\F}(K_{t,t,t,t})=4t^2.
	\]
	The first equality uses the partition $2+5=3+4$. For three parts, Proposition~\ref{prop:independent-blowup} instead gives the exact reduction
	\[
	\Gamma_{\F}(K_{a,b,c})=\omega_{\F}(ab,ac,bc),
	\qquad
	\Gamma_{\F}(K_{t,t,t})=t^2\omega_{\F}.
	\]
	The octahedral graph $K_{2,2,2}$ in Figure~\ref{fig:multipartite}, for instance, has exponent $4\omega_{\F}\leq9$.
	
	\subsection{Split graphs and exact maximum cut formulas}
	\label{subsec:split-cp}
	Complete split graphs provide further exact cases, governed by the balance between a clique and an independent part.
	
	\begin{definition}[split and complete split graphs]
		\label{def:split-graphs}
		A \emph{split graph} admits a vertex partition into a clique and an independent set. For positive integers $a,b$, the \emph{complete split graph}
		\[
		J_{a,b}=K_a\vee\overline K_b
		\]
		also contains every edge between the two parts. Here $\vee$ denotes the join of two graphs, obtained by adding all edges between their vertex sets. The book graph $B_t=J_{2,t}$ consists of $t$ triangles sharing one common edge.
	\end{definition}
	
	\begin{figure}[H]
		\centering
		\begin{tikzpicture}[every node/.style={font=\normalsize},
			vertex/.style={circle,draw=black,fill=white,inner sep=0pt,minimum size=5pt}]
			\coordinate (c1) at (-2,.8);
			\coordinate (c2) at (-2,-.8);
			\coordinate (c3) at (-.8,0);
			\foreach \i/\y in {1/1,2/0,3/-1}\coordinate (v\i) at (1.5,\y);
			\draw[dashed,rounded corners=3pt] (-2.35,-1.2) rectangle (-.4,1.2);
			\draw[dashed,rounded corners=3pt] (1.1,-1.35) rectangle (1.9,1.35);
			\foreach \i in {1,2,3}\foreach \j in {1,2,3}\draw[semithick] (c\i)--(v\j);
			\draw[thick] (c1)--(c2)--(c3)--cycle;
			\foreach \i in {1,2,3}{\node[vertex] at (c\i) {};\node[vertex] at (v\i) {};}
			\node at (-1.375,-1.6) {$K_3$};
			\node at (1.5,-1.6) {$\overline K_3$};
		\end{tikzpicture}
		\caption{The complete split graph $J_{3,3}=K_3\vee\overline K_3$. The left part is a clique, the right part is independent, and all edges between the two parts are present.}
		\label{fig:split-graph}
	\end{figure}
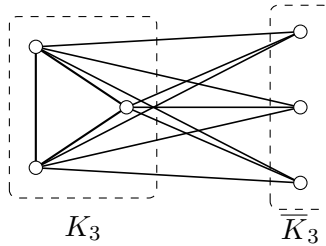
	
	\begin{theorem}[complete split graphs]
		\label{thm:split-cp}
		If $b\geq a$, then
		\[
		\Gamma_{\F}(J_{a,b})=U^{\Delta}_{\CP}(J_{a,b})
		=U_{\CP}(J_{a,b})=ab=\MaxCut(J_{a,b}).
		\]
		If $b<a$, then
		\[
		\left\lfloor\frac{(a+b)^2}{4}\right\rfloor
		\leq \Gamma_{\F}(J_{a,b})\leq\frac{(a+b)^2}{4}.
		\]
		The lower and upper bounds coincide when $a+b$ is even. More generally, for every split graph with a specified partition of sizes $a,b$, the corresponding upper bound above remains valid.
	\end{theorem}
	\begin{proof}
		Assume $b\geq a$. Define a nonnegative factor by assigning coordinate $\sqrt{b/a}$ to each clique vertex and coordinate $\sqrt{a/b}$ to each independent vertex. Its outer product has capacity $1$ on every edge between the two parts and capacity $b/a\geq1$ on every clique edge. The factor coordinates sum to $2\sqrt{ab}$, so the resulting cover costs $ab$. The cut separating the parts has $ab$ edges. Theorem~\ref{thm:weighted-cp} places the graph exponent and both CP values between this cut size and covering cost, proving the first chain of equalities.
		
		To determine the maximum cut in the other regime, start with any cut. Holding the clique vertices fixed, each independent vertex can be placed on whichever side gives it more crossing edges. All independent vertices may therefore be put on the same side without decreasing the cut size. If exactly $x$ clique vertices lie on the other side, the cut has
		\[
		x(a-x)+bx=x(a+b-x)
		\]
		edges, where $0\leq x\leq a$. When $b<a$, an integer nearest to $(a+b)/2$ belongs to this interval, so the maximum is $\lfloor(a+b)^2/4\rfloor$. Since $J_{a,b}$ is a subgraph of $K_{a+b}$, monotonicity in Lemma~\ref{lem:weighted-spectral-formula} and Theorem~\ref{thm:complete-graph-bound} give
		\[
		\Gamma_{\F}(J_{a,b})\le\Gamma_{\F}(K_{a+b})\le\frac{(a+b)^2}{4}.
		\]
		Combining this with the maximum cut formula and Theorem~\ref{thm:weighted-cp} yields
		\[
		\left\lfloor\frac{(a+b)^2}{4}\right\rfloor
		\le\Gamma_{\F}(J_{a,b})\le\frac{(a+b)^2}{4},
		\]
		whose endpoints coincide when $a+b$ is even. Finally, every split graph $G$ with the specified partition is a subgraph of $J_{a,b}$. The same monotonicity gives
		\[
		\Gamma_{\F}(G)\le\Gamma_{\F}(J_{a,b})
		\le\begin{cases}
			ab,&b\ge a,\\[1mm]
			(a+b)^2/4,&b<a,
		\end{cases}
		\]
		which proves the upper-bound assertion for arbitrary split graphs.
	\end{proof}
	
	The exact cases include all books $B_t$ with $t\geq2$, for which $\Gamma_{\F}(B_t)=2t$, and all complete split graphs of even order. For example, deleting any edge $e$ of $K_4$ gives $\Gamma_{\F}(K_4-e)=4$, and $\Gamma_{\F}(K_3\vee\overline K_3)=9$. These are nonbipartite instances of the maximum cut equality asked for in~\cite[Question~1.3.2]{cvz2019}. The restriction $t\geq2$ matters: $B_1=K_3$, whose exponent is $\omega_{\F}$.
	
	\subsection{Wheel graphs}
	\label{subsec:wheel-cp}
	
	For wheels, we combine the path-counting argument for cycles with a nonuniform choice of vertex weights. Weighting the hub differently from the rim lets the quadratic constraint balance all spokes against the cycle at once.
	
	\begin{definition}[wheels]
		\label{def:wheel-family}
		For $m\geq3$, let $W_m=K_1\vee C_m$ be the wheel with hub $0$ and $m$ rim vertices. Thus $W_m$ has $m+1$ vertices and $2m$ edges. The index $m$ counts the rim vertices, not the total number of vertices.
	\end{definition}
	
	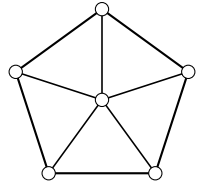
\begin{figure}[H]
		\centering
		\begin{tikzpicture}[every node/.style={font=\normalsize},
			vertex/.style={circle,draw=black,fill=white,inner sep=0pt,minimum size=5pt}]
			\coordinate (h) at (0,0);
			\foreach \i in {1,...,5}\coordinate (v\i) at ({90+72*(\i-1)}:1.2);
			\draw[thick] (v1)--(v2)--(v3)--(v4)--(v5)--cycle;
			\foreach \i in {1,...,5}\draw[semithick] (h)--(v\i);
			\foreach \i in {1,...,5}\node[vertex] at (v\i) {};
			\node[vertex] at (h) {};
		\end{tikzpicture}
		\caption{The wheel $W_5$: a five-vertex rim and one central hub joined to every rim vertex.}
		\label{fig:wheel}
	\end{figure}

	\begin{theorem}[wheel bounds]
		\label{thm:wheel-cp}
		For $m\geq3$,
		\[
		\left\lfloor\frac{3m}{2}\right\rfloor
		\leq \Gamma_{\F}(W_m)
		\leq U^{\Delta}_{\CP}(W_m)
		=U^{\Delta}_{\SDP}(W_m)
		=
		\begin{cases}
			\displaystyle\frac{25m}{16},&m\text{ even},\\[2mm]
			\displaystyle\frac{(5m+1)^2}{16(m+1)},&m\text{ odd}.
		\end{cases}
		\]
		For even $m$, the same value also equals $U_{\CP}(W_m)$.
	\end{theorem}
	\begin{proof}
		We first compute the maximum cut from a coloring of the rim. If $k$ rim vertices lie opposite the hub, the cut contains $k$ spokes and at most $2\min(k,m-k)$ rim edges. Maximizing this quantity over $0\le k\le m$ gives $\lfloor3m/2\rfloor$. Alternating colors around the rim attains it, with the larger rim color class opposite the hub when $m$ is odd. Thus $\MaxCut(W_m)=\lfloor3m/2\rfloor$, and Theorem~\ref{thm:weighted-cp} gives the claimed lower bound on the graph exponent.
		
		To prove $U_{\CP}(W_m)\le25m/16$ for even $m$, we use CP factors supported on a rim edge and the hub. Assign one such factor to each rim edge: its two endpoints have coordinate $1$, the hub has coordinate $1/2$, and all other coordinates are zero. Each rim edge receives unit capacity. Each spoke belongs to the two factors on its adjacent rim edges and receives $1/2$ from each. Every edge is therefore covered, and the total cost is
		\[
		m\cdot\frac{(1+1+1/2)^2}{4}=\frac{25m}{16}.
		\]
		Thus $U^\Delta_{\CP}(W_m)\le U_{\CP}(W_m)\le25m/16$.
		
		For odd $m$, we bound $U^\Delta_{\CP}(W_m)$ using the cycle path inequalities and choose a weight at the hub to balance the rim and spoke terms. Let $P$ be feasible for $U^{\Delta}_{\CP}(W_m)$, and let $A$ be the sum of its entries on rim edges, $B$ the sum on spokes, and $C$ the sum on all pairs of rim vertices, including chords. The objective is $A+B$. The principal submatrix of $P$ on the rim satisfies the cycle path inequalities. Applying the chord-pairing argument from the proof of Theorem~\ref{thm:odd-cycle-cp} to this submatrix gives $C\ge(m+1)A/4$. Testing copositivity on a vector with coordinate $x\ge0$ at the hub and coordinate $1$ at every rim vertex gives
		\[
		(m+1)A+4xB\le4(C+xB)\le(m+x)^2.
		\]
		Choosing $x=(m+1)/4$ makes the coefficients of $A$ and $B$ equal. We obtain
		\[
		A+B\le\frac{(m+x)^2}{m+1}
		=\frac{(5m+1)^2}{16(m+1)}.
		\]
		Taking the maximum over feasible $P$ gives the stated upper bound on $U^\Delta_{\CP}(W_m)$.
		
		We next construct feasible matrices for $U^\Delta_{\SDP}(W_m)$ with objective values equal to these estimates for $U^\Delta_{\CP}(W_m)$. Give every spoke the same entry $s$. This leaves us to choose the rim submatrix so that its path inequalities hold and the full matrix $J-2P$ is positive semidefinite. Order the vertices with the hub first and the rim vertices in cyclic order. In the block matrices below, $J$ has order $m+1$; write $\mathbf1\in\mathbb R^m$ for the all-one vector and $J_m=\mathbf1\mathbf1^{\mathsf T}$ for the all-one rim matrix.
		
		For even $m$, let $v\in\{1,-1\}^m$ alternate around the rim. Set $a=15/16$, put $s=5/8$ on every spoke, and take $a(J_m-vv^{\mathsf T})/2$ as the rim submatrix of $P$. The matrix $P$ is symmetric, nonnegative, and has zero diagonal. Positive semidefiniteness of $J-2P$ follows from the explicit decomposition
		\[
		J-2P=zz^{\mathsf T}
		+a\begin{pmatrix}0&0\\0&vv^{\mathsf T}\end{pmatrix}\succeq0,
		\qquad z=\begin{pmatrix}1\\-\tfrac14\mathbf1\end{pmatrix}\in\mathbb R^{m+1}.
		\]
		The rim submatrix is a scaled cut semimetric. For a triangle through the hub, the two spoke entries equal $s$, and the rim entry lies in $[0,a]\subseteq[0,2s]$; these are exactly the needed path inequalities. Every rim edge has entry $a$, so the objective is $m(a+s)=25m/16$. This proves $U^\Delta_{\SDP}(W_m)\ge25m/16$.
		
		For odd $m$, we use the averaged alternating cuts from the proof of Theorem~\ref{thm:odd-cycle-cp}. Let $v^{(0)},\ldots,v^{(m-1)}\in\{1,-1\}^m$ be the cyclic shifts of $(1,-1,\ldots,1)$, and define
		\[
		Y=\frac1m\sum_{h=0}^{m-1}v^{(h)}(v^{(h)})^{\mathsf T},
		\qquad Y_0=Y-\frac{1}{m^2}J_m.
		\]
		The average of the vectors $v^{(h)}$ is $\mathbf1/m$, so $Y_0$ is their covariance matrix and is positive semidefinite. Use the same $x=(m+1)/4$ as in the estimate for $U^\Delta_{\CP}(W_m)$, and put
		\[
		a=\frac{m-x^2/m}{m+1},\qquad
		c=\frac{ma}{m-1}=\frac{m^2-x^2}{m^2-1},\qquad
		s=\frac{1+x/m}{2}.
		\]
		Set each spoke entry to $s$ and the rim submatrix to $c(J_m-Y)/2$. The choice of $s$ makes every hub--rim entry of $J-2P$ equal to $-x/m$. The remaining rim block is then accounted for by the centered matrix:
		\[
		J-2P=zz^{\mathsf T}
		+c\begin{pmatrix}0&0\\0&Y_0\end{pmatrix},
		\qquad z=\begin{pmatrix}1\\-\tfrac{x}{m}\mathbf1\end{pmatrix}.
		\]
		This identity proves positive semidefiniteness. Here $1\le x\le m$, so $0\le c\le1$. The rim submatrix is a nonnegative combination of cut semimetrics, so it has zero diagonal and satisfies every rim path inequality; its entries lie in $[0,c]$. The path inequalities on triples containing the hub also hold, because the two spoke entries are equal and each rim entry is at most $1\le2s$. Thus $P$ is feasible for the strengthened SDP program. Each rim edge has entry $c(m-1)/m=a$, and consequently
		\[
		\sum_{(i,j)\in E(W_m)}P_{i,j}
		=m(a+s)=\frac{(5m+1)^2}{16(m+1)}.
		\]
		Hence $U^\Delta_{\SDP}(W_m)\ge(5m+1)^2/(16(m+1))$.
		
		For each parity, the feasible SDP objective equals the upper bound proved for $U^\Delta_{\CP}(W_m)$. Since $U^\Delta_{\SDP}(W_m)\le U^\Delta_{\CP}(W_m)$ by~\eqref{eq:sdp-triangle-comparison}, both strengthened programs have this value. Combining their equality with the maximum cut formula and Theorem~\ref{thm:weighted-cp} gives
		\[
		\left\lfloor\frac{3m}{2}\right\rfloor
		\le\Gamma_{\F}(W_m)
		\le U^\Delta_{\CP}(W_m)=U^\Delta_{\SDP}(W_m)
		=\begin{cases}
			25m/16,&m\text{ even},\\[1mm]
			(5m+1)^2/(16(m+1)),&m\text{ odd}.
		\end{cases}
		\]
		When $m$ is even, the additional comparison
		\[
		\frac{25m}{16}=U^\Delta_{\CP}(W_m)
		\le U_{\CP}(W_m)\le\frac{25m}{16}
		\]
		proves $U_{\CP}(W_m)=25m/16$ as well.
	\end{proof}
	
	For odd rims, the remaining gap has the simple form
	\[
	\frac{(5m+1)^2}{16(m+1)}-\frac{3m-1}{2}
	=\frac{(m-3)^2}{16(m+1)}.
	\]
	For example, $W_3=K_4$ has exponent $4$, while
	\[
	7\leq \Gamma_{\F}(W_5)\leq\frac{169}{24}=7+\frac1{24},
	\qquad
	10\leq \Gamma_{\F}(W_7)\leq\frac{81}{8}.
	\]

	\phantomsection
	\section*{Open Problems}
	The conic framework gives exact graph exponents for some graphs and leaves a gap above the cut bound for others. Understanding these gaps leads to two questions about the strength of the bounds, and to a broader question about the relation between conic optimization and asymptotic spectra.
	
	\textbf{A semidefinite upper bound.}
	Is it true that, for every field $\F$ and every nonnegative symmetric zero-diagonal edge-weight matrix $W$,
	\[
	\Gamma_{\F}(W)\le U_{\SDP}(W)?
	\]
	Our theorem gives $\Gamma_{\F}(W)\le U_{\CP}^{\Delta}(W)$ in every dimension. For $2\le n\le4$, Proposition~\ref{prop:small-order-path-equality} identifies the strengthened CP and SDP values. Combining this equality with~\eqref{eq:sdp-triangle-comparison} gives $\Gamma_{\F}(W)\le U^{\Delta}_{\SDP}(W)\le U_{\SDP}(W)$, so the semidefinite upper bound already holds in these dimensions.
	
	The restriction to nonnegative factors is a concrete obstacle in the present proof. For a nonnegative integer factor vector $a$, the grouping construction gives a weighted graph tensor with logarithmic edge weights $a_i a_j$. Rational scaling and continuity extend the resulting quadratic inequality to every $a\ge0$. This argument requires $a\ge0$: a negative coordinate cannot represent a negative number of vertices. A positive semidefinite matrix may require signed factors, even when its entries are nonnegative. Extending the upper bound to general semidefinite covers therefore calls for a new way to use these signed factors, or a tensor argument that avoids this grouping interpretation.
	
	Together with~\eqref{eq:gw-sandwich}, an affirmative answer would give
	\[
	\alpha_{\mathrm{GW}}U_{\SDP}(W)
	\le\MaxCut(W)\le\Gamma_{\F}(W)\le U_{\SDP}(W),
	\]
	where $\alpha_{\mathrm{GW}}\approx0.878567$ and
	$1/\alpha_{\mathrm{GW}}\approx1.138217$.
	Thus the SDP value would provide a polynomial-time upper approximation within a factor $1/\alpha_{\mathrm{GW}}$ to the weighted graph exponent $\Gamma_{\F}(W)$. Equivalently, multiplying that value by $\alpha_{\mathrm{GW}}$ gives an $\alpha_{\mathrm{GW}}$ approximation from below. For rational weights this holds with arbitrarily small numerical slack; logarithmic edge weights arising from integer dimensions can be approximated to the required precision.
	
	\textbf{Beyond the conic value for matrix multiplication.}
	Even an affirmative answer to the semidefinite question would leave the same upper bound for the triangle: both basic conic bounds and both path-strengthened bounds equal $9/4$ on $K_3$, as shown in Proposition~\ref{prop:cone-comparison-examples}. Thus $9/4$ is the exact upper bound supplied by these four conic programs. Whether this value equals $\omega_{\F}$ remains open, as does the possibility $\omega_{\F}=2$. Can one prove $\omega_{\F}>2$ over some field by a technique that goes beyond flattening ranks, or prove $\omega_{\F}<9/4$ by upper-bound techniques that go beyond the CP and SDP programs considered here?
	
	\textbf{Conic optimization across asymptotic spectra.}
	The connection with conic optimization also suggests a question about the spectral framework itself. In the graph semiring, several familiar spectral points already have conic formulations. With disjoint union as addition, strong graph product as multiplication, and the cohomomorphism preorder $G\le H$ defined by a homomorphism $\overline G\to\overline H$, both the Lov\'asz theta function and the fractional clique-cover number are spectral points~\cite{zuiddam2019}. For a graph $G$ on vertex set $\{1,\ldots,n\}$, the former has the semidefinite formulation~\cite{lovasz1979}
	\[
	\vartheta(G)=\max\left(
	\langle J,X\rangle\;\middle|\;
	X\in\mathbb S_+^n,\quad\operatorname{tr}(X)=1,\quad
	X_{i,j}=0\ ((i,j)\in E(G))
	\right).
	\]
	The latter is the value of the linear program
	\[
	\overline\chi_f(G)=\min\left(
	\sum_{C\in\mathcal K(G)}\lambda_C\;\middle|\;
	\lambda_C\ge0\ (C\in\mathcal K(G)),\quad
	\sum_{\substack{C\in\mathcal K(G)\\v\in C}}\lambda_C\ge1\ (v\in V(G))
	\right),
	\]
	where $\mathcal K(G)$ is the set of nonempty cliques. These parameters are additive and multiplicative for the graph semiring operations and monotone for its preorder. By comparison, our CP programs relax the maximization over tensor spectral edge-weight matrices in Lemma~\ref{lem:weighted-spectral-formula}. Can one relate these conic descriptions through a deeper correspondence between the graph spectrum and the spectra of weighted graph tensors, and use that correspondence to transfer spectral inequalities or construct new spectral points?
	
	\phantomsection
	\section*{Concurrent Work}
	\label{sec:concurrent-work}
	Khoruzhii~\cite{khoruzhii2026} independently obtained the complete-graph bound over $\C$ in version 1 (21 October 2026) of \emph{A Square-Root Bound for Asymptotic Tensor Rank}. Version 3 (4 November 2026) extends the bound to arbitrary fields and acknowledges the independent construction presented in the first version of this paper. Its complete-graph theorem therefore agrees with Theorem~\ref{thm:main} over every field. It also proves the general-format estimate
	\[
	\aR_{\F}(T)\le\left(\prod_{i=1}^{k}\dim V_i\right)^{1/2}
	\qquad\left(T\in\bigCprod_{i=1}^{k}V_i,\ k\ge2\right).
	\]
	Both works extend shared-leg separation to higher order. The additional contribution here is the passage from complete graphs to arbitrary weighted graph tensors: grouping gives universal quadratic spectral constraints, and completely positive covers and path inequalities turn them into the weighted bound and its graph-family applications.
	
	\phantomsection
	\section*{AI Disclosure}
	\label{sec:ai-disclosure}
	This project originated from the author's earlier encounter with graph tensors in quantum information and subsequent questions about their asymptotic rank and asymptotic subrank. The author first used GPT-6 Pro to study the existing results and guided it in exploring possible improvements. The author then proposed applying the polynomial multiplication techniques from OpenAI's matrix multiplication result to higher-order graph tensors. Following this direction, GPT-6 Pro developed the higher-order construction and the complete-graph upper bound studied in this paper. The author subsequently used GPT-6 Astra Ultra, under detailed instructions concerning mathematical statements, proof completeness, notation, organization, and writing style, to assist in preparing an early version of the manuscript.
	
	The author initially expected that the technique would be difficult to extend from complete graphs to general graphs. While reading the rectangular matrix multiplication improvements of Amanov, van den Berg, and Zuiddam, the author considered whether the method might nevertheless yield a rough bound for general weighted graph tensors. The author then guided GPT-6 Pro in exploring this extension. Unexpectedly, GPT-6 Pro obtained the general weighted result through a short reduction and derived the completely positive bound. The author next directed GPT-6 Pro to investigate arbitrary coefficient fields; it gave an affirmative argument extending the preceding results from $\mathbb C$ to every field. The author then used GPT-6 Astra Ultra, with the same detailed guidance as for the early version, to assist in preparing the present manuscript. To keep the paper focused, the author removed the Operational Barrier result from the earlier version: the completely positive formulation now gives a more informative account of the origin of the bound $n^2/4$. The author also used Google's DeepReview for pre-submission feedback and GPT-6 Astra for subsequent checks of proofs, references, and presentation. The author takes full responsibility for the correctness, completeness, and integrity of the paper.
	
	\phantomsection
	\section*{Acknowledgements}
	The author thanks Baitian Li for helpful discussions.

	\phantomsection

	\appendix
	\section{Weighted Flattening Ranks and Maximum Cuts}
	\label{app:cuts}
	\label{sec:cuts}
	
	The lower bounds in the main results come from flattening a graph tensor across a vertex cut. This appendix gives that calculation for arbitrary edge dimensions and over every field, using the coordinate conventions of Section~\ref{sec:preliminaries}. The argument is standard in the unweighted setting; see Christandl, Vrana, and Zuiddam~\cite[Definition~1.1.19, Proposition~1.1.20, and Example~1.1.21]{cvz2019}. The same calculation also gives the minimum cut upper bound on asymptotic subrank, which Appendix~\ref{sec:subrank} will show is attained.
	
	\begin{definition}[flattening rank]
		\label{def:flattening-rank}
		Let $k\ge2$ be an integer, let $T$ be a $k$-th order tensor, and choose coordinates of positive dimensions $(d_1,\ldots,d_k)$ for $T$. If $T=0$, use the zero coefficient table of these dimensions. Let $\varnothing\ne S\subsetneq\{1,\ldots,k\}$ and write $\overline S=\{1,\ldots,k\}\setminus S$. In these coordinates, the \emph{flattening across $S$} is the matrix $\operatorname{Flat}_S(T)$ with rows indexed by $\mathbf a\in\prod_{v\in S}[d_v]_0$ and columns indexed by $\mathbf b\in\prod_{v\in\overline S}[d_v]_0$, defined by
		\[
		\operatorname{Flat}_S(T)_{\mathbf a,\mathbf b}=T_{c_1,\ldots,c_k},
		\quad c_v=\begin{cases}a_v,&v\in S,\\b_v,&v\in\overline S,\end{cases}\;(1\le v\le k).
		\]
		The indices on each side are ordered by increasing vertex label and then lexicographically. The \emph{flattening rank} is $\operatorname{FR}_S(T)=\rank(\operatorname{Flat}_S(T))$, with matrix rank taken over $\F$. A different representative of a nonzero tensor permutes the rows and columns, so this rank is independent of the representative. For $T=0$, every choice gives a zero matrix, and hence $\operatorname{FR}_S(0)=0$.
	\end{definition}
	
	\begin{lemma}[properties of flattening rank]
		\label{lem:flattening-properties}
		Let $k\ge2$ be an integer, let $A,B$ be $k$-th order tensors, and let $\varnothing\ne S\subsetneq\{1,\ldots,k\}$, with complement $\overline S$. Then
		\[
		\begin{gathered}
			\Q_{\F}(A)\le \operatorname{FR}_S(A)\le\R_{\F}(A),\\
			\operatorname{FR}_S(A\Ksum B)=\operatorname{FR}_S(A)+\operatorname{FR}_S(B),\\
			\operatorname{FR}_S(A\Kprod B)=\operatorname{FR}_S(A)\operatorname{FR}_S(B).
		\end{gathered}
		\]
		Moreover, $A\preceq B$ implies $\operatorname{FR}_S(A)\le \operatorname{FR}_S(B)$, and $\operatorname{FR}_S(\Id_k(q))=q$ for every nonnegative integer $q$.
	\end{lemma}
	\begin{proof}
		The claims involving a zero tensor follow from $\R_{\F}(0)=\Q_{\F}(0)=\operatorname{FR}_S(0)=0$ and the fact that $A\preceq0$ implies $A=0$. We may therefore assume that $A$ and $B$ are nonzero and work in fixed coordinates.
		
		A simple tensor $x_1\Cprod\cdots\Cprod x_k$ flattens to the outer product of the vector $\bigKprod_{v\in S}x_v$ and the vector $\bigKprod_{v\in\overline S}x_v$. It therefore has matrix rank at most $1$. A rank decomposition of $A$ into $r$ simple tensors gives a sum of $r$ matrices of rank at most $1$, proving $\operatorname{FR}_S(A)\le\R_{\F}(A)$.
		
		Suppose that local matrices $L_1,\ldots,L_k$ map $B$ to $A$. Expanding the entries in Definition~\ref{def:restriction} gives
		\[
		\operatorname{Flat}_S(A)=
		\left(\bigKprod_{v\in S}L_v\right)
		\operatorname{Flat}_S(B)
		\left(\bigKprod_{v\in\overline S}L_v\right)^{\mathsf T}.
		\]
		Here the products are the usual Kronecker products of matrices, with their factors in increasing vertex order. Multiplication on the left or right cannot increase matrix rank, so $\operatorname{FR}_S(A)\le \operatorname{FR}_S(B)$. For $q\ge1$, the flattening of $\Id_k(q)$ has exactly one entry $1$ in each position whose row and column are constant tuples with the same value in $[q]_0$. These positions use $q$ distinct rows and $q$ distinct columns; all other entries are zero. Its rank is therefore $q$. The assertion for $q=0$ follows from the definition of $\Id_k(0)$. Applying monotonicity to $\Id_k(\Q_{\F}(A))\preceq A$ proves $\Q_{\F}(A)\le \operatorname{FR}_S(A)$.
		
		The flattening of $A\Ksum B$ has the flattenings of $A$ and $B$ in two disjoint diagonal blocks, with all other entries zero. Its rank is their rank sum.
		
		After separate permutations of its rows and columns, $\operatorname{Flat}_S(A\Kprod B)$ is the matrix Kronecker product of $\operatorname{Flat}_S(A)$ and $\operatorname{Flat}_S(B)$. To recall why matrix rank is multiplicative under this operation, invertible row and column operations put a matrix of rank $r$ into a rectangular matrix with an $r\times r$ identity block and zeros elsewhere. Applying the Kronecker products of these operations to two matrices of ranks $r$ and $s$ produces a matrix with exactly $rs$ independent nonzero rows and columns. Its rank is $rs$. With $r=\operatorname{FR}_S(A)$ and $s=\operatorname{FR}_S(B)$, this proves $\operatorname{FR}_S(A\Kprod B)=\operatorname{FR}_S(A)\operatorname{FR}_S(B)$.
	\end{proof}

	\begin{lemma}[flattening across a weighted cut]
		\label{lem:cut-flattening}
		Let $Q$ be a dimension matrix of order $n\ge2$ with logarithmic edge-weight matrix $W$. For every nonempty proper subset $S\subset\{1,\ldots,n\}$,
		\[
		\operatorname{FR}_S(\WT(Q))
		=\prod_{\substack{1\le i<j\le n\\|\{i,j\}\cap S|=1}}Q_{i,j}
		=2^{\operatorname{cut}_W(S)}.
		\]
		In particular, $\operatorname{FR}_S(\GT_q(G))=q^{|\delta_G(S)|}$ for every positive integer $q$.
	\end{lemma}
	\begin{proof}
		If exactly one endpoint of $(i,j)$ belongs to $S$, the flattening of $E^{(Q_{i,j})}_{i,j}$ is an identity matrix of size $Q_{i,j}$, up to coordinates of dimension one, and has rank $Q_{i,j}$. If both endpoints belong to $S$, the flattening is the nonzero column vector
		$\sum_{a=0}^{Q_{i,j}-1}e_{aQ_{i,j}+a}^{(Q_{i,j}^2)}$,
		with possible dimension-one coordinates. It has rank one. If neither endpoint belongs to $S$, the flattening is the corresponding row vector and again has rank one. The edge decomposition in Definition~\ref{def:weighted-tensor} and multiplicativity from Lemma~\ref{lem:flattening-properties} therefore give the displayed product. Its logarithm is precisely the cut weight. Setting $Q_{i,j}=q$ on the edges of $G$ and $Q_{i,j}=1$ on nonedges leaves one factor $q$ for each edge in $\delta_G(S)$, giving $\operatorname{FR}_S(\GT_q(G))=q^{|\delta_G(S)|}$.
	\end{proof}
	
	\begin{proposition}[weighted maximum cut lower bound]
		\label{prop:cut}
		For every field $\F$ and every dimension matrix $Q$ of order $n\ge2$ with logarithmic edge-weight matrix $W$,
		\[
		\MaxCut(W)\le\log_2\aR_{\F}(\WT(Q)).
		\]
		If the support graph of $W$ is bipartite, then
		\[
		\log_2\aR_{\F}(\WT(Q))=\MaxCut(W)
		=\sum_{1\le i<j\le n}W_{i,j}.
		\]
		Consequently, $\Gamma_{\F}(G)\ge\MaxCut(G)$ for every simple graph $G$, with equality for bipartite $G$, and
		$\Gamma_{\F}(K_n)\ge\lfloor n^2/4\rfloor$.
	\end{proposition}
	\begin{proof}
		For a nonempty proper vertex set $S$ and every positive integer $N$, the two preceding lemmas give
		\[
		\R_{\F}(\WT(Q)^{\Kprod N})
		\ge\operatorname{FR}_S(\WT(Q)^{\Kprod N})
		=2^{N\operatorname{cut}_W(S)}.
		\]
		Taking $N$-th roots, then the limit, and finally base-two logarithms gives $\log_2\aR_{\F}(\WT(Q))\ge\operatorname{cut}_W(S)$. Maximizing over $S$ therefore proves $\MaxCut(W)\le\log_2\aR_{\F}(\WT(Q))$.
		
		If the support is bipartite, its bipartition crosses every positive-weight edge, so its cut weight is $\sum_{1\le i<j\le n}W_{i,j}$. The defining expansion of $\WT(Q)$ has exactly $\prod_{1\le i<j\le n}Q_{i,j}$ simple terms, one for each edge labeling. This bounds its rank, and hence its asymptotic rank, by that product. Consequently,
		\[
		\sum_{1\le i<j\le n}W_{i,j}=\MaxCut(W)
		\le\log_2\aR_{\F}(\WT(Q))
		\le\sum_{1\le i<j\le n}W_{i,j},
		\]
		so all three quantities are equal. The same equality holds when every edge weight is zero, because the tensor is $\Id_n(1)$.
		
		For $W=A_G$, Lemma~\ref{lem:weighted-spectral-formula} identifies the logarithmic asymptotic rank with $\Gamma_{\F}(G)$. This gives $\Gamma_{\F}(G)\ge\MaxCut(G)$, with equality for bipartite $G$. Finally, a cut of $K_n$ with $s$ vertices on one side has $s(n-s)$ edges, whose maximum is $\lfloor n^2/4\rfloor$. Hence $\Gamma_{\F}(K_n)\ge\lfloor n^2/4\rfloor$, proving the last assertion.
	\end{proof}
	
	\section{Weighted Asymptotic Subrank and the Minimum Cut}
	\label{app:subrank}
	\label{sec:subrank}
	
	The flattening calculation in Appendix~\ref{sec:cuts} bounds logarithmic asymptotic subrank from above by the minimum cut. We now prove the matching lower bound for weighted graph tensors over every field. Vrana and Christandl~\cite[Theorem~1 and Section~3]{vc2017} established the minimum cut formula in the setting of hypergraph entanglement conversion, allowing parallel hyperedges and, in their final discussion, nonuniform edge dimensions. We retain their orthogonal-representation method, express its local restrictions and degenerations through coordinate substitutions over $\F$, and account for the polynomial coefficient-extraction overhead.
	
	The proof starts with binary tensors of loopless multigraphs, where parallel edges encode integer weights. Approximating logarithmic edge weights by these multiplicities in tensor powers then gives the formula for arbitrary positive integer edge dimensions. The integral vectors obtained below specify integer weights in formal coordinate substitutions, so the construction is independent of the coefficient field.
	
	For this appendix, a \emph{loopless multigraph} $H$ has a finite vertex set and a finite set of distinct edge labels, each with two different endpoints; different edge labels may have the same pair of endpoints. Degrees and cut sizes count these labels with multiplicity. The tensor $\GT_q(H)$ has one factor $E^{(q)}_{i,j}$ for every edge label with endpoints $i,j$. Thus, if there are $m_{i,j}$ such labels, this tensor is isomorphic to $\WT(Q)$ with $Q_{i,j}=q^{m_{i,j}}$. Pairing, projection, and flattening of its edge labels have exactly the properties established in Lemmas~\ref{lem:graph-multiplicative} and~\ref{lem:cut-flattening}. We fix an ordering of the edge labels whenever coordinates are used.
	
	\begin{definition}[general-position orthogonal representation]
		\label{def:orthogonal-representation}
		Let $H$ be a finite simple graph, and let $d$ be a positive integer with $d\le|V(H)|$. A \emph{general-position orthogonal representation} of $H$ in $\mathbb R^d$ is a family $(c_v)_{v\in V(H)}$ of vectors in $\mathbb R^d$ such that $\langle c_u,c_v\rangle=0$ for every two distinct nonadjacent vertices $u,v\in V(H)$, and every subfamily of at most $d$ vectors is linearly independent. Here $\langle x,y\rangle=\sum_{j=1}^{d}x_jy_j$.
		
		The \emph{line graph} $L(H)$ of a loopless multigraph $H$ is the simple graph with vertex set $E(H)$; two distinct vertices are adjacent exactly when their edge labels share at least one endpoint in $H$.
	\end{definition}
	
	The next lemma is the graph case of~\cite[Theorem~3]{vc2017}, with parallel edges allowed. Its real-vector input is the theorem of Lov\'asz, Saks, and Schrijver~\cite[Theorem~1.1]{lss1989}, with the corrected proof in~\cite[Theorem~1]{lss2000}. We include the rational approximation argument of~\cite{vc2017} to obtain the integer weights used below.
	\begin{lemma}[integral orthogonal representations]
		\label{lem:integral-orthogonal}
		Let $G$ be a finite connected loopless multigraph with at least three vertices. Put $m=|E(G)|$, $\ell=\MinCut(G)$, and $d=m-\ell$. Then $d\ge1$, and $L(G)$ has a general-position orthogonal representation $(c_e)_{e\in E(G)}$ with $c_e\in\mathbb Z^d$ for every $e\in E(G)$.
	\end{lemma}
	\begin{proof}
		If $G$ has $k\ge3$ vertices, then $\ell\le\min_{v\in V(G)}\bigl(\deg(v)\bigr)\le2m/k<m$. Also $\ell\ge1$ because $G$ is connected. Consequently $1\le d\le m-1$.
		
		Deleting fewer than $\ell$ edges leaves $G$ connected: otherwise a component of the remaining graph would give a cut of $G$ with fewer than $\ell$ edges. It follows that deleting fewer than $\ell$ vertices leaves $L(G)$ connected. Indeed, the remaining vertices are the edges of the connected remaining graph, and a path in that graph connects any two of its edges by a sequence of incident edges. Thus $L(G)$ is $\ell$-vertex-connected: it has more than $\ell$ vertices and remains connected after deleting any set of fewer than $\ell$ vertices.
		
		The theorem of Lov\'asz, Saks, and Schrijver~\cite{lss1989,lss2000} states that an $(m-d)$-vertex-connected graph on $m$ vertices admits a general-position orthogonal representation in $\mathbb R^d$. Apply it to $L(G)$ and fix such a representation $a_1,\ldots,a_m$, after ordering its vertices by the edges of $G$.
		
		The remaining task is to replace the real representation by rational vectors without losing exact orthogonality. Simply rounding each vector would destroy that condition. Instead, we round nearby input vectors and project each one orthogonally away from the span of the previously constructed vectors corresponding to nonneighbors of its vertex; linear independence survives a sufficiently small perturbation. For each $1\le i\le m$, let $A_i$ be the set of indices $j<i$ whose vertices are not adjacent to vertex $i$ in $L(G)$. Every vertex of an $\ell$-vertex-connected graph has degree at least $\ell$: otherwise deleting its fewer than $\ell$ neighbors would isolate it while leaving another vertex. Hence $|A_i|\le m-1-\ell=d-1$. Given vectors $f_1,\ldots,f_m$ sufficiently close to $a_1,\ldots,a_m$, define vectors $z_1,\ldots,z_m$ in order. Let $B_i$ be the matrix whose columns are $z_j$ for $j\in A_i$, ordered increasingly, and set
		\[
		z_i=f_i-B_i(B_i^{\mathsf T}B_i)^{-1}B_i^{\mathsf T}f_i.
		\]
		If $A_i$ is empty, the subtracted term is defined to be zero. At $f_i=a_i$ for every $i$, induction gives $z_i=a_i$: the columns of $B_i$ are independent by general position, and $a_i$ is orthogonal to all of them. Thus every nonempty Gram matrix $B_i^{\mathsf T}B_i$ is invertible at this choice.
		
		There are only finitely many such matrices. Induction on $i$ and continuity of matrix inversion therefore give an open neighborhood of $(a_1,\ldots,a_m)$ in which every displayed inverse exists and every $z_i$ depends continuously on $(f_1,\ldots,f_m)$. Shrinking this neighborhood if necessary keeps every $d\times d$ determinant formed from distinct vectors $z_i$ nonzero, since there are finitely many and all are nonzero at $(a_1,\ldots,a_m)$. The resulting vectors are in general position. Moreover, $z_i$ is orthogonal to each column of $B_i$, so every required orthogonality relation holds exactly.
		
		Choose all coordinates of $f_1,\ldots,f_m$ rational within this open neighborhood. The displayed formula then produces rational coordinates for every $z_i$, because it uses only rational matrix operations and inverses of invertible rational matrices. Multiply all vectors by a common positive integer that clears their denominators. This preserves orthogonality and linear independence. Relabeling the resulting vectors by the edges of $G$ therefore gives a general-position orthogonal representation $(c_e)_{e\in E(G)}$ of $L(G)$ with every $c_e\in\mathbb Z^d$, as claimed.
	\end{proof}
	
	We also use an elementary exact restriction to absorb the coefficient-extraction overhead. It extends the one-copy extraction for paths discussed in~\cite[Section~3]{vc2017} to any connected loopless multigraph by selecting a spanning tree.
	\begin{lemma}[a diagonal restriction from a spanning tree]
		\label{lem:tree-diagonal}
		For every finite connected loopless multigraph $G$ on $k\ge2$ vertices and every integer $q\ge2$, we have $\Id_k(q)\preceq\GT_q(G)$.
	\end{lemma}
	\begin{proof}
		Label the vertices $1,\ldots,k$ and choose a spanning tree $F$ of $G$. Every vertex is incident to at least one edge of $F$. Use the formal multilinear notation of Definition~\ref{def:multilinear-notation}. On each vertex leg $v$, replace a coordinate symbol by $y_a^{(v)}$ if all incident edges outside $F$ have label $0$ and all incident edges in $F$ have the same label $a\in[q]_0$; replace it by zero otherwise. Here the $y_a^{(v)}$ are new output coordinate symbols.
		
		A surviving global edge labeling has label $0$ on every edge outside $F$. If $a_v$ is the common tree-edge label at vertex $v$, then a tree edge $(u,v)$ has both labels $a_u$ and $a_v$, so $a_u=a_v$. Connectivity of $F$ forces $a_v=a$ for all vertices and a single $a\in[q]_0$. Conversely, each such $a$ gives exactly one surviving labeling. The resulting formal multilinear expression is therefore
		\[
		\sum_{a=0}^{q-1}\prod_{v=1}^{k}y_a^{(v)},
		\]
		which represents $\Id_k(q)$. Thus the coordinate substitutions prove $\Id_k(q)\preceq\GT_q(G)$.
	\end{proof}
	
	The orthogonal representation will select a large family of edge labelings that can be separated on every tensor leg. The spanning-tree restriction will absorb the cost of turning the resulting degeneration into ordinary restrictions. Together, these constructions give the minimum cut equality for binary multigraph tensors.
	
	\begin{lemma}[binary multigraph minimum cut formula; cf.~\cite{vc2017}]
		\label{lem:multigraph-mincut}
		For every finite loopless multigraph $G$ with at least two vertices,
		\[
		\log_2\aQ_{\F}(\GT_2(G))=\MinCut(G).
		\]
	\end{lemma}
	\begin{proof}
		Let $V(G)=\{1,\ldots,k\}$ and put $\ell=\MinCut(G)$. For each positive integer $N$ and each nonempty proper subset $S\subset V(G)$, Lemmas~\ref{lem:flattening-properties} and~\ref{lem:cut-flattening} give
		\[
		\Q_{\F}(\GT_2(G)^{\Kprod N})
		\le \operatorname{FR}_S(\GT_2(G)^{\Kprod N})
		=2^{N|\delta_G(S)|}.
		\]
		Taking $N$-th roots, then the limit as $N\to\infty$, and minimizing over $S$ proves $\aQ_{\F}(\GT_2(G))\le2^\ell$.
		
		If $G$ is disconnected, then $\ell=0$. The tensor is nonzero, so its subrank is at least $1$. Together with $\aQ_{\F}(\GT_2(G))\le2^\ell=1$, this gives $\log_2\aQ_{\F}(\GT_2(G))=0=\MinCut(G)$. If $G$ is connected and $k=2$, it has $m=|E(G)|$ parallel edges and $\GT_2(G)\cong\Id_2(2^m)$, whose asymptotic subrank is $2^m$; its minimum cut size is $m$. We may therefore assume that $G$ is connected and $k\ge3$.
		
		Put $m=|E(G)|$ and $d=m-\ell$. Fix integral vectors $(c_e)_{e\in E(G)}$ as in Lemma~\ref{lem:integral-orthogonal}, and write $c_e=(c_{e,1},\ldots,c_{e,d})$ for every $e\in E(G)$. For every vertex $v$, let $E_v$ be the set of edges incident to $v$. Because $\deg(v)\ge\ell$, the family $(c_e)_{e\in E(G)\setminus E_v}$ has at most $m-\ell=d$ members and is linearly independent.
		
		We first select a large set of edge labelings. Fix an integer $q\ge2$ and define
		\[
		B=\max\left(1,\max_{1\le j\le d}\left(\sum_{e\in E(G)}|c_{e,j}|\right)\right),
		\quad C=(2B+1)^d.
		\]
		The constants $B,C$ depend only on the fixed vectors, not on $q$. For each vertex $v$, list its incident edge labels in the fixed order as $e_{v,1},\ldots,e_{v,\deg(v)}$. For an edge labeling $\lambda:E(G)\to[q]_0$, define its local coordinate by
		\[
		I_v(\lambda)=\sum_{t=1}^{\deg(v)}\lambda(e_{v,t})q^{\deg(v)-t}.
		\]
		This is the coordinate convention of Definition~\ref{def:graph-tensor}, with parallel edges kept as distinct labels. Each coordinate of $\sum_{e\in E(G)}c_e\lambda(e)$ is an integer in $[-B(q-1),B(q-1)]$. There are therefore at most $(2Bq+1)^d\le Cq^d$ possible values of this sum. There are $q^m$ labelings. By the pigeonhole principle, some $g\in\mathbb Z^d$ has a fiber
		\[
		\mathcal F_g=
		\left\{\lambda:E(G)\to[q]_0\;\middle|\;
		\sum_{e\in E(G)}c_e\lambda(e)=g\right\}.
		\]
		Writing $R=|\mathcal F_g|$, this choice satisfies $R\ge C^{-1}q^{m-d}=C^{-1}q^\ell$.
		
		We next construct a degeneration that keeps exactly this fiber. The desired total weight is the squared distance $\|\sum_{e\in E(G)}c_e\lambda(e)-g\|^2$. Its only potentially nonlocal terms couple two edges; orthogonality makes those terms vanish unless the edges share a vertex. We can therefore assign every remaining term to one vertex that sees all of its labels. Use the fixed ordering of edge labels. For an edge $e$ with endpoints $u<v$, put $o(e)=u$. For two distinct incident edge labels $e,f$, let $o(e,f)$ be their smallest common endpoint. Thus even when $e,f$ are parallel, the pair is assigned to exactly one vertex. For each vertex $v$, define an integer weight $w_v(j)$ on its local coordinate indices $j\in[q^{\deg(v)}]_0$ by
		\[
		\begin{aligned}
			w_v(I_v(\lambda))={}&
			\mathbf 1_{\{v=1\}}\langle g,g\rangle
			+\sum_{\substack{e\in E(G)\\o(e)=v}}
			\bigl(\langle c_e,c_e\rangle\lambda(e)^2-2\langle g,c_e\rangle\lambda(e)\bigr)\\
			&+2\sum_{\substack{e,f\in E_v\\e<f,\ o(e,f)=v}}
			\langle c_e,c_f\rangle\lambda(e)\lambda(f).
		\end{aligned}
		\]
		Here $\mathbf 1_{\{v=1\}}$ is $1$ when $v=1$ and $0$ otherwise. The coordinate $I_v(\lambda)$ uniquely determines the incident labels, and each summand uses only those labels. Thus this defines a weight on every local coordinate. Every unordered pair of incident edge labels contributes its cross term at exactly one common endpoint. Also, $\langle c_e,c_f\rangle=0$ when $e$ and $f$ are disjoint, by orthogonality in the line graph. Consequently, summing these local weights gives
		\[
		\sum_{v=1}^{k}w_v(I_v(\lambda))
		=\left\langle\sum_{e\in E(G)}c_e\lambda(e)-g,
		\sum_{e\in E(G)}c_e\lambda(e)-g\right\rangle.
		\]
		This is a nonnegative integer and is zero exactly for $\lambda\in\mathcal F_g$. The equality is an identity of integer exponents; these integers are not reduced modulo the characteristic of $\F$. In the formal multilinear notation of Definition~\ref{def:multilinear-notation}, make the coordinate substitutions
		\[
		x_j^{(v)}\mapsto\varepsilon^{w_v(j)}x_j^{(v)}
		\qquad(1\le v\le k,\ 0\le j<q^{\deg(v)}).
		\]
		The total-weight calculation shows that the resulting expression is polynomial in $\varepsilon$ and that its constant coefficient represents the tensor $T_g$ with formal multilinear expression
		\[
		\sum_{\lambda\in\mathcal F_g}\prod_{v=1}^{k}x_{I_v(\lambda)}^{(v)}.
		\]
		Remark~\ref{rem:weights} therefore gives $T_g\unlhd\GT_q(G)$.
		
		We next show $\Id_k(R)\preceq T_g$ by deleting unused coordinates and relabeling the remaining coordinates on each leg. Suppose that $\lambda,\mu\in\mathcal F_g$ agree on all edges incident to a vertex $v$. Subtracting their two fiber equations yields $\sum_{e\in E(G)\setminus E_v}c_e(\lambda(e)-\mu(e))=0$. The vectors in this sum are linearly independent, as observed above. Thus $\lambda(e)=\mu(e)$ for every $e\in E(G)\setminus E_v$, and the labelings agree on all edges. Hence distinct labelings in $\mathcal F_g$ give distinct basis vectors at every vertex. Enumerate the $R$ labelings as $\lambda_0,\ldots,\lambda_{R-1}$, and make the coordinate substitutions
		\[
		x_{I_v(\lambda_h)}^{(v)}\longmapsto y_h^{(v)}
		\qquad(1\le v\le k,\ 0\le h<R),
		\]
		where the $y_h^{(v)}$ are new output coordinate symbols and all unused input symbols are set to zero. The injectivity just proved makes these substitutions well defined. The resulting expression is
		\[
		\sum_{h=0}^{R-1}\prod_{v=1}^{k}y_h^{(v)},
		\]
		which represents $\Id_k(R)$. Thus $\Id_k(R)\preceq T_g$.
		Composing this restriction with $T_g\unlhd\GT_q(G)$ gives $\Id_k(R)\unlhd\GT_q(G)$ with $R\ge C^{-1}q^\ell$.
		
		To pass from this degeneration to ordinary asymptotic subrank, we absorb its polynomial coefficient-extraction cost into a sublinear number of extra graph tensor copies. Fix $q$ and choose polynomial local maps witnessing $\Id_k(R)\unlhd\GT_q(G)$, clearing Laurent powers if necessary. Let $a$ be their leading degree in Definition~\ref{def:degeneration}. Applying Lemma~\ref{lem:degeneration-asymptotic} with $S=\Id_k(R)$ and $T=\GT_q(G)$ gives the following restriction with $h_N=\binom{aN+k-1}{k-1}$:
		\[
		\Id_k(R^N)\preceq
		\GT_q(G)^{\Kprod N}\Kprod\Id_k(h_N)
		\quad(N\ge1).
		\]
		Let $t_N=\lceil\log_q(h_N)\rceil$. Lemma~\ref{lem:tree-diagonal}, applied to $t_N$ copies and followed by coordinate projection, gives
		\[
		\Id_k(h_N)\preceq\Id_k(q^{t_N})
		\preceq\GT_q(G)^{\Kprod t_N}.
		\]
		The same statement holds for $t_N=0$ by the conventions for zeroth powers. Combining the two restrictions gives $\Q_{\F}(\GT_q(G)^{\Kprod(N+t_N)})\ge R^N$.
		
		Now choose $q=2^s$ for a positive integer $s$. By Lemma~\ref{lem:graph-multiplicative}, $\GT_q(G)$ is isomorphic to $\GT_2(G)^{\Kprod s}$. Since $t_N=O(\log N)$ for fixed $q$, taking $s(N+t_N)$-th roots in the preceding inequality and letting $N\to\infty$ gives $\aQ_{\F}(\GT_2(G))\ge R^{1/s}\ge2^\ell C^{-1/s}$. Finally, $C$ does not depend on $s$. Letting $s\to\infty$ gives $\aQ_{\F}(\GT_2(G))\ge2^\ell$. Combined with $\aQ_{\F}(\GT_2(G))\le2^\ell$ proved at the start, this yields
		\[
		\log_2\aQ_{\F}(\GT_2(G))=\ell=\MinCut(G),
		\]
		which completes the proof for every loopless multigraph $G$.
	\end{proof}

	The binary multigraph formula also handles nonuniform dimensions asymptotically. In a large tensor power, each edge dimension can be approximated from below by a power of two; parallel binary edges then realize the rounded logarithmic weights. The following theorem makes this passage precise.
	
	\begin{theorem}[weighted minimum cut equality]
		\label{thm:mincut-subrank}
		For every field $\F$ and every dimension matrix $Q$ of order $n\ge2$ with logarithmic edge-weight matrix $W$,
		\[
		\log_2\aQ_{\F}(\WT(Q))=\MinCut(W).
		\]
	\end{theorem}
	\begin{proof}
		Write $T=\WT(Q)$. The upper bound comes from a minimum cut flattening: for every nonempty proper vertex set $S$ and every positive integer $N$, Lemmas~\ref{lem:flattening-properties} and~\ref{lem:cut-flattening} give
		\[
		\Q_{\F}(T^{\Kprod N})
		\le\operatorname{FR}_S(T^{\Kprod N})
		=2^{N\operatorname{cut}_W(S)}.
		\]
		Taking $N$-th roots and the limit, then minimizing the base-two logarithm over $S$, proves
		$\log_2\aQ_{\F}(T)\le\MinCut(W)$.
		
		For the reverse bound, we approximate the logarithmic edge weights in a tensor power by integer multiplicities, so that Lemma~\ref{lem:multigraph-mincut} applies. Fix a positive integer $s$ and form the loopless multigraph $H_s$ on the same $n$ vertices with
		$M^{(s)}_{i,j}=\lfloor sW_{i,j}\rfloor$
		parallel edges between vertices $i$ and $j$, for $1\le i<j\le n$. Pairing and projection as in Lemma~\ref{lem:graph-multiplicative} give
		\[
		\GT_2(H_s)\cong\WT(Q^{(s)})\preceq T^{\Kprod s},
		\qquad Q^{(s)}_{i,j}=2^{M^{(s)}_{i,j}}\le Q_{i,j}^{s}
		\quad(1\le i<j\le n),
		\]
		where each dimension matrix is completed symmetrically with diagonal one. Monotonicity and exact multiplicativity on powers of asymptotic subrank now imply
		\[
		s\log_2\aQ_{\F}(T)
		\ge\log_2\aQ_{\F}(\GT_2(H_s))
		=\MinCut(H_s),
		\]
		using Lemma~\ref{lem:multigraph-mincut}. For every nonempty proper vertex set $S$, rounding down changes the total cut weight by less than $\binom n2$, so
		\[
		s\operatorname{cut}_W(S)-\binom n2
		\le |\delta_{H_s}(S)|
		\le s\operatorname{cut}_W(S).
		\]
		Taking the minimum over the same finite set of cuts gives
		$s\MinCut(W)-\binom n2\le\MinCut(H_s)\le s\MinCut(W)$.
		Combining the lower estimate for $\MinCut(H_s)$ with the preceding subrank inequality and dividing by $s$ gives
		\[
		\log_2\aQ_{\F}(T)\ge\MinCut(W)-\frac{1}{s}\binom n2.
		\]
		Letting $s\to\infty$ proves $\log_2\aQ_{\F}(T)\ge\MinCut(W)$. Together with the minimum cut flattening bound, and recalling $T=\WT(Q)$, we obtain
		\[
		\log_2\aQ_{\F}(\WT(Q))=\MinCut(W).
		\]
		The argument also covers disconnected supports and zero edge weights, so the equality holds for every dimension matrix in the statement.
	\end{proof}
	
\end{document}